\documentclass[reqno,10pt,centertags]{amsart}
\usepackage{amssymb}
\usepackage{upref}
\usepackage{hyperref}

\newtheorem{theorem}{Theorem}[section]
\newtheorem{lemma}[theorem]{Lemma}

\newtheorem{hypothesis}[theorem]{Hypothesis}
\newtheorem{example}[theorem]{Example}
\theoremstyle{definition}

\newtheorem{remark}[theorem]{Remark}

\allowdisplaybreaks

\newcommand{\bbN}{{\mathbb{N}}}
\newcommand{\bbR}{{\mathbb{R}}}

\newcommand{\bbZ}{{\mathbb{Z}}}
\newcommand{\bbC}{{\mathbb{C}}}

\newcommand{\calA}{{\mathcal A}}
\newcommand{\calB}{{\mathcal B}}
\newcommand{\calC}{{\mathcal C}}
\newcommand{\calD}{{\mathcal D}}
\newcommand{\calE}{{\mathcal E}}
\newcommand{\calF}{{\mathcal F}}
\newcommand{\calG}{{\mathcal G}}
\newcommand{\calH}{{\mathcal H}}
\newcommand{\calI}{{\mathcal I}}

\newcommand{\calK}{{\mathcal K}}
\newcommand{\calL}{{\mathcal L}}
\newcommand{\calM}{{\mathcal M}}

\newcommand{\calS}{{\mathcal S}}

\newcommand{\N}{p}
\newcommand{\f}{\frac}
\newcommand{\lb}{\label}
\newcommand{\no}{\nonumber}
\newcommand{\be}{\begin{equation}}
\newcommand{\ee}{\end{equation}}
\newcommand{\bea}{\begin{eqnarray}}
\newcommand{\eea}{\end{eqnarray}}
\newcommand{\ul}{\underline}
\newcommand{\ol}{\overline}
\newcommand{\ti}{\widetilde}

\renewcommand{\Im}{\text{\rm Im}}
\newcommand{\deven}{\delta_{\mathrm{even}}}
\newcommand{\dodd}{\delta_{\mathrm{odd}}}

\DeclareMathOperator{\sAL}{s-AL}
\DeclareMathOperator{\shAL}{s-\hatt {AL}}
\DeclareMathOperator{\AL}{AL}
\DeclareMathOperator{\hAL}{\hatt{AL}}

\newcommand{\humu}{{\underline{\hat{\mu} }}}
\newcommand{\hunu}{{\underline{\hat{\nu}}}}
\newcommand{\hmu}{{\hat{\mu} }}
\newcommand{\hnu}{{\hat{\nu}}}

\newcommand{\ual}{{\underline{\alpha}}}

\newcommand{\Div}{\operatorname{Div}}

\newcommand{\amap}{\ul{\alpha}_{Q_0}}

\newcommand{\Oh}{O}

\newcommand{\dott}{\,\cdot\,}
\newcommand{\hatt}{\widehat}  
\newcommand{\Pinfp}{{P_{\infty_+}}}
\newcommand{\Pinfm}{{P_{\infty_-}}}
\newcommand{\Pzp}{{P_{0,+}}}
\newcommand{\Pzm}{{P_{0,-}}}
\newcommand{\Pinfpm}{{P_{\infty_\pm}}}
\newcommand{\Pzpm}{{P_{0,\pm}}}

\DeclareMathOperator{\sym}{Sym}

\newcommand{\symq}{{\sym^{p} (\calK_p)}}

\newcommand{\pgam}{\Gamma}

\allowdisplaybreaks
\numberwithin{equation}{section}

\begin{document}

\title[The Ablowitz--Ladik Hierarchy Initial Value Problem]{
The algebro-geometric initial value problem for the Ablowitz--Ladik hierarchy}

\author[F.\ Gesztesy]{Fritz Gesztesy}
\address{Department of Mathematics,
University of Missouri,
Columbia, MO 65211, USA}
\email{\href{mailto:gesztesyf@missouri.edu}{gesztesyf@missouri.edu}}
\urladdr{\href{http://www.math.missouri.edu/personnel/faculty/gesztesyf.html}
{http://www.math.missouri.edu/personnel/faculty/gesztesyf.html}}

\author[H.\ Holden]{Helge Holden}
\address{Department of Mathematical Sciences,
Norwegian University of
Science and Technology, NO--7491 Trondheim, Norway}
\email{\href{mailto:holden@math.ntnu.no}{holden@math.ntnu.no}}
\urladdr{\href{http://www.math.ntnu.no/\~{}holden/}{http://www.math.ntnu.no/\~{}holden/}}

\author[J. Michor]{Johanna Michor}
\address{Faculty of Mathematics\\
University of Vienna\\
Nordbergstrasse 15\\ 1090 Wien\\ Austria\\ and International Erwin Schr\"odinger
Institute for Mathematical Physics\\ Boltzmanngasse 9\\ 1090 Wien\\ Austria}
\email{\href{mailto:Johanna.Michor@univie.ac.at}{Johanna.Michor@univie.ac.at}}
\urladdr{\href{http://www.mat.univie.ac.at/~jmichor/}{http://www.mat.univie.ac.at/\~{}jmichor/}}

\author[G. Teschl]{Gerald Teschl}
\address{Faculty of Mathematics\\
University of Vienna\\
Nordbergstrasse 15\\ 1090 Wien\\ Austria\\ and International Erwin Schr\"odinger
Institute for Mathematical Physics\\ Boltzmanngasse 9\\ 1090 Wien\\ Austria}
\email{\href{mailto:Gerald.Teschl@univie.ac.at}{Gerald.Teschl@univie.ac.at}}
\urladdr{\href{http://www.mat.univie.ac.at/~gerald/}{http://www.mat.univie.ac.at/\~{}gerald/}}
 
\date{\today}
\subjclass[2000]{Primary 37K10, 37K20, 47B36; Secondary 35Q58, 37K60.}
\keywords{Ablowitz--Ladik hierarchy, complex-valued solutions, initial value
problem.}
\thanks{Research supported in part by the Research Council of Norway,
the US National Science Foundation under Grant No.\ DMS-0405526, and
the Austrian Science Fund (FWF) under Grants No.\ Y330, J2655}
\thanks{Discrete Contin. Dyn. Syst. {\bf 26:1}, 151--196 (2010)}

\dedicatory{Dedicated with great pleasure to Percy Deift on the
occasion of his 60th birthday}

\begin{abstract}
We discuss the algebro-geometric initial value problem for the Ablowitz--Ladik hierarchy
with complex-valued initial data and prove unique solvability globally in
time for a set of initial (Dirichlet divisor) data of full measure. To this effect we develop a new algorithm for constructing stationary complex-valued algebro-geometric solutions of the Ablowitz--Ladik hierarchy, which is of independent interest as it solves the inverse 
algebro-geometric spectral problem for general (non-unitary) Ablowitz--Ladik Lax operators, starting from a suitably chosen set of initial divisors of full measure. Combined with an appropriate first-order system of differential equations with respect to time (a substitute for the well-known Dubrovin-type equations), this yields the construction of global algebro-geometric solutions of the time-dependent Ablowitz--Ladik hierarchy. 

The treatment of general (non-unitary) Lax operators associated with general coefficients for the Ablowitz--Ladik hierarchy poses a variety of difficulties that, to the best of our knowledge, are successfully overcome here for the first time. Our approach is not confined to the Ablowitz--Ladik hierarchy but applies generally to $(1+1)$-dimensional completely integrable soliton equations of differential-difference type. 
\end{abstract}

\maketitle

\section{Introduction} \lb{s1} 

The principal aim of this paper is an explicit construction of unique global solutions of the algebro-geometric initial value problem for the Ablowitz--Ladik hierarchy for a general class of initial data. However, to put this circle of ideas into a proper perspective, we first very briefly recall the origins of this subject: In the mid-seventies, Ablowitz and Ladik, in a series of papers
\cite{AblowitzLadik:1975}--\cite{AblowitzLadik2:1976} (see also
\cite{Ablowitz:1977}, \cite[Sect.\ 3.2.2]{AblowitzClarkson:1991},
\cite[Ch.\ 3]{AblowitzPrinariTrubatch:2004}), used inverse scattering methods to analyze certain integrable differential-difference systems. One of their integrable
variants of such systems included a discretization of the celebrated AKNS-ZS system, the pair of coupled nonlinear differential-difference equations,
\begin{align}
\begin{split}
-i\alpha_t - (1 - \alpha \beta) (\alpha^- +\alpha^+) + 2\alpha &=0, \lb{AL1.12} \\
-i\beta_t + (1 - \alpha\beta)(\beta^- + \beta^+) - 2\beta &=0   
\end{split}
\end{align}with $\alpha=\alpha(n,t)$, $\beta=\beta(n,t)$, $(n,t)\in\bbZ\times\bbR$. 
Here we used the notation $f^{\pm}(n) = f(n\pm 1)$, $n\in\bbZ$, for complex-valued 
sequences $f=\{f(n)\}_{n\in\bbZ}$. In particular, Ablowitz and Ladik 
\cite{AblowitzLadik:1976} (see also \cite[Ch.\ 3]{AblowitzPrinariTrubatch:2004}) showed that in the focusing case, where $\beta = -\ol\alpha$, and in the defocusing case, where $\beta = \ol\alpha$, \eqref{AL1.12} yields the discrete analog of the nonlinear Schr\"odinger equation
\begin{equation}
-i\alpha_t - (1 \pm |\alpha|^2)(\alpha^- + \alpha^+) + 2 \alpha = 0. \lb{AL1.13}
\end{equation}
Since then there has been an enormous activity in this area and we refer, for instance, to \cite[Ch.\ 3]{AblowitzPrinariTrubatch:2004},  
\cite{GeronimoGesztesyHolden:2004}, \cite{GesztesyHoldenMichorTeschl:2007a}, 
\cite{GesztesyHoldenMichorTeschl:2007}, \cite{GesztesyHoldenMichorTeschl:2007b},  
\cite{KillipNenciu:2006}, \cite{Li:2005}, \cite{MillerErcolaniKricheverLevermore:1995}, 
\cite{Nenciu:2005a}, \cite{Nenciu:2006}, \cite{Schilling:1989}, \cite{Suris:2003} and the extensive literature cited therein, for developments leading up to current research in this particular area of completely integrable differential-difference systems. Particularly relevant to this paper are algebro-geometric (and periodic) solutions of the AL system \eqref{AL1.12} and its associated hierarchy of integrable equations. The first systematic and detailed treatment of algebro-geometric solutions of the AL system \eqref{AL1.12} was performed by Miller, Ercolani, Krichever, and Levermore 
\cite{MillerErcolaniKricheverLevermore:1995} (see also \cite{AhmadChowdhury:1987}, 
\cite{AhmadChowdhury:1987a}, \cite{BogolyubovPrikarpatskiiSamoilenko:1981},  
\cite{BogolyubovPrikarpatskii:1982}, \cite{Miller:1995}, \cite{Vaninsky:2001}). 
Algebro-geometric solutions of the AL hierarchy were discussed in great detail in 
\cite{GesztesyHoldenMichorTeschl:2007} (see also \cite{GengDaiCao:2003}, 
\cite{GengDaiZhu:2007}, \cite{Vekslerchik:1999}). The initial value problem for the half-infinite discrete linear Schr\"odinger equation and the Schur flow were discussed by Common 
\cite{Common:1992} (see also \cite{CommonHafez:1990}) using a continued fraction approach. The corresponding nonabelian cases on a finite interval were studied by 
Gekhtman \cite{Gekhtman:1993}. In addition to these developments within integrable systems and their applications to fields such as nonlinear optics, the study of AL systems recently gained considerable momentum due to its connections with the theory of orthogonal polynomials. Especially, the particular defocusing case $\beta = \ol \alpha$ and the associated CMV matrices and orthogonal polynomials on the unit circle attracted great interest. In this context we refer the interested reader to the two-volume treatise by Simon \cite{Simon:2005} (see also \cite{Simon:2006}) and the survey by Deift \cite{Deift:2007} and the detailed reference lists therein. 

Returning to the principal scope of this paper, we intend to describe a solution of the following problem: Given $\ul p=(p_-,p_+)\in\bbN_0^2\setminus\{(0,0)\}$, $\ul r\in\bbN_0^2$, assume $\alpha^{(0)}, \beta^{(0)}$ to be solutions of the $\ul p$th stationary Ablowitz--Ladik system $\sAL_{\ul p}(\alpha,\beta)= 0$ associated with a prescribed hyperelliptic curve $\calK_p$ of genus $p=p_- + p_+ -1$
(with nonsingular affine part).  We want to construct a unique global solution
$\alpha=\alpha(t_{\ul r}), \beta=\beta(t_{\ul r})$ of the $\ul r$th AL flow $\AL_{\ul r}(\alpha,\beta)=0$ with
$\alpha(t_{0,\ul r})=\alpha^{(0)}, \beta(t_{0,\ul r})=\beta^{(0)}$ for some $t_{0,\ul r}\in\bbR$. Thus, we seek the unique global solution of the initial value problem
\begin{align}
\begin{split}
&\AL_{\ul r}(\alpha,\beta)=0, \label{1.1} \\
&(\alpha,\beta)\big|_{t_{\ul r}=t_{0,\ul r}}=\big(\alpha^{(0)},\beta^{(0)}\big),
\end{split} \\
&\sAL_{\ul p}\big(\alpha^{(0)},\beta^{(0)}\big)=0, \label{1.2}
\end{align}
where $\alpha=\alpha(n,t_{\ul r}), \beta=\beta(n,t_{\ul r})$ satisfy the conditions in 
\eqref{AL2.01a}. 

Given the particularly familiar case of real-valued algebro-geometric solutions of the Toda hierarchy (see, e.g., \cite{BullaGesztesyHoldenTeschl:1997}, 
\cite[Sect.\ 1.3]{GesztesyHolden:2005}, \cite[Sect.\ 8.3]{Teschl:2000} and the extensive literature cited therein), the actual solution of this algebro-geometric initial value problem, naively, might consist of the following two-step procedure:\footnote{We freely use the notation of divisors of degree $p$ as introduced in Appendix \ref{AL.sA} (see also the beginning of Section \ref{s3}).}

$(i)$ An algorithm that constructs admissible (cf.\ Section \ref{ALSs4}) nonspecial divisors $\calD_{\humu(n)}\in \sym^{p}\calK_p$ for all $n\in\bbZ$, starting from a nonspecial initial  Dirichlet divisor $\calD_{\humu(n_0)}\in \sym^{p}\calK_p$.  ``Trace formulas'' of the type \eqref{ALtr1} and \eqref{ALtr2} (the latter requires prior construction of the Neumann divisor $\calD_{\hunu}$ from the Dirichlet divisor $\calD_{\humu}$, though, cf.\ \eqref{ALS4.2a}) should then construct  the stationary solutions 
$\alpha^{(0)}, \beta^{(0)}$ of $\sAL_{\ul p}(\alpha,\beta)= 0$.

$(ii)$ The first-order Dubrovin-type system of differential equations \eqref{AL4.69a}, augmented by the initial divisor $\calD_{\humu(n_0, t_{0,\ul r})}=\calD_{\humu(n_0)}$ 
(cf.\ step $(i)$) together with the analogous ``trace formulas'' \eqref{ALtr1}, \eqref{ALtr2} in the time-dependent context should then yield unique global solutions 
$\alpha=\alpha(t_{\ul r}), \beta=\beta(t_{\ul r})$ of the 
$\ul r$th AL flow $\AL_{\ul r}(\alpha,\beta)=0$ satisfying 
$\alpha(t_{0,\ul r})=\alpha^{(0)}, \beta(t_{0,\ul r})=\beta^{(0)}$.

However, this approach can be expected to work only if the Dirichlet divisors 
$\calD_{\humu(n,t_{\ul r})}\in \sym^{p}\calK_p$, 
$\humu(n,t_{\ul r})=(\hat\mu_1(n,t_{\ul r}),\dots,\hat\mu_{p}(n,t_{\ul r}))$, 
yield pairwise distinct Dirichlet eigenvalues $\mu_j(n,t_{\ul r})$, $j=1,\dots,p$,  for fixed 
$(n,t_{\ul r})\in\bbZ\times\bbR$, such that formula \eqref{AL4.69a} for the 
time-derivatives $\mu_{j,t_{\ul r}}$, $j=1,\dots,p$, is well-defined. Analogous considerations apply to the Neumann divisors 
$\calD_{\hunu}\in \sym^{p}\calK_p$.

Unfortunately, this scenario of pairwise distinct Dirichlet eigenvalues is not realistic and ``collisions'' between them can occur at certain values of 
$(n,t_{\ul r})\in\bbZ\times\bbR$. Thus, the stationary algorithm in step $(i)$ as well as the Dubrovin-type first-order system of differential equations \eqref{AL4.69a} in step $(ii)$ above breaks down at such values of $(n,t_{\ul r})$. A priori, one has no control over such collisions, especially, it is not possible to identify initial conditions 
$\calD_{\humu(n_0, t_{0,\ul r})}$ at some $(n_0,t_{0,\ul r})\in\bbZ\times\bbR$, which avoid collisions for all $(n,t_{\ul r})\in\bbZ\times\bbR$. We solve this problem head on by explicitly permitting collisions in the stationary as well as the time-dependent context from the outset. In the stationary context, we introduce an appropriate algorithm alluded to in step $(i)$ by employing a general interpolation formalism 
(cf.\ Appendix \ref{AL.sB}) for polynomials, going beyond the usual Lagrange interpolation formulas. In the time-dependent context we replace the first-order system of Dubrovin-type equations \eqref{AL4.69a}, augmented with the initial divisor 
$\calD_{\humu(n_0, t_{0,\ul r})}$, by a different first-order system of differential equations \eqref{ALT6.22}, \eqref{ALT6.37}, and \eqref{ALT6.37a} with initial conditions 
\eqref{ALT6.33a} which focuses on symmetric functions of 
$\mu_1(n,t_{\ul r}),\dots,\mu_{p}(n,t_{\ul r})$ rather than individual Dirichlet eigenvalues 
$\mu_j(n,t_{\ul r})$, $j=1,\dots,p$. In this manner it will be shown that collisions of Dirichlet eigenvalues no longer pose a problem.  

In addition, there is an additional complication: In general, it cannot be guaranteed that 
$\mu_j(n,t_{\ul r})$ and $\nu_j(n,t_{\ul r})$, $j=1,\dots,p$, stay finite and nonzero for all 
$(n,t_{\ul r})\in\bbZ\times\bbR$. We solve this particular problem in the stationary as well as the time-dependent case by properly restricting the initial Dirichlet and Neumann divisors 
$\calD_{\humu(n_0, t_{0,\ul r})}, \calD_{\hunu(n_0, t_{0,\ul r})}\in \sym^{p}\calK_p$ to a dense set of full measure.

Summing up, we offer a new algorithm to solve the inverse algebro-geometric spectral problem for general Ablowitz--Ladik Lax operators, starting from a properly chosen dense set of initial divisors of full measure. Combined with an appropriate first-order system of differential equations with respect to time (a substitute for the Dubrovin-type equations), this yields the construction of global algebro-geometric solutions of the time-dependent Ablowitz--Ladik hierarchy. 

We emphasize that the approach described in this paper is not limited to the Ablowitz--Ladik  hierarchy but applies universally to constructing algebro-geometric solutions of 
$(1+1)$-dimensional integrable soliton equations. In particular, it applies to the Toda lattice hierarchy as discussed in \cite{GesztesyHoldenTeschl:2007}. Moreover, the principal idea of replacing Dubrovin-type equations by a first-order system of the type \eqref{ALT6.22}, \eqref{ALT6.37}, and \eqref{ALT6.37a} is also relevant in the context of general (non-self-adjoint) Lax operators for the continuous models in 
$(1+1)$-dimensions. 
In particular, the models studied in detail in \cite{GesztesyHolden:2003} can be revisited from this point of view. (However, the fact that the set in \eqref{ALT6.78} is of measure zero relies on the fact that $n$ varies in the countable set $\bbZ$ and hence is not applicable to continuous models in $1+1$-dimensions.) We also note that while the periodic case with complex-valued $\alpha, \beta$ is of course included in our analysis, we throughout consider the more general algebro-geometric case (in which $\alpha, \beta$ need not be quasi-periodic).  

Finally we briefly describe the content of each section. Section \ref{s2} presents a quick summary of the basics of the Ablowitz--Ladik hierarchy, its recursive construction, Lax pairs, and zero-curvature equations. The stationary algebro-geometric Ablowitz--Ladik hierarchy solutions, the underlying hyperelliptic curve, trace formulas, etc., are the subject of Section \ref{s3}. A new algorithm solving the algebro-geometric inverse spectral problem for general Ablowitz--Ladik Lax operators is presented in Section \ref{ALSs4}. In Section \ref{s5} we briefly summarize the properties of algebro-geometric time-dependent solutions of the Ablowitz--Ladik hierarchy and formulate the algebro-geometric initial value problem. Uniqueness and existence of global solutions of the algebro-geometric initial value problem as well as their explicit construction are then presented in our final and principal Section \ref{ALTs6}. Appendix \ref{AL.sA} reviews the basics of hyperelliptic Riemann surfaces of the Ablowitz--Ladik-type and sets the stage of much of the notation used in this paper. Finally, various interpolation formulas of fundamental importance to our stationary inverse spectral algorithm developed in Section \ref{ALSs4} are summarized in Appendix \ref{AL.sB}. These appendices support our intention to make this paper reasonably self-contained.

\section{The Ablowitz--Ladik Hierarchy in a Nutshell} 
\label{s2}

We briefly review the recursive construction of the Ablowitz--Ladik hierarchy
and zero-curvature equations following   
\cite{GesztesyHoldenMichorTeschl:2007a}  and 
\cite{GesztesyHoldenMichorTeschl:2007b}.

Throughout this section we suppose the following hypothesis. 

\begin{hypothesis} \lb{hAL2.1} 
In the stationary case we assume that $\alpha, \beta$ satisfy
\begin{equation}
\alpha, \beta\in \bbC^{\bbZ}, \quad  \alpha(n)\beta(n)\notin \{0,1\}, 
\;  n\in\bbZ.   \lb{AL2.01}
\end{equation}
In the time-dependent case we assume that $\alpha, \beta$ satisfy
\begin{align}
\begin{split}
& \alpha(\dott,t), \beta(\dott,t) \in \bbC^{\bbZ}, \; t\in\bbR, \quad 
\alpha(n,\dott), \beta(n,\dott)\in C^1(\bbR), \; n\in\bbZ,   \lb{AL2.01a}  \\
& \alpha(n,t)\beta(n,t)\notin \{0,1\}, \; (n,t)\in\bbZ\times \bbR.
\end{split}
\end{align}
\end{hypothesis}

Here $\bbC^{\bbZ}$ denotes the set of complex-valued sequences indexed by $\bbZ$. For a discussion of assumptions \eqref{AL2.01} and \eqref{AL2.01a} we refer to 
Remark \ref{rAL3.4}. 

We denote by $S^\pm$ the shift operators acting on complex-valued sequences 
$f=\{f(n)\}_{n\in\bbZ} \in\bbC^{\bbZ}$ according to
\begin{equation}
(S^\pm f)(n)=f(n\pm1), \quad n\in\bbZ. \lb{AL2.02}
\end{equation}
Moreover, we will frequently use the notation
\begin{equation}
f^\pm = S^{\pm} f, \quad f\in\bbC^{\bbZ}. 
\end{equation}

Next, we define sequences $\{f_{\ell,\pm}\}_{\ell\in \bbN_0}$, $\{g_{\ell,\pm}\}_{\ell\in \bbN_0}$, and $\{h_{\ell,\pm}\}_{\ell\in \bbN_0}$ recursively by
\begin{align} \label{AL0+}
g_{0,+} &= \tfrac12 c_{0,+}, \quad f_{0,+} = - c_{0,+}\alpha^+, 
\quad h_{0,+} = c_{0,+}\beta, \\ \label{ALg_l+}
g_{\ell+1,+} - g_{\ell+1,+}^- &= \alpha h_{\ell,+}^- + \beta f_{\ell,+}, \quad \ell\in \bbN_0,\\ \label{ALf_l+}
f_{\ell+1,+}^- &= f_{\ell,+} - \alpha (g_{\ell+1,+} + g_{\ell+1,+}^-), \quad \ell\in \bbN_0, \\  \label{ALh_l+}
h_{\ell+1,+} &= h_{\ell,+}^- + \beta (g_{\ell+1,+} + g_{\ell+1,+}^-), \quad \ell\in \bbN_0,  
\end{align}
and
\begin{align} \label{AL0-}
g_{0,-} &= \tfrac12 c_{0,-}, \quad f_{0,-} = c_{0,-}\alpha, 
\quad h_{0,-} = - c_{0,-}\beta^+, \\ \label{ALg_l-}
g_{\ell+1,-} - g_{\ell+1,-}^- &= \alpha h_{\ell,-} + \beta f_{\ell,-}^-, \quad \ell\in \bbN_0,\\ \label{ALf_l-}
f_{\ell+1,-} &= f_{\ell,-}^- + \alpha (g_{\ell+1,-} + g_{\ell+1,-}^-), \quad \ell\in \bbN_0, \\ \label{ALh_l-}
h_{\ell+1,-}^- &= h_{\ell,-} - \beta (g_{\ell+1,-} + g_{\ell+1,-}^-), \quad \ell\in \bbN_0.
\end{align}
Here $c_{0,\pm}\in\bbC$ are given constants. For later use we also introduce
\begin{equation}\lb{ALminus}
f_{-1,\pm}= h_{-1,\pm}=0.
\end{equation}

\begin{remark}\lb{rAL2.2}
The sequences $\{f_{\ell,+}\}_{\ell\in \bbN_0}$, 
$\{g_{\ell,+}\}_{\ell\in \bbN_0}$, and
$\{h_{\ell,+}\}_{\ell\in \bbN_0}$ can be computed recursively as follows: 
Assume that $f_{\ell,+}$,
$g_{\ell,+}$, and $h_{\ell,+}$ are known.  Equation \eqref{ALg_l+} is a 
first-order difference equation in $g_{\ell+1,+}$ that can be solved directly
and yields a local lattice function that is determined up to a new constant denoted
by $c_{\ell+1,+}\in\bbC$. Relations \eqref{ALf_l+} and \eqref{ALh_l+}
then determine $f_{\ell+1,+}$ and $h_{\ell+1,+}$, etc.  The sequences 
$\{f_{\ell,-}\}_{\ell\in \bbN_0}$, $\{g_{\ell,-}\}_{\ell\in \bbN_0}$, and 
$\{h_{\ell,-}\}_{\ell\in \bbN_0}$ are determined similarly.
\end{remark}

Upon setting 
\begin{equation}
\gamma = 1 - \alpha \beta, \lb{ALgamma}
\end{equation}
one explicitly obtains 
\begin{align}
\begin{split}
f_{0,+} &= c_{0,+}(-\alpha^+), \quad 
f_{1,+} = c_{0,+}\big(- \gamma^+ \alpha^{++} + (\alpha^+)^2 \beta\big) 
+ c_{1,+} (-\alpha^+), \\
g_{0,+} &= \tfrac{1}{2}c_{0,+},  \quad 
 g_{1,+} = c_{0,+}(-\alpha^+ \beta) + \tfrac{1}{2}c_{1,+}, \\
h_{0,+} &= c_{0,+}\beta, \quad 
 h_{1,+} = c_{0,+}\big(\gamma \beta^- - \alpha^+ \beta^2\big) 
+ c_{1,+} \beta, \\
f_{0,-} &= c_{0,-}\alpha, \quad  
f_{1,-} = c_{0,-}\big(\gamma \alpha^- - \alpha^2 \beta^+\big) + c_{1,-} \alpha, \\
g_{0,-} &= \tfrac{1}{2}c_{0,-}, \quad  
g_{1,-} = c_{0,-}(-\alpha \beta^+) + \tfrac{1}{2}c_{1,-}, \\
h_{0,-} &= c_{0,-}(-\beta^+), \quad  
h_{1,-} = c_{0,-}\big(- \gamma^+ \beta^{++} 
+ \alpha (\beta^+)^2 \big) + c_{1,-} (- \beta^+), \, \text{ etc.}
\end{split}
\end{align}
Here $\{c_{\ell,\pm}\}_{\ell \in \bbN}$ denote summation constants
which naturally arise when solving the difference equations for 
$g_{\ell, \pm}$ in \eqref{ALg_l+}, \eqref{ALg_l-}.  
Subsequently, it will also be useful to work with the corresponding homogeneous coefficients $\hat f_{\ell, \pm}$,
$\hat g_{\ell, \pm}$, and $\hat h_{\ell, \pm}$, defined by the vanishing of all summation constants $c_{k,\pm}$ for $k=1,\dots,\ell$, and choosing $c_{0,\pm}=1$,
\begin{align}
& \hat f_{0,+}=-\alpha^+, \quad \hat f_{0,-}=\alpha, \quad 
 \hat f_{\ell,\pm}=f_{\ell,\pm}|_{c_{0,\pm}=1, \, c_{j,\pm}=0, j=1,\dots,\ell},  \quad \ell\in\bbN, 
 \lb{AL2.04a} \\
& \hat g_{0,\pm}=\tfrac12, \quad 
\hat g_{\ell,\pm}=g_{\ell,\pm}|_{c_{0,\pm}=1, \, c_{j,\pm}=0, j=1,\dots,\ell}, 
\quad \ell\in\bbN,  \lb{AL2.04b} \\
& \hat h_{0,+}=\beta, \quad \hat h_{0,-}=-\beta^+,  \quad 
\hat h_{\ell,\pm}=h_{\ell,\pm}|_{c_{0,\pm}=1, \, c_{j,\pm}=0, j=1,\dots,\ell}, 
\quad \ell\in\bbN.  \lb{AL2.04c}
\end{align}
By induction one infers that
\begin{equation} \label{ALhat f}
f_{\ell, \pm} = \sum_{k=0}^\ell c_{\ell-k, \pm} \hat f_{k, \pm}, \quad
g_{\ell, \pm} = \sum_{k=0}^\ell c_{\ell-k, \pm} \hat g_{k, \pm}, \quad 
h_{\ell, \pm} = \sum_{k=0}^\ell c_{\ell-k, \pm} \hat h_{k, \pm}.  
\end{equation} 

Next we define the $2\times 2$ zero-curvature matrices
\begin{equation}
U(z) = \begin{pmatrix} z & \alpha  \\ z \beta & 1\\ \end{pmatrix}   \lb{AL2.03}
\end{equation}
and 
\begin{equation} \lb{AL_v}
V_{\ul p}(z) = i  \begin{pmatrix}
G_{\ul p}^-(z) & - F_{\ul p}^-(z)     \\[1.5mm]
H_{\ul p}^-(z) & - K_{\ul p}^-(z)  \\
\end{pmatrix},  \quad \ul p=(p_-,p_+) \in \bbN_0^2,
\end{equation}
for appropriate Laurent polynomials $F_{\ul p}(z)$, $G_{\ul p}(z)$, $H_{\ul p}(z)$, and $K_{\ul p}(z)$ in the  
spectral parameter $z\in \bbC\setminus\{0\}$ to be determined shortly. By postulating the stationary zero-curvature relation,  
\begin{equation}   \lb{ALstatzc}
0=U V_{\ul p} - V_{\ul p}^+ U,
\end{equation}
one concludes that \eqref{ALstatzc} is equivalent to the following relations 
 \begin{align} \label{AL1,1}
z (G_{\ul p}^- - G_{\ul p}) + z \beta F_{\ul p} + \alpha H_{\ul p}^- &= 0,\\ \label{AL2,2}
z \beta F_{\ul p}^- + \alpha H_{\ul p} - K_{\ul p} + K_{\ul p}^- &= 0,\\
\label{AL1,2}
 - F_{\ul p} + z F_{\ul p}^- + \alpha (G_{\ul p} + K_{\ul p}^-) &= 0,\\ \label{AL2,1}
 z \beta (G_{\ul p}^- + K_{\ul p}) - z H_{\ul p} + H_{\ul p}^- &= 0.
\end{align}
In order to make the connection between the zero-curvature formalism and the recursion relations 
\eqref{AL0+}--\eqref{ALh_l-}, we now define Laurent polynomials $F_{\ul p}$, $G_{\ul p}$, $H_{\ul p}$, and $K_{\ul p}$ by\footnote{In this paper, a sum is interpreted as zero whenever the upper limit in the sum is strictly less than its lower limit.}
\begin{align}
F_{\ul p}(z) &= \sum_{\ell=1}^{p_-} f_{p_- -\ell,-} z^{-\ell} 
+ \sum_{\ell=0}^{p_+ -1} f_{p_+ -1-\ell,+}z^\ell,  
\label{ALF_p} \\ 
G_{\ul p}(z) &= \sum_{\ell=1}^{p_-} g_{p_- -\ell,-}z^{-\ell}  
+ \sum_{\ell=0}^{p_+} g_{p_+ -\ell,+}z^\ell,  
 \label{ALG_p}  \\ 
H_{\ul p}(z) &= \sum_{\ell=0}^{p_- -1} h_{p_- -1-\ell,-}z^{-\ell}  
+ \sum_{\ell=1}^{p_+} h_{p_+ -\ell,+}z^\ell, 
 \label{ALH_p}  \\
K_{\ul p}(z) &= \sum_{\ell=0}^{p_-} g_{p_- -\ell,-}z^{-\ell}  
+  \sum_{\ell=1}^{p_+} g_{p_+ -\ell,+}z^\ell 
= G_{\ul p}(z)+g_{p_-,-}-g_{p_+,+}.   \label{ALK_p}
\end{align}

The stationary zero-curvature relation \eqref{ALstatzc}, $0=U V_{\ul p} - V_{\ul p}^+ U$, is then equivalent to 
\begin{align}
 -\alpha(g_{p_+,+} + g_{p_-,-}^-) + f_{p_+ -1,+} - f_{p_- -1,-}^-&=0,  \lb{AL2.50}\\ 
 \beta(g_{p_+,+}^- + g_{p_-,-}) + h_{p_+ -1,+}^- - h_{p_- -1,-} &=0.  \lb{AL2.51}
\end{align}
Thus, varying $p_\pm \in \bbN_0$,  equations \eqref{AL2.50} and \eqref{AL2.51} give rise to the stationary Ablowitz--Ladik (AL) hierarchy which we introduce as follows
\begin{align}\lb{ALstat}
\begin{split}
& \sAL_{\ul p}(\alpha, \beta) = \begin{pmatrix} 
- \alpha(g_{p_+,+} + g_{p_-,-}^-) + f_{p_+ -1,+} - f_{p_- -1,-}^-\\  
\beta(g_{p_+,+}^- + g_{p_-,-}) + h_{p_+ -1,+}^- - h_{p_- -1,-}  \end{pmatrix}=0, \\  
& \hspace*{6.85cm}  \ul p=(p_-,p_+)\in\bbN_0^2. 
\end{split}
\end{align}
Explicitly (recalling $\gamma=1-\alpha\beta$ and taking $p_-=p_+$ for simplicity), 
\begin{align} \no
\sAL_{(0,0)} (\alpha, \beta) &=  \begin{pmatrix}  -c_{(0,0)} \alpha\\ 
c_{(0,0)}\beta\end{pmatrix} 
=0,\\ \no
\sAL_{(1,1)} (\alpha, \beta) &=  \begin{pmatrix} -\gamma (c_{0,-}\alpha^- + c_{0,+}\alpha^+) - c_{(1,1)} \alpha \\
 \gamma (c_{0,+}\beta^- + c_{0,-}\beta^+) +
c_{(1,1)} \beta\end{pmatrix}=0,\\ \no
\sAL_{(2,2)} (\alpha, \beta) &=  \begin{pmatrix}\begin{matrix}
-\gamma \big(c_{0,+}\alpha^{++} \gamma^+ + c_{0,-}\alpha^{--} \gamma^-
- \alpha (c_{0,+}\alpha^+\beta^- + c_{0,-}\alpha^-\beta^+)\\
- \beta (c_{0,-}(\alpha^-)^2 + c_{0,+}(\alpha^+)^2)\big)\end{matrix}\\[3mm] 
\begin{matrix}
 \gamma \big(c_{0,-}\beta^{++} \gamma^+ + c_{0,+}\beta^{--} \gamma^-
- \beta (c_{0,+}\alpha^+\beta^- + c_{0,-}\alpha^-\beta^+)\\
- \alpha (c_{0,+}(\beta^-)^2 + c_{0,-}(\beta^+)^2)\big)\end{matrix}\end{pmatrix}
 \\ & \quad+ \begin{pmatrix}
-\gamma (c_{1,-} \alpha^- + c_{1,+} \alpha^+) - c_{(2,2)} \alpha\\
 \gamma (c_{1,+} \beta^- + c_{1,-} \beta^+) + c_{(2,2)} \beta\end{pmatrix}
=0,  \, \text{ etc.,}
\end{align}
represent the first few equations of the stationary Ablowitz--Ladik hierarchy. 
Here we introduced  
\begin{equation}
c_{\ul p} = (c_{p_-,-} + c_{p_+,+})/2, \quad p_\pm\in\bbN_0.  \lb{ALdefcp}
\end{equation}
By definition, the set of solutions of \eqref{ALstat}, with $\ul p$ ranging in $\bbN_0^2$ and $c_{\ell,\pm}\in\bbC$, $\ell\in\bbN_0$, represents the class of algebro-geometric Ablowitz--Ladik solutions. 

In the following we will frequently assume that $\alpha, \beta$ satisfy the $\ul p$th
stationary Ablowitz--Ladik system supposing a particular choice of summation
constants $c_{\ell,\pm}\in\bbC$, $\ell=0,\dots,p_\pm$, $p_\pm\in\bbN_0$, has been made.

In accordance with our notation introduced in 
\eqref{AL2.04a}--\eqref{AL2.04c}, the corresponding homogeneous stationary Ablowitz--Ladik equations are defined by  
\begin{equation}
\shAL_{\ul p} (\alpha, \beta) 
= \sAL_{\ul p} (\alpha, \beta)\big|_{c_{0,\pm}=1, \, c_{\ell,\pm}=0, \, \ell=1,\dots,p_\pm}, \quad \ul p=(p_-,p_+)\in\bbN_0^2.   \lb{ALstathom}
\end{equation}

One can show (cf.\ \cite[Lemma\ 2.2]{GesztesyHoldenMichorTeschl:2007a}) that 
$g_{p_+,+} = g_{p_-,-}$ up to a lattice constant which can be set equal to zero without loss of generality. Thus, we will henceforth assume that
\begin{equation}
g_{p_+,+} = g_{p_-,-},   \lb{g+=g-}
\end{equation} 
which in turn implies that
\begin{equation}
K_{\ul p}= G_{\ul p}   \lb{ALK=G}
\end{equation}
and hence renders $V_{\ul p}$ in \eqref{AL_v} traceless in the stationary context. (We note that equations \eqref{g+=g-} and \eqref{ALK=G} cease to be valid in the time-dependent context, though.)

Next, still assuming \eqref{AL2.01} and taking into account 
\eqref{ALK=G}, one infers by taking determinants in the stationary zero-curvature equation \eqref{ALstatzc} that the quantity
\begin{equation} \label{ALR}
R_{\ul p} = G_{\ul p}^2 - F_{\ul p} H_{\ul p}  
\end{equation}
is a lattice constant, that is, $R_{\ul p} - R_{\ul p}^- = 0$. Hence, $R_{\ul p}(z)$  only depends on $z$, and assuming in addition to \eqref{AL2.01}  that 
\begin{equation}
c_{0,\pm} \in \bbC\setminus \{0\},   \lb{ALc0}
\end{equation}
one may write $R_{\ul p}$ as\footnote{We use the convention that a product is to be interpreted equal to $1$ whenever the upper limit of the product is strictly less than its lower limit.}  
\begin{equation}  \label{ALE_m}
R_{\ul p}(z) = (c_{0,+}^2/4) z^{-2p_-} \prod_{m=0}^{2p+1}(z-E_m), \quad
\{E_m\}_{m=0}^{2p+1} \subset \bbC \setminus\{ 0\}, 
\end{equation}
where
\begin{equation}  \label{AL_ny}
p=p_- + p_+ -1.
\end{equation}

 In addition, we note that the summation constants $c_{1,\pm},\dots,c_{p_\pm,\pm}$ in \eqref{ALstat} can be expressed as symmetric functions in the zeros 
 $E_0,\dots,E_{2p+1}$ of the associated Laurent polynomial $R_{\ul p}$ in 
 \eqref{ALE_m}. In fact, one can prove (cf.\ 
 \cite{GesztesyHoldenMichorTeschl:2007a}) that 
\begin{equation}
c_{\ell,\pm}= c_{0,\pm} c_\ell\big(\ul E^{\pm 1}\big), \quad \ell=0,\dots,p_\pm, 
\lb{ALcell}
\end{equation}
where   
\begin{align}
&c_{0}\big(\ul E^{\pm 1}\big)=1,\no \\
&c_{k}\big(\ul E^{\pm 1}\big)   \label{ALc_ell}  \\
&=-\!\!\!\!\!\!\!\sum_{\substack{j_0,\dots,j_{2p+1}=0\\
   j_0+\cdots+j_{2p+1}=k}}^{k}\!\!
\f{(2j_0)!\cdots(2j_{2p+1})!}
{2^{2k} (j_0!)^2\cdots (j_{2p+1}!)^2 (2j_0-1)\cdots(2j_{2p+1}-1)}
E_0^{\pm j_0}\cdots E_{2p+1}^{\pm j_{2p+1}},  \no  \\
& \hspace*{10.95cm} k\in\bbN,   \no
\end{align}
are symmetric functions of $\ul E^{\pm 1}=(E_0^{\pm 1},\dots,E_{2p+1}^{\pm 1})$.  

Next we turn to the time-dependent Ablowitz--Ladik hierarchy. For that purpose the coefficients $\alpha$ and $\beta$ are now considered as functions of both the lattice point and time. For each system in the hierarchy, that is, for each 
$\ul p\in\bbN_0^2$, we introduce a deformation (time) parameter 
$t_{\ul p}\in\bbR$ in $\alpha, \beta$, replacing $\alpha(n), \beta(n)$ by 
$\alpha(n,t_{\ul p}), \beta(n,t_{\ul p})$. Moreover, the definitions 
\eqref{AL2.03},  \eqref{AL_v}, and \eqref{ALF_p}--\eqref{ALK_p}  of $U, V_{\ul p}$, and $F_{\ul p}, G_{\ul p}, H_{\ul p}, K_{\ul p}$, respectively, still apply. Imposing the zero-curvature relation
\begin{equation}
U_{t_{\ul p}} + U V_{\ul p} - V_{\ul p}^+ U =0, \quad \ul p\in\bbN_0^2,   \lb{ALzc p}
\end{equation}
then results in the equations
\begin{align}  \label{ALalphat}
\alpha_{t_{\ul p}} &= i \big(z F_{\ul p}^- + \alpha (G_{\ul p} + K_{\ul p}^-) - F_{\ul p}\big),    \\ \label{ALbetat}
\beta_{t_{\ul p}} &= - i \big(\beta (G_{\ul p}^- + K_{\ul p}) - H_{\ul p} 
+ z^{-1} H_{\ul p}^-\big), \\ \label{AL1,1r}
0 &= z (G_{\ul p}^- - G_{\ul p}) + z\beta F_{\ul p} + \alpha H_{\ul p}^-,   \\ \label{AL2,2r}
0 &= z \beta F_{\ul p}^- + \alpha H_{\ul p} + K_{\ul p}^- - K_{\ul p}.
\end{align}
Varying $\ul p \in \bbN_0^2$, the collection of evolution equations   
\begin{align}   \label{AL_p}
\begin{split}
& \AL_{\ul p} (\alpha, \beta) =
\begin{pmatrix}-i\alpha_{t_{\ul p}} - \alpha(g_{p_+,+} + g_{p_-,-}^-) 
+ f_{p_+ -1,+} - f_{p_- -1,-}^-\\
  -i\beta_{t_{\ul p}}+ \beta(g_{p_+,+}^- + g_{p_-,-}) - h_{p_- -1,-} 
  + h_{p_+ -1,+}^- \end{pmatrix}=0,  \\
& \hspace*{6.44cm} t_{\ul p}\in\bbR, \; \ul p=(p_-,p_+)\in\bbN_0,   
\end{split}
\end{align}
then defines the time-dependent Ablowitz--Ladik hierarchy. Explicitly, 
\begin{align} \no
& \AL_{(0,0)} (\alpha, \beta) =  \begin{pmatrix} -i \alpha_{t_{(0,0)}}- c_{(0,0)}\alpha \\ 
-i\beta_{t_{(0,0)}}+c_{(0,0)}\beta \end{pmatrix} 
=0,\\ \no
& \AL_{(1,1)} (\alpha, \beta) =  \begin{pmatrix}  -i \alpha_{t_{(1,1)}} 
- \gamma (c_{0,-}\alpha^- + c_{0,+}\alpha^+) 
- c_{(1,1)} \alpha \\
-i\beta_{t_{(1,1)}}+ \gamma (c_{0,+}\beta^- + c_{0,-}\beta^+) +
c_{(1,1)} \beta\end{pmatrix}=0,\\ 
& \AL_{(2,2)} (\alpha, \beta)  \\
& \quad =  \begin{pmatrix}\begin{matrix}-i \alpha_{t_{(2,2)}}-
\gamma \big(c_{0,+}\alpha^{++} \gamma^+ + c_{0,-}\alpha^{--} \gamma^-
- \alpha (c_{0,+}\alpha^+\beta^- + c_{0,-}\alpha^-\beta^+)\\
- \beta (c_{0,-}(\alpha^-)^2 + c_{0,+}(\alpha^+)^2)\big)\end{matrix}\\[3mm] 
\begin{matrix}-i\beta_{t_{(2,2)}}+
 \gamma \big(c_{0,-}\beta^{++} \gamma^+ + c_{0,+}\beta^{--} \gamma^-
- \beta (c_{0,+}\alpha^+\beta^- + c_{0,-}\alpha^-\beta^+)\\
- \alpha (c_{0,+}(\beta^-)^2 + c_{0,-}(\beta^+)^2)\big)\end{matrix}\end{pmatrix}
\no  \\ 
 & \qquad+ \begin{pmatrix}
-\gamma (c_{1,-} \alpha^- + c_{1,+} \alpha^+) - c_{(2,2)} \alpha\\
 \gamma (c_{1,+} \beta^- + c_{1,-} \beta^+) + c_{(2,2)} \beta\end{pmatrix}
=0, \, \text{ etc.,}   \no
\end{align}
represent the first few equations of the time-dependent Ablowitz--Ladik hierarchy. 
Here we recall the definition of $c_{\ul p}$ in \eqref{ALdefcp}.

The special case $\ul p =(1,1)$, $c_{0,\pm}=1$, and $c_{(1,1)}=-2$, that is,
\begin{equation}
\begin{pmatrix} -i \alpha_{t_1}- \gamma (\alpha^- + \alpha^+)  + 2\alpha \\    
-i\beta_{t_1}+ \gamma (\beta^- + \beta^+) - 2 \beta  \end{pmatrix}=0,
\end{equation}
represents \textit{the} Ablowitz--Ladik system \eqref{AL1.12}.

By \eqref{AL_p}, \eqref{ALg_l+}, and \eqref{ALg_l-},
the time derivative of $\gamma=1-\alpha \beta$ is given by
\begin{equation} \lb{AL2.14}
\gamma_{t_{\ul p}} = i \gamma \big((g_{p_+,+} - g_{p_+,+}^-) 
- (g_{p_-,-} - g_{p_-,-}^-) \big).
\end{equation}

The corresponding homogeneous equations are then defined by 
\begin{equation}  
\hAL_{\ul p} (\alpha, \beta) = 
\AL_{\ul p} (\alpha, \beta)\big|_{c_{0, \pm}=1, \, c_{\ell, \pm}=0, \, \ell=1,\dots,p_\pm}=0,  \quad \ul p=(p_-,p_+)\in\bbN_0^2.  
\end{equation} 

\begin{remark}  \lb{rAL2.13}
From \eqref{AL1,1}--\eqref{AL2,1} and the explicit computations of the coefficients 
$f_{\ell,\pm}$, $g_{\ell,\pm}$, and $h_{\ell,\pm}$, one concludes that the zero-curvature equation \eqref{ALzc p} 
and hence the Ablowitz--Ladik hierarchy is invariant under the scaling transformation 
\begin{equation}
\alpha \rightarrow \alpha_c = \{c\, \alpha(n)\}_{n\in\bbZ}, \quad 
\beta \rightarrow \beta_c = \{ \beta(n)/c\}_{n\in\bbZ}, \quad c \in \bbC\setminus \{0\}.
\end{equation}
In particular, solutions $\alpha, \beta$ of the stationary and time-dependent AL equations are determined only up to a multiplicative constant. 
\end{remark}

\begin{remark} \lb{rAL2.14}
$(i)$ The special choices $\beta=\pm\ol\alpha$, $c_{0,\pm}=1$ lead to the discrete nonlinear Schr\"odinger hierarchy. In particular, choosing $c_{(1,1)}=-2$ yields the discrete nonlinear Schr\"odinger equation in its usual form (see, e.g., 
\cite[Ch.\ 3]{AblowitzPrinariTrubatch:2004} and the references cited therein), 
\begin{equation}
-i\alpha_t - (1 \mp |\alpha|^2)(\alpha^- + \alpha^+) + 2\alpha = 0,   
\end{equation}
as its first nonlinear element. The choice $\beta = \ol \alpha$ is called the {\it defocusing} case, $\beta = - \ol \alpha$ represents the {\it focusing} case of the discrete nonlinear Schr\"odinger hierarchy.\\
$(ii)$ The alternative choice $\beta = \ol \alpha$, $c_{0,\pm} = \mp i$, leads to the hierarchy of Schur flows. In particular, choosing $c_{(1,1)} = 0$ yields  
\begin{equation}
\alpha_t - (1 - |\alpha|^2)(\alpha^+ - \alpha^-) = 0   
\end{equation}
as the first nonlinear element of this hierarchy (cf.\  \cite{AmmarGragg:1994}, 
\cite{FaybusovichGekhtman:1999}, \cite{FaybusovichGekhtman:2000},  
\cite{Golinskii:2006}, \cite{MukaihiraNakamura:2002}, \cite{Simon:2007}).
\end{remark}

Finally, we briefly recall the Lax pair $(L,P_{\ul p})$ for the Ablowitz--Ladik hierarchy and refer to \cite{GesztesyHoldenMichorTeschl:2007b} for detailed discussions of this topic. 

In the standard basis $\{\delta_m\}_{m\in\bbZ}$ in $\ell^2(\bbZ)$ defined by
\begin{equation}
\delta_m=\{\delta_{m,n}\}_{n\in\bbZ}, \; m\in\bbZ, \quad
\delta_{m,n}=\begin{cases} 1, &m=n, \\ 0, & m\neq n, \end{cases}
\lb{ALbasis}
\end{equation}
the underlying Lax difference expression $L$ is given by 
\begin{align}
L &= \left(\begin{smallmatrix} \ddots &&\hspace*{-8mm}\ddots
&\hspace*{-10mm}\ddots &\hspace*{-12mm}\ddots
&\hspace*{-14mm}\ddots &&&
\raisebox{-3mm}[0mm][0mm]{\hspace*{-6mm}{\Huge $0$}}
\\
&0& -\alpha(0) \rho(-1) & -\beta(-1)\alpha(0) &
-\alpha(1)\rho(0) & \rho(0) \rho(1)
\\
&& \rho(-1) \rho(0) & \beta(-1) \rho(0) &
-\beta(0) \alpha(1) & \beta(0) \rho(1) & 0
\\
&&&0& -\alpha(2) \rho(1) & -\beta(1) \alpha(2) &
-\alpha(3) \rho(2) & \rho(2) \rho(3)
\\
&&\raisebox{-4mm}[0mm][0mm]{\hspace*{-6mm}{\Huge $0$}} &&
\rho(1) \rho(2) & \beta(1) \rho(2) & -\beta(2) \alpha(3)
& \beta(2) \rho(3) & 0
\\
&&&&&\hspace*{-14mm}\ddots &\hspace*{-14mm}\ddots
&\hspace*{-14mm}\ddots &\hspace*{-8mm}\ddots &\ddots
\end{smallmatrix}\right)    \lb{ALLop} \\ 
&= \rho^- \rho \, \deven \, S^{--} + (\beta^-\rho \, \deven - \alpha^+\rho \, \dodd) S^- 
- \beta\alpha^+   \no \\ \lb{ALLrec}
& \quad + (\beta \rho^+ \, \deven - \alpha^{++} \rho^+ \, \dodd) S^+ 
+ \rho^+ \rho^{++} \, \dodd \, S^{++}, 
\end{align}
where $\deven$ and $\dodd$ denote the characteristic functions of the even and odd integers,
\begin{equation}
\deven = \chi_{_{2\bbZ}}, \quad \dodd = 1 - \deven = \chi_{_{2\bbZ +1}}.
\end{equation}
In particular, terms of the form $-\beta(n) \alpha(n+1)$ 
represent the diagonal $(n,n)$-entries, $n\in\bbZ$, in the infinite matrix
\eqref{ALLop}. In addition, we used the abbreviation
\begin{equation}
\rho = \gamma^{1/2} = (1-\alpha \beta)^{1/2}.  \lb{ALga}
\end{equation}

Next, let $T$ be a bounded operator in the Hilbert space $\ell^2(\bbZ)$ (with scalar product denoted by $(\dott,\dott)$). Given the standard basis 
\eqref{ALbasis} in $\ell^2(\bbZ)$, we represent $T$ by
\begin{equation}
T=\big(T(m,n)\big)_{(m,n)\in\bbZ^2}, \quad 
T(m,n)=(\delta_m,T \, \delta_n), \quad (m, n) \in\bbZ^2. \lb{ALTop}
\end{equation}
Actually, for our purpose below, it is sufficient that $T$ is an $N$-diagonal matrix for some $N\in\bbN$. Moreover, we introduce the upper and lower triangular parts 
$T_\pm$ of $T$ by
\begin{equation}
T_\pm=\big(T_\pm (m,n)\big)_{(m,n)\in\bbZ^2}, \quad
T_\pm (m,n)=\begin{cases} T(m,n), &\pm(n-m)>0, \\ 0, & \text{otherwise.}
\end{cases}
\lb{ALTpm}
\end{equation}

Then, the finite difference expression $P_{\ul p}$ is given by 
\begin{align}
\begin{split}
P_{\ul p} & = \f{i}{2} \sum_{\ell=1}^{p_+} c_{p_+ -\ell,+} \big( (L^\ell)_+ - (L^\ell)_- \big)
- \f{i}{2} \sum_{\ell=1}^{p_-} c_{p_- -\ell,-} \big( (L^{-\ell})_+ - (L^{-\ell})_- \big) \\
& \quad  - \f{i}{2} c_{\ul p} \, Q_d,  \quad  \ul p=(p_-,p_+)\in\bbN_0^2,    \lb{ALP_p} 
\end{split}
\end{align} 
with $Q_d$ denoting the doubly infinite diagonal matrix
\begin{equation} 
Q_d=\big((-1)^k \delta_{k,\ell} \big)_{k,\ell \in\bbZ} 
\end{equation}
and $c_{\ul p}=(c_{p_-,-} + c_{p_+,+})/2$.  
The commutator relations $[P_{\ul p},L]=0$ and $L_{t_{\ul p}}-[P_{\ul p},L]=0$ are 
then equivalent to the stationary and time-dependent 
Ablowitz--Ladik equations \eqref{ALstat} and \eqref{AL_p}, respectively.

\section{Properties of Stationary Algebro-Geometric Solutions \\
of the Ablowitz--Ladik Hierarchy} \lb{s3} 

In this section we present a quick review of properties of algebro-geometric
solutions of the stationary Ablowitz--Ladik hierarchy. We refer to 
\cite{GesztesyHoldenMichorTeschl:2007a}  and 
\cite{GesztesyHoldenMichorTeschl:2007} for detailed presentations.

We recall the hyperelliptic curve $\calK_p$ of genus $p$, where 
\begin{align} \label{ALcalK_p}
\begin{split}
& \calK_p \colon \calF_{p}(z,y) = y^2 - 4c_{0,+}^{-2}z^{2p_-}R_{\ul p}(z) 
= y^2 - \prod_{m=0}^{2p+1}(z-E_m) = 0, \\
& R_{\ul p}(z) = \bigg(\frac{c_{0,+}}{2z^{p_-}}\bigg)^2\prod_{m=0}^{2p+1}(z-E_m), 
\quad  \{E_m\}_{m=0}^{2p+1} \subset \bbC \setminus\{ 0\}, \; p=p_- + p_+ -1.   
\end{split}
\end{align}
Throughout this section we make the assumption: 

\begin{hypothesis} \lb{hAL3.1} 
Suppose that \begin{equation}
\alpha, \beta\in \bbC^{\bbZ} \, \text{ and } \,  \alpha(n)\beta(n)\notin \{0,1\}  
\, \text{ for all } \,  n\in\bbZ.   \lb{ALneq 0,1}
\end{equation}
In addition, assume that the affine part of the hyperelliptic curve $\calK_p$ in \eqref{ALcalK_p} is nonsingular, that is, suppose that 
\begin{equation}
E_m\neq E_{m'} \text{  for $m\neq m'$, \; $m,m'=0,1,\dots,2p+1$.}   \lb{ALEneqE}
\end{equation} 
\end{hypothesis}

The curve $\calK_p$ is compactified by joining two points $P_{\infty_\pm}$,
$P_{\infty_+}\neq P_{\infty_-}$, but for notational simplicity  the
compactification is also denoted by $\calK_p$. Points $P$ on
$\calK_p\setminus\{\Pinfp, \Pinfm\}$ are  represented as pairs $P=(z,y)$, where
$y(\dott)$ is the meromorphic function on $\calK_p$ satisfying
$\calF_{p}(z,y)=0$. The complex structure on $\calK_p$ is then defined in the usual way, see Appendix \ref{AL.sA}. Hence, $\calK_p$ becomes a two-sheeted hyperelliptic Riemann surface of genus $p$ in a standard manner.

We also emphasize that by fixing the curve $\calK_p$ (i.e., by fixing
$E_0,\dots,E_{2p+1}$), the summation constants $c_{1,\pm},\dots,c_{p_\pm,\pm}$ in 
$f_{p_\pm,\pm}$, $g_{p_\pm,\pm}$, and $h_{p_\pm,\pm}$ 
(and hence in the corresponding stationary $\sAL_{\ul p}$ equations) are uniquely determined as is clear from 
\eqref{ALc_ell} which establishes the summation constants $c_{\ell,\pm}$ as symmetric functions of 
$E_0^{\pm 1},\dots,E_{2p+1}^{\pm 1}$. 

For notational simplicity we will usually tacitly assume that $p\in\bbN$. 

We denote by $\{\mu_j(n)\}_{j=1,\dots,p}$ and $\{\nu_j(n)\}_{j=1,\dots,p}$ the zeros of $(\dott)^{p_-}F_{\ul p}(\dott,n)$ and $(\dott)^{p_- -1} H_{\ul p}(\dott,n)$, respectively. Thus, we may write 
\begin{align}   
F_{\ul p}(z)&= - c_{0,+}\alpha^+ z^{-p_-}\prod_{j=1}^{p}(z-\mu_j),   \label{ALmu(n)}  \\
H_{\ul p}(z)&= c_{0,+}\beta z^{-p_- +1}\prod_{j=1}^{p}(z-\nu_j), \label{ALnu(n)}
\end{align}
and we recall that (cf.\ \eqref{ALR})
\begin{equation}  
R_{\ul p} - G_{\ul p}^2 = - F_{\ul p} H_{\ul p}.   \lb{ALquad}
\end{equation}
The next step is crucial; it permits us to ``lift'' the zeros $\mu_j$ and $\nu_j$ from the complex plane $\bbC$ to the curve $\calK_p$.
From \eqref{ALquad} one infers that
\begin{equation}
R_{\ul p}(z) -G_{\ul p}(z)^2=0, \quad
z\in\{\mu_j,\nu_k\}_{j,k=1,\dots,p}. \lb{3.3.7A}
\end{equation}
We now introduce $\{ \hat \mu_j \}_{j=1,\dots,p}\subset \calK_p$ and
$\{ \hat \nu_j \}_{j=1,\dots,p}\subset \calK_p$ by
\begin{equation} \label{ALhmu}
\hat \mu_j(n)=(\mu_j(n), (2/c_{0,+})\mu_j(n)^{p_-} G_{\ul p}(\mu_j(n),n)), 
\quad j=1, \dots, p, \; n\in\bbZ,   
\end{equation}
and 
\begin{equation}  \label{ALhnu}
\hat \nu_j(n)=(\nu_j(n), - (2/c_{0,+})\nu_j(n)^{p_-} G_{\ul p}(\nu_j(n),n)), 
\quad j=1, \dots, p, \; n\in\bbZ.
\end{equation}

We also introduce the points $P_{0,\pm}$ by 
\begin{equation}
    \Pzpm=(0,\pm (c_{0,-}/c_{0,+}))\in\calK_p, \quad 
    \f{c_{0,-}^2}{c_{0,+}^2} = \prod_{m=0}^{2p+1} E_m.   \lb{AL3.10}
\end{equation}
We emphasize that $\Pzpm$ and $\Pinfpm$ are not necessarily on the same
sheet of $\calK_p$. Moreover,
\begin{align} \label{ALy_asymp}
y(P)&\underset{\zeta\to 0}{=} \begin{cases}
\mp \zeta^{-2p}(1+\Oh(\zeta)), & P\to\Pinfpm, \quad \zeta=1/z, \\
\pm (c_{0,-}/c_{0,+}) + \Oh(\zeta), & P\to\Pzpm, \quad \zeta=z. 
 \end{cases}      
\end{align}

Next we introduce the fundamental meromorphic function
 on $\calK_p$ by  
\begin{align} 
\phi(P,n) &= \frac{(c_{0,+}/2)z^{-p_-} y + G_{\ul p}(z,n)}{F_{\ul p}(z,n)}  \label{ALphi} \\
&= \frac{-H_{\ul p}(z,n)}{(c_{0,+}/2)z^{-p_-} y - G_{\ul p}(z,n)},   \label{ALphi1}  \\
& \hspace*{.9cm}  P=(z,y)\in \calK_p, \; n\in \bbZ,   \no 
\end{align}
with divisor  $(\phi(\dott,n))$ of $\phi(\dott,n)$ given by
\begin{equation} \label{AL(phi)}
(\phi(\dott,n)) = \calD_{P_{0,-} \hunu(n)} - \calD_{\Pinfm \humu(n)},   
\end{equation}
using \eqref{ALmu(n)} and \eqref{ALnu(n)}. Here we abbreviated 
\begin{equation}
\humu = \{\hat \mu_1, \dots, \hat \mu_{p}\}, \, 
\hunu = \{\hat \nu_1, \dots, \hat \nu_{p}\} \in\symq.
\end{equation}
For brevity, and in close analogy to the Toda hierarchy, we will frequently refer to 
$\humu$ and $\hunu$ as the Dirichlet and Neumann divisors, respectively. 

Given $\phi(\dott,n)$, the meromorphic stationary Baker--Akhiezer vector 
$\Psi(\dott,n,n_0)$ on $\calK_p$ is then defined by
\begin{align} \no
\Psi(P,n,n_0) &= \binom{\psi_1(P,n,n_0)}{\psi_2(P,n,n_0)}, \\  \label{ALpsi1}
\psi_1(P,n,n_0) &= \begin{cases}      
\prod_{n'=n_0 + 1}^n \big(z + \alpha(n') \phi^-(P,n')\big), & n \geq n_0 +1, \\
1,                      &  n=n_0, \\
\prod_{n'=n + 1}^{n_0} \big(z + \alpha(n') \phi^-(P,n')\big)^{-1}, & n \leq n_0 -1,
\end{cases}   \\
\psi_2(P,n,n_0) &= \phi(P,n_0)
\begin{cases}      
\prod_{n'=n_0 + 1}^n \big(z \beta(n') \phi^-(P,n')^{-1} + 1\big), & n \geq n_0 +1, \\
1,                      &  n=n_0, \\
\prod_{n'=n + 1}^{n_0} \big(z \beta(n') \phi^-(P,n')^{-1} + 1\big)^{-1}, & n \leq n_0 -1.
\end{cases}         \label{ALpsi2}
\end{align}
Basic properties of $\phi$ and $\Psi$ are summarized in the following result.

\begin{lemma} [\cite{GesztesyHoldenMichorTeschl:2007a}]  \lb{lAL3.1}
Suppose that $\alpha, \beta$ satisfy \eqref{ALneq 0,1} and the $\ul p$th stationary 
Ablowitz--Ladik system \eqref{ALstat}. Moreover, assume \eqref{ALcalK_p} and 
\eqref{ALEneqE} and let
$P=(z,y) \in \calK_p\setminus \{\Pinfp, \Pinfm,\Pzp,\Pzm\}$, $(n, n_0) \in \bbZ^2$.
Then $\phi$ satisfies the Riccati-type equation
\begin{align} \label{ALriccati} 
& \alpha \phi(P)\phi^-(P) - \phi^-(P) + z \phi(P) = z \beta,   \\
\intertext{as well as}
  \label{ALphi 1}
& \phi(P) \phi(P^*) = \frac{H_{\ul p}(z)}{F_{\ul p}(z)},\\ \label{ALphi 2}
& \phi(P) + \phi(P^*) = 2\frac{G_{\ul p}(z)}{F_{\ul p}(z)},\\ \label{ALphi 3}
& \phi(P) - \phi(P^*) = c_{0,+}z^{-p_-} \frac{y(P)}{F_{\ul p}(z)}.
\end{align}
The vector $\Psi$ satisfies
\begin{align} 
& U(z) \Psi^-(P)=\Psi(P),  \label{ALpsi 2} \\ 
\label{ALpsi 3}
& V_{\ul p}(z)\Psi^-(P)= - (i/2)c_{0,+} z^{-p_-} y \Psi^-(P), \\ 
& \psi_2(P,n,n_0) = \phi(P,n) \psi_1(P,n,n_0),   \label{ALpsi 1} \\ \label{ALpsi 4}
& \psi_1(P,n,n_0) \psi_1(P^*,n,n_0) = z^{n-n_0} \frac{F_{\ul p}(z,n)}{F_{\ul p}(z,n_0)} 
\pgam(n,n_0),
\\ \label{ALpsi 5}
& \psi_2(P,n,n_0) \psi_2(P^*,n,n_0) = z^{n-n_0} \frac{H_{\ul p}(z,n)}{F_{\ul p}(z,n_0)} 
\pgam(n,n_0),\\
& \psi_1(P,n,n_0) \psi_2(P^*,n,n_0) +\psi_1(P^*,n,n_0) \psi_2(P,n,n_0) \label{ALpsi 6} \\
& \quad =2 z^{n-n_0} \frac{G_{\ul p}(z,n)}{F_{\ul p}(z,n_0)} 
\pgam(n,n_0),\no \\
& \psi_1(P,n,n_0) \psi_2(P^*,n,n_0) -\psi_1(P^*,n,n_0) \psi_2(P,n,n_0) \label{ALpsi 7} \\
& \quad =-c_{0,+} z^{n-n_0-p_-} \frac{y}{F_{\ul p}(z,n_0)}  \pgam(n,n_0), \no
\end{align}
where we used the abbreviation 
\begin{equation} \lb{ALpgam}
\pgam(n,n_0) = \begin{cases}      
\prod_{n'=n_0 + 1}^n \gamma(n') & n \geq n_0 +1, \\
1                      &  n=n_0, \\
\prod_{n'=n + 1}^{n_0} \gamma(n')^{-1}  & n \leq n_0 -1.
\end{cases}
\end{equation}
\end{lemma}

Combining the Laurent polynomial recursion approach of Section \ref{s2} with 
\eqref{ALmu(n)} and \eqref{ALnu(n)} readily yields trace formulas for $f_{\ell,\pm}$ and $h_{\ell,\pm}$ in terms of symmetric functions of the zeros $\mu_j$ and $\nu_k$ of $(\dott)^{p_-}F_{\ul p}$ and $(\dott)^{p_- -1}H_{\ul p}$, respectively. For simplicity we just record the simplest cases.

\begin{lemma}  [\cite{GesztesyHoldenMichorTeschl:2007a}]  \lb{lAL3.2}
Suppose that $\alpha, \beta$ satisfy \eqref{ALneq 0,1} and the $\ul p$th stationary  
Ablowitz--Ladik system \eqref{ALstat}. Then, 
\begin{align} 
 \frac{\alpha}{\alpha^+}&= (-1)^{p+1}\f{c_{0,+}}{c_{0,-}}
\prod_{j=1}^{p}\mu_j,   \label{ALtr1} \\ 
\frac{\beta^+}{\beta}&= (-1)^{p+1}\f{c_{0,+}}{c_{0,-}}
\prod_{j=1}^{p}\nu_j,   \label{ALtr2} \\  
\sum_{j=1}^{p}\mu_j &= \alpha^+ \beta
- \gamma^+ \frac{\alpha^{++}}{\alpha^+} 
- \frac{c_{1,+}}{c_{0,+}},   \label{ALtr3} \\
\sum_{j=1}^{p}\nu_j &= \alpha^+ \beta
- \gamma \frac{\beta^-}{\beta} 
- \frac{c_{1,+}}{c_{0,+}}.  \label{ALtr4}
\end{align}
\end{lemma}

\begin{remark}  \lb{rAL3.4}
The trace formulas in Lemma \ref{lAL3.2} illustrate why we assumed the condition 
$\alpha(n)\beta(n)\neq 0$ for all $n\in\bbN$ throughout this paper. Moreover, the following section shows that this condition is intimately connected with admissible divisors 
$\calD_{\humu}, \calD_{\hunu}$ avoiding the exceptional points $\Pinfpm, \Pzpm$. On the other hand, as is clear from the matrix representation \eqref{ALLop} of the Lax difference expression $L$, if $\alpha(n_0)\beta(n_0) = 1$ for some $n_0\in\bbN$, and hence $\rho(n_0)=0$, the infinite matrix $L$ splits into a direct sum of two half-line matrices $L_\pm (n_0)$ (in analogy to the familiar singular case of infinite Jacobi matrices $aS^+ + a^- S^- + b$ on $\bbZ$ with $a(n_0)=0$). This explains why we assumed $\alpha(n)\beta(n)\neq 1$ for all $n\in\bbN$ throughout this paper.
\end{remark}

Since nonspecial divisors and the linearization property of the Abel map when applied
to $\calD_{\humu}$ and $\calD_{\hunu}$
will play a fundamental role later on, we also recall the following facts.

\begin{lemma} [\cite{GesztesyHoldenMichorTeschl:2007a}, 
\cite{GesztesyHoldenMichorTeschl:2007}] \label{lAL3.4} 
Suppose that $\alpha, \beta$ satisfy \eqref{ALneq 0,1} and the
$\ul p$th stationary Ablowitz--Ladik system \eqref{ALstat}. Moreover, assume 
\eqref{ALcalK_p} and \eqref{ALEneqE} and let $n\in\bbZ$. Let $\calD_{\humu}$,
$\humu=\{\hmu_1,\dots,\hmu_{p}\}$, and $\calD_{\hunu}$,
$\hunu=\{\hunu_1,\dots,\hunu_{p}\}$, be the pole and zero divisors of degree
$p$, respectively, associated with $\alpha$, $\beta$, and $\phi$ defined
according to \eqref{ALhmu} and \eqref{ALhnu}, that is,
\begin{align}
\begin{split}
\hat\mu_j (n) &= (\mu_j (n), (2/c_{0,+}) \mu_j(n)^{p_-} G_{\ul p}(\mu_j(n),n)), 
\quad j=1,\dots,p, \\
\hat\nu_j (n) &= (\nu_j (n),- (2/c_{0,+}) \nu_j(n)^{p_-} G_{\ul p}(\nu_j(n),n)), 
\quad j=1,\dots,p.
\end{split}
\end{align}
Then $\calD_{\humu(n)}$ and $\calD_{\hunu(n)}$ are nonspecial for all
$n\in\bbZ$. Moreover, the Abel map linearizes the auxiliary divisors 
$\calD_{\humu}$ and $\calD_{\hunu}$ in the sense that
\begin{align}
\amap(\calD_{\humu(n)}) &=\amap(\calD_{\humu(n_{0})})
+ (n-n_{0}) \underline{A}_{\Pzm}(\Pinfp), \label{AL3.63} \\
\amap(\calD_{\hunu(n)})&=\amap(\calD_{\hunu(n_{0})})
+ (n-n_{0}) \underline{A}_{\Pzm}(\Pinfp), \label{AL3.64}
\end{align}
where $Q_0 \in \calK_p$ is a given base point. In addition,
\begin{equation}
\ual_{Q_0} (\calD_{\hat{\ul{\nu}}(n)}) = \ual_{Q_0}
(\calD_{\humu (n)}) + \underline{A}_{P_{0,-}}
(\Pinfm).    \lb{ALS3.65} 
\end{equation}
\end{lemma}

For theta function representations of $\alpha$ and $\beta$ we refer to 
\cite{GesztesyHoldenMichorTeschl:2007} and the references cited therein. 
These theta function representations also show that $\gamma(n)\notin \{0,1\}$ 
for all $n\in\bbZ$, that is, the second condition in \eqref{ALneq 0,1} is 
satisfied for the stationary algebro-geometric AL solutions discussed in 
this section provided the associated Dirichlet and Neumann divisors are 
admissible.

\section{The Stationary Algorithm} 
\lb{ALSs4}  

The aim of this section is to derive an algorithm that enables
one to construct algebro-geometric solutions for the
stationary Ablowitz--Ladik hierarchy for general initial data. Equivalently, we offer 
a solution of the inverse algebro-geometric spectral problem for general 
Lax operators $L$ in \eqref{ALLrec}, starting with initial divisors in general position.

Up to the end of Section \ref{s3} the material was based on the assumption that 
$\alpha, \beta \in \bbC^{\bbZ}$
satisfy the $\ul p$th stationary AL system \eqref{ALstat}. Now we
embark on the corresponding inverse problem consisting of constructing
a solution of \eqref{ALstat} given certain initial data. More precisely,
we seek to construct solutions $\alpha, \beta \in\bbC^{\bbZ}$  satisfying the
$\ul p$th stationary Ablowitz--Ladik system \eqref{ALstat} starting from a properly restricted set $\calM_0$ of admissible nonspecial Dirichlet divisor initial data
$\calD_{\humu(n_0)}$ at some fixed $n_0\in\bbZ$, 
\begin{align}
\begin{split}
\humu(n_0)&=\{\hmu_1(n_0),\dots,\hmu_{p}(n_0)\}\in \calM_0, \quad 
\calM_0\subset\symq, \\
\hat\mu_j (n_0) &= (\mu_j (n_0), (2/c_{0,+}) \mu_j(n_0)^{p_-} 
G_{\ul p}(\mu_j(n_0),n_0)), 
\quad j=1,\dots,p.  \lb{ALS4.1}
\end{split}
\end{align}

For convenience we will frequently use the phrase that $\alpha, \beta$ 
{\it blow up} in this manuscript whenever one of the divisors $\calD_{\humu}$ or $\calD_{\hunu}$ hits one of the points $\Pinfpm$, $\Pzpm$.

Of course we would like to ensure that the sequences $\alpha, \beta$ obtained via our 
algorithm do not blow up. To investigate when this happens, we study 
the image of our divisors under the Abel map. A key
ingredient in our analysis will be \eqref{AL3.63} which yields a linear
discrete dynamical system on the Jacobi variety
$J(\calK_p)$. In particular, we will be led to investigate solutions 
$\calD_{\humu}$,  $\calD_{\hunu}$ of the discrete initial value problem
\begin{align}
\begin{split}
&\amap(\calD_{\humu(n)}) =\amap(\calD_{\humu(n_{0})})
+ (n-n_{0}) \underline{A}_{\Pzm}(\Pinfp), \lb{ALS4.2} \\
&\humu(n_0)=\{\hmu_1(n_0),\dots,\hmu_{p}(n_0)\}\in \symq,
\end{split}
\end{align}
respectively
\begin{align}
\begin{split}
&\amap(\calD_{\hunu(n)}) =\amap(\calD_{\humu(n_{0})}) + \ul{A}_{P_{0,-}}(\Pinfm)
+ (n-n_{0}) \underline{A}_{\Pzm}(\Pinfp), \lb{ALS4.2a} \\
&\hunu(n_0)=\{\hnu_1(n_0),\dots,\hnu_{p}(n_0)\}\in \symq,
\end{split}
\end{align}
where $Q_0\in\calK_p$ is a given base point.  Eventually, we will be
interested in solutions $\calD_{\humu}$, $\calD_{\hunu}$ of \eqref{ALS4.2}, 
\eqref{ALS4.2a} with initial
data $\calD_{\humu(n_0)}$ satisfying \eqref{ALS4.1} and $\calM_0$ to be
specified as in (the proof of) Lemma \ref{ALSl4.1}.

Before proceeding to develop the stationary Ablowitz--Ladik algorithm, we briefly
analyze the dynamics of \eqref{ALS4.2}. 

\begin{lemma} \lb{ALSl4.1}
Let $n\in\bbZ$ and suppose that 
$\calD_{\humu(n)}$ is defined via \eqref{ALS4.2} for some divisor 
$\calD_{\humu(n_0)}\in\symq$. \\
$(i)$ If $\calD_{\humu(n)}$ is nonspecial and does not contain any of the points
$P_{0,\pm}$, $P_{\infty_\pm}$, and $\calD_{\humu(n+1)}$
contains one of the points $P_{0,\pm}$, $P_{\infty_\pm}$, then 
$\calD_{\humu(n+1)}$ contains $\Pzm$ or $\Pinfm$ but not $\Pinfp$ or $\Pzp$. \\
$(ii)$ If $\calD_{\humu(n)}$ is nonspecial and $\calD_{\humu(n+1)}$ is
special, then $\calD_{\humu(n)}$ contains at least one of the points 
$P_{\infty_+}$, $P_{\infty_-}$ and one of the points $P_{0,+}$, $P_{0,-}$. \\
$(iii)$ Item $(i)$ holds if $n+1$ is replaced by $n-1$, $P_{\infty_+}$
by $P_{\infty_-}$, and $\Pzp$ by $\Pzm$.  \\
$(iv)$ Items $(i)$--$(iii)$ also hold for $\calD_{\hunu(n)}$.
\end{lemma} 
\begin{proof}
$(i)$ Suppose one point in $\calD_{\humu(n+1)}$ equals $P_{\infty_+}$ and 
denote the remaining ones by $\calD_{\tilde{\underline{\mu}}(n+1)}$. Then
\eqref{ALS4.2} implies that 
$\underline{\alpha}_{Q_0}(\calD_{\tilde{\underline{\mu}}(n+1)}) 
+ \underline{A}_{Q_0}(P_{\infty_+})
= \underline{\alpha}_{Q_0}(\calD_{\humu(n)}) + \underline{A}_{P_{0,-}}(P_{\infty_+})$. 
Since $\calD_{\humu(n)}$ is assumed to be nonspecial one concludes 
$\calD_{\humu(n)}=\calD_{\tilde{\underline{\mu}}(n+1)} + \calD_{P_{0,-}}$, contradicting
our assumption on $\calD_{\humu(n)}$. The statement for $P_{0,+}$ follows similarly;
here we choose $Q_0$ to be a branch point of $\calK_p$ such that 
$\underline{A}_{Q_0}(P^*)=-\underline{A}_{Q_0}(P)$. \\
$(ii)$ Next, we choose $Q_0$ to be a branch point of $\calK_p$. If 
$\calD_{\humu(n+1)}$ is special, then it contains a pair of points $(Q,Q^*)$
whose contribution will cancel under the Abel map, that is,
$\underline{\alpha}_{Q_0}(\calD_{\humu(n+1)})= 
\underline{\alpha}_{Q_0}(\calD_{{\underline{\hat \eta}}(n+1)})$
for some $\calD_{{\underline{\hat\eta}}(n+1)}\in \sym^{p-2} (\calK_p)$. 
Invoking \eqref{ALS4.2} then shows that 
$\underline{\alpha}_{Q_0}(\calD_{\humu(n)})=
\underline{\alpha}_{Q_0}(\calD_{{\underline{\hat\eta}}(n+1)}) + 
\underline{A}_{Q_0}(P_{\infty_-}) + \underline{A}_{Q_0}(P_{0,-})$. 
As $\calD_{\humu(n)}$ was assumed to be 
nonspecial, this shows that $\calD_{\humu(n)} =
\calD_{{\underline{\hat \eta}}(n+1)} + \calD_{P_{\infty_-}} + \calD_{P_{0,-}}$, as claimed.\\
$(iii)$ This is proved as in item $(i)$. \\
$(iv)$ Since $\calD_{\hunu(n)}$ satisfies the same equation as 
$\calD_{\humu(n)}$ in \eqref{ALS4.2} (cf.\ \eqref{AL3.64}), 
items $(i)$--$(iii)$ also hold for $\calD_{\hunu(n)}$.
\end{proof} 

We also note the following result: 

\begin{lemma} \lb{ALSl4.1a}
Let $n\in\bbZ$ and assume that $\calD_{\humu(n)}$ and $\calD_{\hunu(n)}$ are nonspecial. Then $\calD_{\humu(n)}$ contains $\Pzm$ if and only if $\calD_{\hunu(n)}$ contains $\Pinfm$. Moreover, $\calD_{\humu(n)}$ contains $\Pinfp$ if and only if $\calD_{\hunu(n)}$ contains $\Pzp$.
\end{lemma} 
\begin{proof}
Suppose a point in $\calD_{\humu(n)}$ equals $P_{0,-}$ and 
denote the remaining ones by $\calD_{\tilde{\underline{\mu}}(n)}$.
By \eqref{ALS3.65}, 
\begin{equation}
\underline{\alpha}_{Q_0}(\calD_{{\hunu}(n)}) 
= \underline{\alpha}_{Q_0}(\calD_{\tilde{\underline{\mu}}(n)}) 
+ \underline{A}_{Q_0}(P_{0,-}) + \underline{A}_{P_{0,-}}(P_{\infty_-})
= \underline{\alpha}_{Q_0}(\calD_{\tilde{\underline{\mu}}(n)}) 
+ \underline{A}_{Q_0}(P_{\infty_-}). 
\end{equation} 
Since $\calD_{\hunu(n)}$ is nonspecial, 
$\calD_{\hunu(n)}$ contains $P_{\infty_-}$, and vice versa.
The second statement follows similarly.
\end{proof} 

Let us call the points $P_{\infty_+}$, $P_{\infty_-}$, $\Pzp$, and $\Pzm$ 
{\it exceptional points}. Then Lemma \ref{ALSl4.1} yields the following behavior of 
$\calD_{\humu(n)}$ assuming one starts 
with some nonspecial initial divisor $\calD_{\humu(n_0)}$ without exceptional points:
As $n$ increases, $\calD_{\humu(n)}$ stays nonspecial as long as it does not include 
exceptional points. If an exceptional point appears, $\calD_{\humu(n)}$ is still nonspecial
and contains $\Pzm$ or $P_{\infty_-}$ at least once (but not $\Pzp$ and $P_{\infty_+}$).
Further increasing $n$, all instances of $\Pzm$ and $P_{\infty_-}$ will be rendered into
$\Pzp$ and $P_{\infty_+}$, until we have again a nonspecial divisor that has the same
number of $\Pzp$ and $P_{\infty_+}$ as the first one had of $\Pzm$ and $P_{\infty_-}$.
Generically, one expects the subsequent divisor to be nonspecial without exceptional points again.

Next we show that most initial divisors are well-behaved in the 
sense that their iterates stay away from $\Pinfpm$, $\Pzpm$. Since we 
want to show that this set is of full measure, it will be 
convenient to identify $\symq$ with the Jacobi variety $J(\calK_p)$
via the Abel map and take the Haar measure on $J(\calK_p)$.
Of course, the Abel map is only injective when restricted to the
set of nonspecial divisors, but these are the only ones we are interested in. 

\begin{lemma} \lb{ALSl4.2} 
The set $\calM_0\subset\symq$ of initial divisors $\calD_{\humu(n_0)}$
for which $\calD_{\humu(n)}$ and $\calD_{\hunu(n)}$, defined via
\eqref{ALS4.2} and \eqref{ALS4.2a}, are admissible $($i.e., do not contain the points 
$P_{\infty_\pm}, P_{0,\pm}$$)$ and hence are nonspecial for all $n\in\bbZ$, forms a dense set of full measure  in the set $\symq$ of positive divisors of degree $p$.  
\end{lemma} 
\begin{proof}
Let $\calM_{\infty,0}$ be the set of divisors in $\symq$ for which 
(at least) one point is equal to $P_{\infty_{\pm}}$ or $\Pzpm$. The image
$\underline{\alpha}_{Q_0}(\calM_{\infty,0})$ of $\calM_{\infty,0}$ is then contained in the following set, 
\begin{align}
\underline{\alpha}_{Q_0}(\calM_{\infty,0}) & \subseteq 
\bigcup_{P\in \{\Pzpm,P_{\infty_\pm}\}} \big(\underline{A}_{Q_0}(P) + 
\underline{\alpha}_{Q_0}(\sym^{p-1} (\calK_p))\big) \subset J(\calK_p). 
\end{align}
Since the (complex) dimension of $\sym^{p-1} (\calK_p)$ is 
$p-1$, its image must be of measure zero by Sard's theorem (see, e.g.,
\cite[Sect.\ 3.6]{AbrahamMarsdenRatiu:1988}). Similarly, let $\calM_{\rm sp}$ be the set of
special divisors, then its image is given by 
\begin{equation}
\underline{\alpha}_{Q_0}(\calM_{\rm sp}) =
\underline{\alpha}_{Q_0}(\sym^{p-2} (\calK_p)), 
\end{equation}
assuming $Q_0$ to be a branch point. In particular, we conclude that 
$\underline{\alpha}_{Q_0}(\calM_{\rm sp})\subset
\underline{\alpha}_{Q_0}(\calM_{\infty,0})$ and thus
$\underline{\alpha}_{Q_0}(\calM_{\rm sing})=
\underline{\alpha}_{Q_0}(\calM_{\infty,0})$ has measure zero, where
\begin{equation}
\calM_{\rm sing}=\calM_{\infty,0}\cup \calM_{\rm sp}. 
\end{equation}
Hence, 
\begin{equation}
\calS_\mu = \bigcup_{n\in\bbZ} \big(\underline{\alpha}_{Q_0}(\calM_{\rm sing}) 
+ n\underline{A}_{P_{0,-}}(P_{\infty_+}) \big) \quad\text{and}\quad
\calS_\nu = \calS_\mu + \ul{A}_{P_{0,-}}(\Pinfm) \lb{ALS4.5}
\end{equation}
are of measure zero as well. But the set $\calS_\mu\cup \calS_\nu$ contains all initial divisors for which $\calD_{\humu(n)}$ or $\calD_{\hunu(n)}$ will hit $\Pinfpm$ or 
$\Pzpm$, or become special at some $n\in\bbZ$. We denote by $\calM_0$ the inverse image of the complement of the set $\calS_\mu\cup \calS_\nu$ under the Abel map,
\begin{equation}
\calM_0= \underline{\alpha}_{Q_0}^{-1} \big(\symq \setminus (\calS_\mu\cup \calS_\nu)\big).
\end{equation}
Since $\calM_0$ is of full measure, it is automatically dense in $\symq$.
\end{proof} 

\medskip

Next, we describe the stationary Ablowitz--Ladik algorithm. Since this is a somewhat
lengthy affair, we will break it up into several steps.  

\smallskip

\noindent {\bf The Stationary Ablowitz--Ladik Algorithm:} 

We prescribe the following data\\
$(i)$ The coefficient $\alpha(n_0) \in \bbC \setminus\{0\}$ and the constant 
$c_{0,+}\in \bbC \setminus\{0\}$. \\
$(ii)$ The set
\begin{equation}
\{E_m\}_{m=0}^{2p+1} \subset\bbC\setminus\{0\}, \quad E_m\neq E_{m'} \, \text{ for } \, 
m\neq m', \quad m,m'=0,\dots,2p+1,   \lb{ALS4.8}
\end{equation}
for some fixed $p\in\bbN$. Given $\{E_m\}_{m=0}^{2p+1}$, we introduce 
the function 
\begin{equation} \label{R_2p}
R_{\ul p}(z) = \bigg(\frac{c_{0,+}}{2 z^{p_-}}\bigg)^2 \prod_{m=0}^{2p+1}(z-E_m)
\end{equation}
and the hyperelliptic curve $\calK_p$ with nonsingular affine part as in 
\eqref{ALcalK_p}. \\
$(iii)$ The nonspecial divisor 
\begin{equation}
\calD_{\humu(n_0)}\in \symq,   \lb{ALS4.9}
\end{equation}
where $\humu(n_0)$ is of the form
\begin{equation}
\humu(n_0) =\{\underbrace{\hat\mu_1(n_0),\dots,\hat\mu_1(n_0)}_{p_1(n_0) \text{
times}},
\dots,\underbrace{\hat\mu_{q(n_0)},\dots,
\hat\mu_{q(n_0)}}_{p_{q(n_0)}(n_0) \text{ times}}\}    \lb{ALS4.10}
\end{equation}
with 
\begin{equation}
\hat\mu_k(n_0)=(\mu_k(n_0),y(\hat\mu_k(n_0))), \quad  
\mu_k(n_0)\neq \mu_{k'}(n_0) \, \text{ for } \, k\neq k', \; 
k,k'=1,\dots,q(n_0),    \lb{ALS4.11}
\end{equation}
and
\begin{equation}
p_k(n_0)\in\bbN, \; k=1,\dots,q(n_0), \quad 
\sum_{k=1}^{q(n_0)} p_k(n_0) = p.   \lb{ALS4.12}
\end{equation}
With $\{E_m\}_{m=0}^{2p+1}$, $\calD_{\humu(n_0)}$, $\alpha(n_0)$, and $c_{0,+}$ prescribed, we next
introduce the following quantities (for $z\in\bbC\setminus\{0\}$):
\begin{align} \lb{ALS4.13}
\alpha^+(n_0) &= \alpha(n_0) \bigg(\prod_{m=0}^{2p+1}E_m\bigg)^{1/2} 
\prod_{k=1}^{q(n_0)}\mu_k(n_0)^{-p_k(n_0)}, \\  \lb{ALS4.14}
c_{0,-}^2 &= c_{0,+}^2 \prod_{m=0}^{2p+1}E_m, \\ \lb{ALS4.15}
F_{\ul p}(z,n_0) &= - c_{0,+}\alpha^+(n_0)z^{-p_-} \prod_{k=1}^{q(n_0)} (z-\mu_k(n_0))^{p_k(n_0)},\\ \label{ALS4.16}
G_{\ul p}(z,n_0) 
& = \frac{1}{2} \bigg(\frac{1}{\alpha(n_0)}-\frac{z}{\alpha^+(n_0)} \bigg) F_{\ul p}(z,n_0)\\ \no
&\hspace*{.35cm} -\frac{z}{2\alpha^+(n_0)}F_{\ul p}(z,n_0)
\sum_{k=1}^{q(n_0)}\sum_{\ell=0}^{p_k(n_0)-1} 
\f{\big(d^\ell
\big(\zeta^{-1} y(P)\big)/d\zeta^\ell\big)\big|_{P=(\zeta,\eta)
=\hat \mu_k(n_0)}}{\ell!(p_k(n_0)-\ell-1)!}\\ \no
& \hspace*{-.5cm} \times 
\Bigg(\f{d^{p_k(n_0)-\ell-1}}{d \zeta^{p_k(n_0)-\ell-1}}\Bigg(
(z-\zeta)^{-1}\prod_{k'=1, \, k'\neq k}^{q(n_0)} 
(\zeta-\mu_{k'}(n_0))^{-p_{k'}(n_0)} 
\Bigg)\Bigg)\Bigg|_{\zeta=\mu_k(n_0)}.
\end{align}
Here the sign of the square root is chosen according to \eqref{ALS4.11}. 

Next we record a series of facts: \\
\smallskip
$\textbf{(I)}$ By construction (cf.\ Lemma \ref{ALlB.1}),
\begin{align}
\begin{split}
& \f{d^\ell \big(G_{\ul p}(z,n_0)^2\big)}{dz^\ell}\bigg|_{z=\mu_k(n_0)}= 
\f{d^\ell R_{\ul p}(z)}{dz^\ell}\bigg|_{z=\mu_k(n_0)}, \\
& z\in\bbC\setminus\{0\}, \quad \ell=0,\dots,p_k(n_0)-1, \; k=1,\dots,q(n_0).  
\lb{ALS4.17}
\end{split}
\end{align}
$\textbf{(II)}$ Since $\calD_{\humu(n_0)}$ is nonspecial by hypothesis,
one concludes that
\begin{equation}
p_k(n_0)\geq 2 \, \text{ implies } \, R_{\ul p}(\mu_k(n_0))\neq 0, \quad 
k=1,\dots,q(n_0).    \lb{ALS4.18}
\end{equation}
$\textbf{(III)}$ By \eqref{ALS4.16} and \eqref{ALS4.17} one infers that $F_{\ul p}$
divides $G_{\ul p}^2-R_{\ul p}$. \\
$\textbf{(IV)}$ By \eqref{R_2p} and \eqref{ALS4.16} one verifies
that   
\begin{align} \lb{ALS4.20}
G_{\ul p}(z,n_0)^2-R_{\ul p}(z) &\underset{z\to \infty}{=} \Oh(z^{2p_+ -1}), \\  
G_{\ul p}(z,n_0)^2-R_{\ul p}(z) &\underset{z\to 0}{=} \Oh(z^{-2p_- +1}).  
\lb{ALS4.20z}
\end{align}
By $\textbf{(III)}$ and $\textbf{(IV)}$ we may write
\begin{equation}
G_{\ul p}(z,n_0)^2-R_{\ul p}(z)=F_{\ul p}(z,n_0) \check H_{q,r}(z,n_0), 
\quad z\in\bbC\setminus\{0\},   \lb{ALS4.20a}
\end{equation}
for some $q\in\{0,\dots,p_- -1\}$, $r\in\{0,\dots,p_+\}$, where 
$\check H_{q,r}(z,n_0)$ is a Laurent polynomial of the form 
$c_{-q}z^{-q}+\dots+c_{r}z^r$. 
If, in fact, $\check H_{0,0} =0$, then $R_{\ul p}(z)=G_{\ul p}(z,n_0)^2$ would
yield double zeros of $R_{\ul p}$, contradicting our basic hypothesis
\eqref{ALS4.8}. Thus we conclude that in the case $r=q=0$, $\check H_{0,0}$ 
cannot vanish identically and hence we may break up \eqref{ALS4.20a} in the
following manner 
\begin{equation}
\begin{split}
\check \phi(P,n_0)=\f{G_{\ul p}(z,n_0)+(c_{0,+}/2)z^{-p_-}y}{F_{\ul p}(z,n_0)}
=\f{\check H_{q,r}(z,n_0)}{G_{\ul p}(z,n_0)-(c_{0,+}/2)z^{-p_-}y}, \\
 P=(z,y)\in\calK_p. \lb{ALS4.20b}
\end{split}
\end{equation} 
Next we decompose
\begin{equation}
\check H_{q,r}(z,n_0) = C z^{-q}\prod_{j=1}^{r+q}(z-\nu_j(n_0)), 
\quad z \in \bbC\setminus\{0\},
\lb{ALS4.20c}
\end{equation} 
where $C \in\bbC\setminus\{0\}$ and
$\{\nu_j(n_0)\}_{j=1}^{r+q}\subset 
\bbC$ (if $r=q=0$ we replace the product in \eqref{ALS4.20c} by $1$). 
By inspection of the local zeros and poles as well
as the behavior near $P_{0,\pm}$, $\Pinfpm$ of the function $\check\phi (\dott,n_0)$ using \eqref{ALy_asymp},
its divisor, $\big(\check \phi(\dott,n_0) \big)$, is given by
\begin{equation}
\big(\check \phi(\dott,n_0) \big)=\calD_{P_{0,-} 
\hat{\underline{\nu}}(n_0)}-\calD_{\Pinfm\hat{\underline{\mu}}(n_0)},
\lb{ALS4.20d}
\end{equation}
where
\begin{equation}
\hat{\ul{\nu}}(n_0)=\{\underbrace{P_{0,-},\dots,P_{0,-}}_{p_- -1-q \text{ times}},
\hat{\nu}_1(n_0),\dots,
\hat{\nu}_{r+q}(n_0),\underbrace{\Pinfp,\dots,\Pinfp}_{p_+ -r 
\text{ times}}\}.  
\lb{ALS4.20e}
\end{equation}

In the following we call a positive divisor of degree $p$ {\it admissible}
if it does not contain any of the points $P_{\infty_\pm}, P_{0,\pm}$.

Hence, 
\begin{equation}
\calD_{\hat{\ul{\nu}}(n_0)} \, \text{ is an admissible divisor if
and only if } \, r=p_+ \, \text{ and } \, q=p_- -1.  \lb{ALS4.20f}
\end{equation}
We note that 
\begin{equation}
\ual_{Q_0} (\calD_{\hat{\ul{\nu}}(n_0)}) = \ual_{Q_0}
(\calD_{\humu (n_0)}) + \underline{A}_{P_{0,-}}
(\Pinfm), \lb{ALS4.20g} 
\end{equation}
in accordance with \eqref{ALS3.65}. \\
$\textbf{(V)}$ Assuming that \eqref{ALS4.20}, \eqref{ALS4.20z} are precisely of order
$z^{\pm(2p_\pm -1)}$, that is, assuming $r=p_+$ and $q=p_- -1$ in 
\eqref{ALS4.20a}, we rewrite \eqref{ALS4.20a}
in the more appropriate manner 
\begin{equation}
G_{\ul p}(z,n_0)^2-R_{\ul p}(z)=F_{\ul p}(z,n_0) H_{\ul p}(z,n_0), 
\quad z\in\bbC\setminus\{0\}.   \lb{ALS4.21}
\end{equation}
(We will later discuss conditions which indeed guarantee that $q=p_- -1$ and 
$r=p_+$, cf.\ \eqref{ALS4.20f} and the discussion in step $\textbf{(X)}$ below.)
By construction, $H_{\ul p}(\dott,n_0)$ is then of the type
\begin{align}
&H_{\ul p}(z,n_0)=c_{0,+}\beta(n_0) z^{-p_- +1}\prod_{k=1}^{\ell(n_0)} 
(z-\nu_k(n_0))^{s_k(n_0)}, 
\quad \sum_{k=1}^{\ell(n_0)} s_k(n_0) =p,  \no \\
& \hspace*{1.3cm} \nu_k(n_0) \neq \nu_{k'}(n_0) \, \text{ for } \, k\neq k', \; 
k, k' =1,\dots,\ell(n_0), \;\; z\in\bbC\setminus\{0\},  \lb{ALS4.22}
\end{align}
where we introduced the coefficient $\beta(n_0)$. We define 
\begin{equation}
\hat\nu_k(n_0)=(\nu_k(n_0),-(2/c_{0,+})\nu_k(n_0)^{p_-} G_{\ul p}(\nu_k(n_0),n_0)), \quad 
k=1,\dots,\ell(n_0).  \lb{ALS4.23}
\end{equation}
An explicit computation of $\beta(n_0)$ then yields
\begin{align}
\alpha^+(n_0)\beta(n_0) &= -\f{1}{2} \sum_{k=1}^{q(n_0)}
\f{\big(d^{p_k(n_0)-1} \big(\zeta^{-1}y(P)\big)/d \zeta^{p_k(n_0)-1}\big)\big|_{P=(\zeta,\eta)
=\hat\mu_k(n_0)}}{(p_k(n_0)-1)!}  \no \\
& \quad \times \prod_{k'=1, \, k'\neq k}^{q(n_0)}
(\mu_k(n_0)-\mu_{k'}(n_0))^{-p_k(n_0)}  \no \\
& \quad
+\f{1}{2}\bigg(\f{\alpha^+(n_0)}{\alpha(n_0)} + \sum_{k=1}^{q(n_0)}p_k(n_0)\mu_k(n_0) 
- \f{1}{2}\sum_{m=0}^{2p+1}E_m \bigg). \lb{ALS4.24a}
\end{align}
The result \eqref{ALS4.24a} is obtained by inserting the expressions \eqref{ALS4.15}, \eqref{ALS4.16}, and  \eqref{ALS4.22} 
for $F_{\ul p}(\dott,n_0)$, $G_{\ul p}(\dott,n_0)$, and $H_{\ul p}(\dott,n_0)$ into \eqref{ALS4.21}
and collecting all terms of order $z^{2p_+-1}$.  \\
$\textbf{(VI)}$ Introduce
\begin{equation} \lb{ALS4.24}
\beta^+(n_0) = \beta(n_0) \prod_{k=1}^{\ell(n_0)}\nu_k(n_0)^{s_k(n_0)} 
\bigg(\prod_{m=0}^{2p+1} E_m \bigg)^{-1/2}. 
\end{equation}
$\textbf{(VII)}$ 
Using $G_{\ul p}(z,n_0)$, $H_{\ul p}(z,n_0)$, $F_{\ul p}(z,n_0)$, $\beta(n_0)$, $\alpha^+(n_0)$, and $\beta^+(n_0)$,  we next construct the $n_0 \pm 1$ terms from the following equations:
\begin{align} \label{F^-}
F_{\ul p}^- &= \frac{1}{z\gamma} (\alpha^2 H_{\ul p} - 2\alpha G_{\ul p} + F_{\ul p}),\\ \lb{H^-}
H_{\ul p}^- &= \frac{z}{\gamma} (\beta^2 F_{\ul p} - 2\beta G_{\ul p} + H_{\ul p}),\\ \lb{G^-}
G_{\ul p}^- &= \frac{1}{\gamma} ((1+\alpha\beta) G_{\ul p} - \alpha H_{\ul p} - \beta F_{\ul p}),
\end{align}
respectively, 
\begin{align} \label{F^+}
F_{\ul p}^+ &= \frac{1}{z\gamma^+} ((\alpha^+)^2 H_{\ul p} + 2\alpha^+ z G_{\ul p} + z^2 F_{\ul p}),\\ \label{H^+}
H_{\ul p}^+ &= \frac{1}{z\gamma^+} ((\beta^+ z)^2 F_{\ul p} + 2\beta^+ z G_{\ul p} + H_{\ul p}),\\ \label{G^+}
G_{\ul p}^+ &= \frac{1}{z\gamma^+} ((1+\alpha^+\beta^+) z G_{\ul p} + \alpha^+ H_{\ul p} + \beta^+ z^2 F_{\ul p}).
\end{align}
Moreover,  
\begin{equation}
(G_{\ul p}^-)^2 - F_{\ul p}^- H_{\ul p}^-= R_{\ul p}, \quad (G_{\ul p}^+)^2 - F_{\ul p}^+ H_{\ul p}^+=R_{\ul p}.
\end{equation}
Inserting \eqref{ALS4.15}, \eqref{ALS4.16}, and \eqref{ALS4.22} in 
\eqref{F^-}--\eqref{G^-} one verifies
\begin{align} 
F_{\ul p}^-(z,n_0) &\underset{z\to \infty}{=} - c_{0,+}\alpha(n_0) z^{p_+ -1} 
+ \Oh(z^{p_+ -2}), \\  
H_{\ul p}^-(z,n_0) &\underset{z\to \infty}{=} \Oh(z^{p_+}), \\ 
F_{\ul p}^-(z,n_0) &\underset{z\to 0}{=} \Oh(z^{-p_-}), \\ 
H_{\ul p}^-(z,n_0) &\underset{z\to 0}{=} - c_{0,-}\beta(n_0) z^{-p_- +1} 
+ \Oh(z^{-p_- +2}), \\
G_{\ul p}^-(z,n_0) & = \tfrac{1}{2} c_{0,-}z^{-p_-} + \dots + \tfrac{1}{2} c_{0,+}z^{p_+}.
\lb{ALS4.25}
\end{align}
The last equation implies  
\begin{align} \lb{ALS4.26}
G_{\ul p}(z,n_0-1)^2-R_{\ul p}(z) &\underset{z\to \infty}{=} \Oh(z^{2p_+ -1}), \\  
\lb{ALS4.26a}
G_{\ul p}(z,n_0-1)^2-R_{\ul p}(z) &\underset{z\to 0}{=} \Oh(z^{-2p_- +1}),  
\end{align}
so we may write 
\begin{equation} \lb{ALS4.27}
G_{\ul p}(z,n_0-1)^2 - R_{\ul p}(z) = \check F_{s,p_+ -1}(z,n_0-1) 
\check H_{p_- -1,r}(z,n_0-1), \quad z \in \bbC\setminus\{0\}, 
\end{equation}
for some $s \in \{1, \dots, p_-\}$, $r \in \{1, \dots, p_+\}$, where 
\begin{align*}
\check F_{s,p_+ -1}(n_0-1)&= c_{-s}z^{-s} + \dots - c_{0,+}\alpha(n_0) z^{p_+ -1}, \\
\check H_{p_- -1,r}(n_0-1)&= - c_{0,-}\beta(n_0) z^{-p_- +1} + \dots + c_{r} z^r.
\end{align*}
The right-hand side of \eqref{ALS4.27} cannot vanish identically (since otherwise 
$R_{\ul p}(z)=G_{\ul p}(z,n_0-1)^2$ would yield double zeros of $R_{\ul p}(z)$), and hence, 
\begin{align} 
& \check \phi(P,n_0-1)=\f{G_{\ul p}(z,n_0-1)+(c_{0,+}/2)z^{-p_-}y}{\check F_{s,p_+ -1}(z,n_0-1)}
=\f{\check H_{p_- -1,r}(z,n_0-1)}{G_{\ul p}(z,n_0-1)-(c_{0,+}/2)z^{-p_-}y},   \no \\
& \hspace*{8.5cm} P=(z,y)\in\calK_p. \lb{ALS4.28}
\end{align}
Next, we decompose
\begin{align} \lb{ALS4.29}
\check F_{s,p_+ -1}(z,n_0-1) &= - c_{0,+}\alpha(n_0)z^{-s}\prod_{j=1}^{p_+ -1+s} 
(z-\mu_j(n_0-1)), \\
\check H_{p_- -1,r}(z,n_0-1) &= C z^{-p_- +1}\prod_{j=1}^{p_- -1+r}(z-\nu_j(n_0-1)), 
\end{align}
where $C \in \bbC \setminus\{0\}$ and 
$\{\mu_j(n_0-1)\}_{j=1}^{p_+ -1+s}\subset \bbC$, 
$\{\nu_j(n_0-1)\}_{j=1}^{p_- -1+r}\subset \bbC$.
The divisor of $\check \phi(\dott,n_0-1)$ is then given by
\begin{equation} \lb{ALS4.30}
\big(\check \phi(\dott,n_0-1) \big)=\calD_{P_{0,-} 
\hat{\underline{\nu}}(n_0-1)}-\calD_{\Pinfm\hat{\underline{\mu}}(n_0-1)},
\end{equation}
where
\begin{align} \lb{ALS4.31}
\hat{\ul{\mu}}(n_0-1)&=\{\underbrace{P_{0,+},\dots,P_{0,+}}_{p_- -s \text{ times}},
\hat{\mu}_1(n_0-1),\dots,
\hat{\mu}_{p_+ -1+s}(n_0-1)\}, \\  \lb{ALS4.31a}
\hat{\ul{\nu}}(n_0-1)&=\{\hat{\nu}_1(n_0-1),\dots,
\hat{\nu}_{p_- -1+r}(n_0-1),\underbrace{\Pinfp,\dots,\Pinfp}_{p_+ -r 
\text{ times}}\}.  
\end{align}
In particular, 
\begin{align} \lb{ALS4.32}
& \calD_{\hat{\ul{\mu}}(n_0-1)} \, \text{ is an admissible divisor if
and only if } \, s=p_-,  \\ \lb{ALS4.32a}
& \calD_{\hat{\ul{\nu}}(n_0-1)} \, \text{ is an admissible divisor if
and only if } \, r=p_+.
\end{align}
$\textbf{(VIII)}$ Assuming that \eqref{ALS4.26}, \eqref{ALS4.26a} are precisely of order $z^{\pm(2p_\pm -1)}$, that is, assuming $s=p_-$ and $r=p_+$ in 
\eqref{ALS4.28}, we rewrite \eqref{ALS4.28} as
\begin{equation} \lb{ALS4.33}
G_{\ul p}(z,n_0-1)^2-R_{\ul p}(z)=F_{\ul p}(z,n_0-1) H_{\ul p}(z,n_0-1),  
\quad z\in\bbC\setminus\{0\}.   
\end{equation}
By construction, $F_{\ul p}(\dott,n_0-1)$ and $H_{\ul p}(\dott,n_0-1)$ are then of the type
\begin{align}
& F_{\ul p}(z,n_0-1)=- c_{0,+}\alpha(n_0)z^{-p_-}\prod_{k=1}^{q(n_0-1)}(z-\mu_j(n_0-1))^{p_k(n_0-1)},   \no \\
& \sum_{k=1}^{q(n_0-1)} p_k(n_0-1) =p,   \lb{ALS4.34} \\
& \mu_k(n_0-1) \neq \mu_{k'}(n_0-1) \, \text{ for } \, k\neq k', \; 
k, k' =1,\dots,q(n_0-1), \;\; z\in\bbC\setminus\{0\},   \no \\
& H_{\ul p}(z,n_0-1)=c_{0,+}\beta(n_0-1) z^{-p_- +1}\prod_{k=1}^{\ell(n_0-1)} (z-\nu_k(n_0-1))^{s_k(n_0-1)},   \no \\
& \sum_{k=1}^{\ell(n_0-1)} s_k(n_0-1) =p,    \lb{ALS4.34a} \\
& \nu_k(n_0-1) \neq \nu_{k'}(n_0-1) \, \text{ for } \, k\neq k', \; 
k, k' =1,\dots,\ell(n_0-1), \;\; z\in\bbC\setminus\{0\},   \no 
\end{align}
where we introduced the coefficient $\beta(n_0-1)$. We define 
\begin{align}
\begin{split}
\hat\mu_k(n_0-1)=(\mu_k(n_0-1),(2/c_{0,+})\mu_k(n_0-1)^{p_-} 
G_{\ul p}(\mu_k(n_0-1),n_0-1)), \\ \quad 
k=1,\dots,q(n_0-1),  \lb{ALS4.35} \\
\hat\nu_k(n_0-1)=(\nu_k(n_0-1),-(2/c_{0,+})\nu_k(n_0-1)^{p_-} 
G_{\ul p}(\nu_k(n_0-1),n_0-1)), \\ \quad 
k=1,\dots,\ell(n_0-1).  
\end{split}
\end{align}
$\textbf{(IX)}$ At this point one can iterate the procedure step by step to
construct $F_{\ul p}(\dott,n)$, $G_{\ul p}(\dott,n)$, $H_{\ul p}(\dott,n)$, $\alpha(n)$, $\beta(n)$,
$\mu_j(n)$, $\nu_j(n)$, etc., for $n\in (-\infty,n_0]\cap\bbZ$,
subject to the following assumption (cf.\ \eqref{ALS4.32}, \eqref{ALS4.32a}) at each step:
\begin{align}
&\calD_{\humu(n-1)} \, \text{ is an admissible divisor (and hence   
$\alpha(n-1) \neq 0)$}    \lb{ALS4.61} \\
& \quad \text{for all $n\in(-\infty,n_0] \cap \bbZ$},   \no \\   
&\calD_{\hunu(n-1)} \, \text{ is an admissible divisor (and hence   
$\beta(n-1) \neq 0)$}  \lb{ALS4.61a} \\
& \quad \text{for all $n\in(-\infty,n_0] \cap \bbZ$}.   \no 
\end{align}
The formalism is symmetric with respect to $n_0$ and can equally well be
developed for $n\in (-\infty,n_0]\cap\bbZ$ subject to the analogous
assumption
\begin{align} 
&\calD_{\humu(n+1)} \, \text{ is an admissible divisor (and hence 
$\alpha(n+2) \neq 0)$}   \lb{ALS4.62} \\
& \quad \text{for all $n\in[n_0,\infty) \cap \bbZ$,}   \no \\ 
&\calD_{\hunu(n+1)} \, \text{ is an admissible divisor (and hence 
$\beta(n+2) \neq 0)$}   \lb{ALS4.62a} \\
& \quad \text{for all $n\in[n_0,\infty) \cap \bbZ$.}  \no 
\end{align}
$\textbf{(X)}$ Choosing the initial data $\calD_{\humu(n_0)}$ such that 
\begin{equation}
\calD_{\humu(n_0)}\in \calM_0,  \lb{ALS4.64}
\end{equation}
where $\calM_0\subset\symq$ is the set of admissible initial divisors introduced in
Lemma \ref{ALSl4.2}, then guarantees that assumptions 
\eqref{ALS4.61}--\eqref{ALS4.62a} are satisfied for all $n\in\bbZ$. 

Equations \eqref{F^-}--\eqref{G^+} (for arbitrary $n\in\bbZ$) are equivalent to 
$\sAL_{\ul p} (\alpha, \beta)=0$.

At this stage we have verified the basic hypotheses of Section \ref{s3} (i.e., 
 \eqref{ALneq 0,1} and the assumption that $\alpha, \beta$ satisfy the $\ul p$th stationary AL  system \eqref{ALstat}) and hence all results of Section \ref{s3} apply. 

In summary, we proved the following result:

\begin{theorem} \lb{ALSt4.3}
Let $n\in\bbZ$, suppose the set $\{E_m\}_{m=0}^{2p+1} \subset\bbC$ satisfies 
$E_m\neq E_{m'}$ for $m\neq m'$, $m,m'=0,\dots,2p+1$, and introduce the
function $R_{\ul p}$ and the hyperelliptic curve $\calK_p$ as in
\eqref{ALcalK_p}. Choose $\alpha(n_0)\in \bbC \setminus \{ 0\}$, 
$c_{0,+} \in \bbC \setminus \{ 0\}$, and 
a nonspecial divisor $\calD_{\humu(n_0)}\in
\calM_0$, where $\calM_0\subset\symq$ is the set of admissible initial divisors 
introduced in Lemma \ref{ALSl4.2}. Then the stationary $($complex$)$ 
Ablowitz--Ladik algorithm as outlined in steps \textbf{$($I\,$)$}--\textbf{$($X\,$)$}
produces solutions $\alpha, \beta$ of the $\ul p$th stationary Ablowitz--Ladik system,
\begin{align}
\begin{split}
& \sAL_{\ul p}(\alpha, \beta) = \begin{pmatrix} 
- \alpha(g_{p_+,+} + g_{p_-,-}^-) + f_{p_+ -1,+} - f_{p_- -1,-}^-\\  
\beta(g_{p_+,+}^- + g_{p_-,-}) + h_{p_+ -1,+}^- - h_{p_- -1,-}  \end{pmatrix}=0,  \\
& \hspace*{6.85cm}   \ul p=(p_-,p_+)\in\bbN_0^2, 
\end{split}
\end{align}
satisfying \eqref{ALneq 0,1} and 
\begin{align}
& \alpha(n) = \bigg(\prod_{m=0}^{2p+1}E_m\bigg)^{(n-n_0)/2} \calA(n,n_0) \, \alpha(n_0),  
 \lb{ALS4.65} \\ 
& \beta(n) = \Bigg(-\f{1}{2} \sum_{k=1}^{q(n)}
\f{\big(d^{p_k(n)-1} \big(\zeta^{-1}y(P)\big)/d \zeta^{p_k(n)-1}\big)\big|_{P=(\zeta,\eta)
=\hat\mu_k(n)}}{(p_k(n)-1)!}  \no \\
& \hspace*{2.8cm} \times \prod_{k'=1, \, k'\neq k}^{q(n)}
(\mu_k(n)-\mu_{k'}(n))^{-p_k(n)}  \no \\
& \quad +\f{1}{2}\bigg(\bigg(\prod_{m=0}^{2p+1}E_m\bigg)^{1/2} \, 
\prod_{k=1}^{q(n)}\mu_k(n)^{-p_k(n)} + \sum_{k=1}^{q(n)}p_k(n)\mu_k(n) 
- \f{1}{2}\sum_{m=0}^{2p+1}E_m \bigg)\Bigg)   \no \\ 
& \qquad \times \bigg(\prod_{m=0}^{2p+1}E_m\bigg)^{-(n+1-n_0)/2} 
\calA(n+1,n_0)^{-1} \, \alpha(n_0)^{-1},     \lb{ALS4.56}
\end{align}
where
\begin{equation}
\calA(n,n_0)=\begin{cases} \prod_{n'=n_0}^{n-1} 
\prod_{k=1}^{q(n')} \mu_k(n')^{-p_k(n')},  & n\geq n_0 +1, \\
1, & n=n_0, \\
\prod_{n'=n}^{n_0-1} \prod_{k=1}^{q(n')} \mu_k(n')^{p_k(n')},  & n\leq n_0 -1.  
\end{cases}  
\end{equation}
Moreover, Lemmas \ref{lAL3.1}--\ref{lAL3.4} apply.  
\end{theorem}

Finally, we briefly illustrate some aspects of this analysis in the special case
$\ul p=(1,1)$ (i.e., the case where \eqref{ALcalK_p} represents an elliptic
Riemann surface) in more detail. 

\begin{example} \lb{ALSe4.3}  The case $\ul p=(1,1)$. \\ 
In this case one has
\begin{align}
F_{(1,1)}(z,n)&= - c_{0,+} \alpha(n+1) z^{-1} (z - \mu_1(n)), \no \\
G_{(1,1)}(z,n)&= \frac{1}{2}\bigg(\frac{1}{\alpha(n)} - \frac{z}{\alpha(n+1)}\bigg) 
F_{(1,1)}(z,n) + R_{(1,1)}(\hmu_1(n))^{1/2},  \\
R_{(1,1)}(z)&= \bigg(\frac{c_{0,+} \alpha^+}{z}\bigg)^2\prod_{m=0}^3 (z-E_m), \no
\end{align}
and hence a straightforward calculation shows that
\begin{align}
G_{(1,1)}(z,n)^2& - R_{(1,1)}(z) = - c_{0,+}^2 \alpha(n+1) 
\beta(n)z^{-1} (z-\mu_1(n))(z-\nu_1(n))  \no \\ 
&=  -\frac{c_{0,+}^2}{2 z} (z-\mu_1(n))\bigg( \bigg(-\frac{y(\hat\mu_1(n))}{\mu_1(n)}
+\frac{\ti E^{1/2}}{\mu_1(n)}+
\mu_1(n)-\frac{\hatt{E}^+}{2} \bigg) z  \\ 
& \quad  - \frac{\ti E}{\mu_1(n)}
\bigg(-\frac{1}{\ti E^{1/2}}\frac{y(\hat\mu_1(n))}{\mu_1(n)}
+\frac{\mu_1(n)}{\ti E^{1/2}}+
\frac{1}{\mu_1(n)}-\frac{\hatt{E}^-}{2}\bigg)\bigg),   \no 
\end{align}
where
\begin{equation}
\hatt{E}^\pm= \sum_{m=0}^3 E_m^{\pm 1}, \quad \ti{E} = \prod_{m=0}^3 E_m.
\end{equation}
Solving for $\nu_1(n)$ one then obtains
\begin{equation}
\nu_1(n) = \frac{\ti{E}}{\mu_1(n)} \frac{-\frac{y(\hat\mu_1(n))}{\mu_1(n)}
+\frac{\ti E^{1/2}}{\mu_1(n)}+
\mu_1(n)-\frac{\hatt{E}^+}{2}}{-\frac{1}{\ti E^{1/2}}\frac{y(\hat\mu_1(n))}{\mu_1(n)}
+\frac{\mu_1(n)}{\ti E^{1/2}}+\frac{1}{\mu_1(n)}-\frac{\hatt{E}^-}{2}}.
\end{equation}
Thus, $\nu_1(n_0)$ could be $0$ or $\infty$ even if $\mu_1(n_0)\ne 0,\infty$. 
\end{example} 

\section{Properties of Algebro-Geometric Solutions \\ of the  
Time-Dependent Ablowitz--Ladik Hierarchy} \lb{s5}  

In this section we present a quick review of properties of
algebro-geometric solutions of the time-dependent Ablowitz--Ladik hierarchy.
Again we omit all proofs and refer to \cite{GesztesyHolden:2005},  \cite{GesztesyHoldenMichorTeschl:2007a}, and \cite{GesztesyHoldenMichorTeschl:2007} 
for details.

For most of this section we assume the following hypothesis.

\begin{hypothesis} \label{hAL4.1}
$(i)$ Suppose that $\alpha, \beta$ satisfy
\begin{align}
\begin{split}
& \alpha(\dott,t), \beta(\dott,t)\in \bbC^\bbZ,\; t\in\bbR,
\quad 
\alpha(n,\dott), \, \beta(n,\dott) \in C^1(\bbR), \; n\in\bbZ,  \\
& \alpha(n,t)\beta(n,t)\notin\{0,1\}, \; (n,t)\in\bbZ\times\bbR.   \lb{AL4.1A}
\end{split}
\end{align}
$(ii)$  Assume that the hyperelliptic curve $\calK_p$ satisfies 
\eqref{ALcalK_p} and \eqref{ALEneqE}.
\end{hypothesis}

In order to briefly analyze algebro-geometric solutions of the
time-dependent Ablowitz--Ladik hierarchy we proceed as follows. Given 
$\ul p\in\bbN_0^2$,
consider a complex-valued solution $\alpha^{(0)}, \beta^{(0)}$ of the $\ul p$th stationary Ablowitz--Ladik system $\sAL_{\ul p}(a,b)= 0$, associated with $\calK_p$ and a given set of summation constants 
$\{c_{\ell,\pm}\}_{\ell=1,\dots,p_\pm}\subset\bbC$. Next, let
$\ul r\in\bbN_0^2$; we intend to consider solutions $\alpha=\alpha(t_{\ul r}), \beta=\beta(t_{\ul r})$ of the $\ul r$th AL  flow $\AL_{\ul r}(\alpha,\beta)=0$ with $\alpha(t_{0,\ul r})=\alpha^{(0)},
\beta(t_{0,\ul r})=\beta^{(0)}$ for some $t_{0,\ul r}\in\bbR$. To emphasize that the summation constants in the definitions of the stationary and the time-dependent Ablowitz--Ladik equations are independent of each other, we indicate
this by adding a tilde on all the time-dependent quantities. Hence we shall employ the notation $\ti V_{\ul r}$, $\ti F_{\ul r}$, $\ti G_{\ul r}$, $\ti H_{\ul r}$,  $\ti K_{\ul r}$, 
$\tilde f_{s,\pm}$,
$\tilde g_{s,\pm}$, $\tilde h_{s,\pm}$, $\tilde c_{s,\pm}$, in order to distinguish them from $V_{\ul p}$, $F_{\ul p}$, $G_{\ul p}$, $H_{\ul p}$, $K_{\ul p}$, $f_{\ell,\pm}$, 
$g_{\ell,\pm}$, 
$h_{\ell,\pm}$, $c_{\ell,\pm}$, in the following. In addition, we will follow a more elaborate notation inspired by Hirota's $\tau$-function approach and indicate the individual $\ul r$th Ablowitz--Ladik flow by  a separate time variable 
$t_{\ul r} \in \bbR$. 
More precisely, we will review properties of solutions $\alpha, \beta$ of the 
time-dependent algebro-geometric initial value problem 
\begin{align}
\begin{split}
& \ti \AL_{\ul r} (\alpha, \beta) =
\begin{pmatrix}-i\alpha_{t_{\ul r}} 
- \alpha(\tilde g_{r_+,+} + \tilde g_{r_-,-}^-) 
+ \tilde f_{r_+ -1,+} - \tilde f_{r_- -1,-}^-\\
-i\beta_{t_{\ul r}}+ \beta(\tilde g_{r_+,+}^- 
+ \tilde g_{r_-,-}) - \tilde h_{r_- -1,-} + \tilde h_{r_+ -1,+}^-   
  \end{pmatrix} = 0,  \lb{ALal_ivp}  \\
& (\alpha, \beta)\big|_{t_{\ul r}=t_{0,\ul r}} = \big(\alpha^{(0)}, \beta^{(0)}\big) ,  
\end{split} \\
& \sAL_{\ul p} \big(\alpha^{(0)}, \beta^{(0)}\big) =  \begin{pmatrix} 
-\alpha^{(0)}(g_{p_+,+} + g_{p_-,-}^-) + f_{p_+ -1,+} - f_{p_- -1,-}^-\\
  \beta^{(0)}(g_{p_+,+}^- + g_{p_-,-}) - h_{p_- -1,-} + h_{p_+ -1,+}^-       
\end{pmatrix}=0   \lb{AL4.3A}
\end{align}
for some $t_{0,\ul r}\in\bbR$, where $\alpha=\alpha(n,t_{\ul r})$, $\beta=\beta(n,t_{\ul r})$ satisfy \eqref{AL4.1A} and a fixed curve $\calK_p$ is associated
with the stationary solutions $\alpha^{(0)}, \beta^{(0)}$ in \eqref{AL4.3A}. Here,
\begin{equation}
\ul p=(p_-,p_+)\in\bbN_0^2\setminus\{(0,0)\}, \quad \ul r=(r_-,r_+)\in\bbN_0^2, 
\quad p=p_- + p_+ -1. 
\end{equation}
In terms of the zero-curvature formulation this amounts to solving 
\begin{align} 
U_{t_{\ul r}}(z,t_{\ul r}) + U(z,t_{\ul r}) \ti V_{\ul r}(z,t_{\ul r}) - \ti V_{\ul r}^+(z,t_{\ul r}) U(z,t_{\ul r}) &= 0, 
\label{ALzc tilde}  \\
U(z,t_{0,\ul r}) V_{\ul p}(z,t_{0,\ul r}) - V_{\ul p}^+(z,t_{0,\ul r}) U(z,t_{0,\ul r}) &= 0. 
\lb{ALzcstat}
\end{align}
One can show (cf.\ Lemma \ref{ALTl6.2}) that the stationary Ablowitz--Ladik system \eqref{ALzcstat} is actually satisfied for all times $t_{\ul r}\in\bbR$. Thus, we impose
\begin{align} 
U_{t_{\ul r}} + U \ti V_{\ul r} - \ti V_{\ul r}^+ U & = 0, 
\label{ALzctilde}  \\ 
U V_{\ul p} - V_{\ul p}^+ U & = 0,  \lb{ALzctstat}
\end{align}
instead of \eqref{ALzc tilde} and \eqref{ALzcstat}. 
For further reference, we recall the relevant quantities here 
(cf.\ \eqref{AL2.03}, \eqref{AL_v}, \eqref{ALF_p}--\eqref{ALK_p}, \eqref{ALK=G}):
\begin{align}
\begin{split}
U(z) &= \begin{pmatrix}
z & \alpha     \\
z \beta & 1\\
\end{pmatrix},    \\
V_{\ul p}(z) &= i  \begin{pmatrix}
G_{\ul p}^-(z) & - F_{\ul p}^-(z)     \\[1.5mm]
H_{\ul p}^-(z)  & - G_{\ul p}^-(z)  \\
\end{pmatrix}, \quad 
\ti V_{\ul r}(z) = i  \begin{pmatrix}
\ti G_{\ul r}^-(z)  & - \ti F_{\ul r}^-(z)     \\[1.5mm]
\ti H_{\ul r}^-(z) & - \ti K_{\ul r}^-(z)  \\
\end{pmatrix},   \lb{ALv_osv} 
\end{split}
\end{align}
and 
\begin{align}
F_{\ul p}(z) &= \sum_{\ell=1}^{p_-} f_{p_- -\ell,-} z^{-\ell}  
+ \sum_{\ell=0}^{p_+ -1} f_{p_+ -1-\ell,+} z^\ell 
=- c_{0,+}\alpha^+ z^{-p_-}\prod_{j=1}^{p}(z-\mu_j),  \no \\ \no
G_{\ul p}(z) &= \sum_{\ell=1}^{p_-} g_{p_- -\ell,-} z^{-\ell} 
+ \sum_{\ell=0}^{p_+} g_{p_+ -\ell,+} z^\ell,  \\ \no
H_{\ul p}(z) &= \sum_{\ell=0}^{p_- -1} h_{p_- -1-\ell,-} z^{-\ell} 
+ \sum_{\ell=1}^{p_+} h_{p_+ -\ell,+} z^\ell 
= c_{0,+}\beta z^{-p_- +1}\prod_{j=1}^{p}(z-\nu_j), \\
\ti F_{\ul r}(z) &= \sum_{s=1}^{r_-} \tilde f_{r_- -s,-} z^{-s}  
+ \sum_{s=0}^{r_+ -1} \tilde f_{r_+ -1-s,+} z^s,   \lb{AL4.9}  \\ \no
\ti G_{\ul r}(z) &= \sum_{s=1}^{r_-} \tilde g_{r_- -s,-} z^{-s}  
+ \sum_{s=0}^{r_+} \tilde g_{r_+ -s,+} z^s,\\ \no
\ti H_{\ul r}(z) &= \sum_{s=0}^{r_- -1} \tilde h_{r_- -1-s,-} z^{-s}  
+ \sum_{s=1}^{r_+} \tilde h_{r_+ -s,+} z^s, \\
\ti K_{\ul r}(z) &= \sum_{s=0}^{r_-} \tilde g_{r_- -s,-} z^{-s}  
+  \sum_{s=1}^{r_+} \tilde g_{r_+ -s,+} z^s 
= \ti G_{\ul r}(z)+\tilde g_{r_-,-}-\tilde g_{r_+,+}   \no 
\end{align}
for fixed $\ul p\in\bbN_0^2\setminus\{(0,0)\}$, $\ul r\in\bbN_0^2$. Here 
$f_{\ell,\pm}$, $\tilde f_{s,\pm}$, $g_{\ell,\pm}$, 
$\tilde g_{s,\pm}$, $h_{\ell,\pm}$, and $\tilde h_{s,\pm}$ are defined as in 
\eqref{AL0+}--\eqref{ALh_l-} with appropriate sets of summation constants 
$c_{\ell,\pm}$, $\ell\in\bbN_0$, and $\tilde c_{k,\pm}$, $k\in\bbN_0$. Explicitly, 
\eqref{ALzctilde} and \eqref{ALzctstat} are equivalent to 
(cf.\ \eqref{AL1,1}--\eqref{AL2,1},  \eqref{ALalphat}--\eqref{AL2,2r})
\begin{align}  \label{ALalpha_t}
\alpha_{t_{\ul r}} &= i \big(z \ti F_{\ul r}^- + \alpha (\ti G_{\ul r} + \ti{K}_{\ul r}^-)
- \ti F_{\ul r}\big),\\ \label{ALbeta_t}
\beta_{t_{\ul r}} &= - i \big(\beta (\ti G_{\ul r}^- + \ti{K}_{\ul r}) - \ti H_{\ul r} 
+ z^{-1} \ti H_{\ul r}^-\big), \\ \label{AL1,1 r}
0 &= z (\ti G_{\ul r}^- - \ti G_{\ul r}) + z\beta \ti F_{\ul r} + \alpha \ti H_{\ul r}^-,\\ \label{AL2,2 r}
0 &= z \beta \ti F_{\ul r}^- + \alpha \ti H_{\ul r} + \ti{K}_{\ul r}^- - \ti{K}_{\ul r}, \\
0 &= z (G_{\ul p}^- - G_{\ul p}) + z \beta F_{\ul p} + \alpha H_{\ul p}^-,  \lb{AL11} \\  
0 &=z \beta F_{\ul p}^- + \alpha H_{\ul p} - G_{\ul p} + G_{\ul p}^-,  \lb{AL12} \\
0 &= - F_{\ul p} + z F_{\ul p}^- + \alpha (G_{\ul p} + G_{\ul p}^-), \lb{AL21}  \\  
0 &= z \beta (G_{\ul p} + G_{\ul p}^-) - z H_{\ul p} + H_{\ul p}^-,  \lb{AL22}
\end{align}
respectively. In particular, \eqref{ALR} holds in the present $t_{\ul r}$-dependent setting, that is,  
\begin{equation} \label{ALR_t}
G_{\ul p}^2 - F_{\ul p} H_{\ul p} = R_{\ul p}.
\end{equation}

As in the stationary context \eqref{ALhmu}, \eqref{ALhnu} we introduce
\begin{align}
\begin{split}
\hat \mu_j(n,t_{\ul r})&=(\mu_j(n,t_{\ul r}), (2/c_{0,+}) \mu_j(n,t_{\ul r})^{p_-} G_{\ul p}(\mu_j(n,t_{\ul r}),n,t_{\ul r}))\in\calK_p, \\ 
& \hspace*{4.4cm} j=1, \dots, p, \; (n,t_{\ul r})\in\bbZ\times\bbR,   \lb{AL4.20}
\end{split}
\end{align}
and
\begin{align}
\begin{split}
\hat \nu_j(n,t_{\ul r})&=(\nu_j(n,t_{\ul r}), - (2/c_{0,+}) \nu_j(n,t_{\ul r})^{p_-} G_{\ul p}(\nu_j(n,t_{\ul r}),n,t_{\ul r}))\in\calK_p, \\
& \hspace*{4.6cm}  j=1, \dots, p, \; (n,t_{\ul r})\in\bbZ\times\bbR,   \lb{AL4.21}
\end{split}
\end{align}
and note that the regularity assumptions \eqref{AL4.1A} on $\alpha, \beta$ imply
continuity of $\mu_j$ and $\nu_k$ with respect to $t_{\ul r}\in\bbR$ (away from collisions of these zeros, $\mu_j$ and $\nu_k$ are of course $C^\infty$).  

In analogy to \eqref{ALphi}, \eqref{ALphi1}, one defines the following
meromorphic function $\phi (\dott,n,t_{\ul r})$ on $\calK_p$,
\begin{align}
\phi(P,n,t_{\ul r}) & = \frac{(c_{0,+}/2) z^{-p_-} y + G_{\ul p}(z,n,t_{\ul r})}{F_{\ul p}(z,n,t_{\ul r})}    \lb{AL4.22}  \\
& = \frac{-H_{\ul p}(z,n,t_{\ul r})}{(c_{0,+}/2) z^{-p_-} y - G_{\ul p}(z,n,t_{\ul r})},    \lb{AL4.23}  \\ 
&  \hspace*{-.05cm}  P=(z,y)\in\calK_p, \; (n,t_{\ul r})\in\bbZ\times\bbR,  \no
\end{align}
with divisor $(\phi(\dott,n,t_{\ul r}))$ of $\phi(\dott,n,t_{\ul r})$ given by 
\begin{equation} 
(\phi(\dott,n,t_{\ul r})) = \calD_{P_{0,-} \hunu(n,t_{\ul r})} - \calD_{\Pinfm \humu(n,t_{\ul r})}.  
\end{equation}
The time-dependent Baker--Akhiezer vector is then defined in terms of $\phi$ by 
\begin{align}
\Psi(P,n,n_0,t_{\ul r},t_{0,\ul r}) &= \begin{pmatrix}\psi_1(P,n,n_0,t_{\ul r},t_{0,\ul r})\\
\psi_2(P,n,n_0,t_{\ul r},t_{0,\ul r})\end{pmatrix}, \\ 
\psi_1(P,n,n_0,t_{\ul r},t_{0,\ul r}) &= \exp \Big(i \int_{t_{0,\ul r}}^{t_{\ul r}} ds 
\big(\ti G_{\ul r} (z,n_0,s) - \ti F_{\ul r}(z,n_0,s) \phi(P,n_0,s)\big)\Big) \\ \no
&\quad \times 
\begin{cases}
\prod_{n'=n_0+1}^{n} \big(z + \alpha(n',t_{\ul r}) \phi^-(P,n',t_{\ul r})\big), & n \geq n_0 +1,\\
1, & n=n_0, \\
\prod_{n'=n+1}^{n_0} \big(z + \alpha(n',t_{\ul r}) \phi^-(P,n',t_{\ul r})\big)^{-1}, & n \leq n_0 -1, 
\end{cases} 
\\
\psi_2(P,n,n_0,t_{\ul r},t_{0,\ul r}) &= \exp \Big(i \int_{t_{0,\ul r}}^{t_{\ul r}} ds 
\big(\ti G_{\ul r} (z,n_0,s) - \ti F_{\ul r}(z,n_0,s)  \phi(P,n_0,s)\big)\Big) \\ \no
&\hspace*{-2.5cm} \times \phi(P,n_0,t_{\ul r})
\begin{cases}
\prod_{n'=n_0+1}^{n} \big(z \beta(n',t_{\ul r})\phi^-(P,n',t_{\ul r})^{-1} + 1\big), 
& n \geq n_0 +1,\\
1, & n=n_0, \\
\prod_{n'=n+1}^{n_0} \big(z \beta(n',t_{\ul r})\phi^-(P,n',t_{\ul r})^{-1} + 1\big)^{-1}, 
& n \leq n_0 -1, 
\end{cases} \no \\ \no
&\hspace*{-.4cm} P=(z,y)\in\calK_p\setminus\{\Pinfp,\Pinfm,\Pzp,\Pzm\}, 
\; (n,t_{\ul r})\in\bbZ\times\bbR.
\end{align}
One observes that
\begin{align} 
\begin{split}
& \psi_1(P,n,n_0,t_{\ul r},\tilde t_{\ul r}) = \psi_1(P,n_0,n_0,t_{\ul r},\tilde t_{\ul r})\psi_1(P,n,n_0,t_{\ul r},t_{\ul r}), \label{ALpsi t n_0}  \\
& P=(z,y)\in\calK_p\setminus
\{\Pinfp,\Pinfm,\Pzp,\Pzm\}, \; (n,n_0,t_{\ul r},\tilde t_{\ul r})\in\bbZ^2\times\bbR^2.
\end{split}
\end{align}

The following lemma records basic properties of $\phi$ and $\Psi$ in analogy to the stationary case discussed in Lemma \ref{lAL3.1}. 

\begin{lemma} [\cite{GesztesyHoldenMichorTeschl:2007a}]  \lb{lAL4.2}
Assume Hypothesis \ref{hAL4.1} and suppose that 
\eqref{ALzctilde}, \eqref{ALzctstat} hold. In addition, let
$P=(z,y)\in\calK_p\setminus\{\Pinfp, \Pinfm\}$, 
$(n,n_0,t_{\ul r},t_{0,\ul r})\in\bbZ^2\times\bbR^2$. Then $\phi$ satisfies 
\begin{align} \label{ALriccati_time} 
& \alpha \phi(P)\phi^-(P) - \phi^-(P) + z \phi(P) = z \beta, \\
& \phi_{t_{\ul r}}(P) =i \ti F_{\ul r} \phi^2(P) - 
i \big(\ti G_{\ul r}(z) + \ti K_{\ul r}(z)\big)\phi(P) +  i \ti H_{\ul r}(z),    \label{ALphi_t}  \\
  \label{ALphi 1t}
& \phi(P) \phi(P^*) = \frac{H_{\ul p}(z)}{F_{\ul p}(z)},\\ \label{ALphi 2t}
& \phi(P) + \phi(P^*) = 2\frac{G_{\ul p}(z)}{F_{\ul p}(z)},\\ \label{ALphi 3t}
& \phi(P) - \phi(P^*) = c_{0,+} z^{-p_-} \frac{y(P)}{F_{\ul p}(z)}.
\end{align}
Moreover, assuming  
$P=(z,y)\in\calK_p\setminus\{\Pinfp, \Pinfm,\Pzp,\Pzm\}$, then 
$\Psi$ satisfies
\begin{align}  \label{ALpsi t 1}
& \psi_2(P,n,n_0,t_{\ul r},t_{0,\ul r})= \phi(P,n,t_{\ul r}) \psi_1(P,n,n_0,t_{\ul r},t_{0,\ul r}),\\ 
\label{ALpsi 2t}
& U(z) \Psi^-(P)=\Psi(P),\\ \label{ALpsi 3t}
& V_{\ul p}(z)\Psi^-(P)= - (i/2) c_{0,+} z^{-p_-} y \Psi^-(P), \\
& \Psi_{t_{\ul r}}(P) = \ti V_{\ul r}^+(z) \Psi(P),   
\lb{ALtime_AL} \\
& \psi_1(P,n,n_0,t_{\ul r},t_{0,\ul r}) \psi_1(P^*,n,n_0,t_{\ul r},t_{0,\ul r}) = z^{n-n_0} \frac{F_{\ul p}(z,n,t_{\ul r})}{F_{\ul p}(z,n_0,t_{0,\ul r})}  \pgam(n,n_0,t_{\ul r}), \label{ALpsi 4t}\\ 
& \psi_2(P,n,n_0,t_{\ul r},t_{0,\ul r}) \psi_2(P^*,n,n_0,t_{\ul r},t_{0,\ul r}) 
= z^{n-n_0} \frac{H_{\ul p}(z,n,t_{\ul r})}{F_{\ul p}(z,n_0,t_{0,\ul r})}  \pgam(n,n_0,t_{\ul r}),
\label{ALpsi 5t} \\
& \psi_1(P,n,n_0,t_{\ul r},t_{0,\ul r}) \psi_2(P^*,n,n_0,t_{\ul r},t_{0,\ul r}) 
+\psi_1(P^*,n,n_0,t_{\ul r},t_{0,\ul r}) \psi_2(P,n,n_0,t_{\ul r},t_{0,\ul r}) \no \\
&\quad =2 z^{n-n_0} \frac{G_{\ul p}(z,n,t_{\ul r})}{F_{\ul p}(z,n_0,t_{0,\ul r})}
\pgam(n,n_0,t_{\ul r}), \label{ALpsi 6t} \\
& \psi_1(P,n,n_0,t_{\ul r},t_{0,\ul r}) \psi_2(P^*,n,n_0,t_{\ul r},t_{0,\ul r}) 
-\psi_1(P^*,n,n_0,t_{\ul r},t_{0,\ul r}) \psi_2(P,n,n_0,t_{\ul r},t_{0,\ul r}) \no\\
& \quad =-c_{0,+} z^{n-n_0-p_-} \frac{y}{F_{\ul p}(z,n_0,t_{0,\ul r})} \pgam(n,n_0,t_{\ul r}), 
\label{ALpsi 7t}
\end{align}
where  
\begin{equation} \lb{ALTpgam}
\pgam(n,n_0,t_{\ul r}) = \begin{cases}      
\prod_{n'=n_0 + 1}^n \gamma(n',t_{\ul r}) & n \geq n_0 +1, \\
1                      &  n=n_0, \\
\prod_{n'=n + 1}^{n_0} \gamma(n',t_{\ul r})^{-1}  & n \leq n_0 -1.
\end{cases}
\end{equation}
In addition, as long as the zeros $\mu_j(n_0,s)$ of 
$(\dott)^{p_-}F_{\ul p} (\dott,n_0,s)$ are all  simple and distinct from zero for 
$s \in\calI_{\mu}$, $\calI_{\mu}\subseteq\bbR$ an open interval, 
$\Psi(\dott,n,n_0,t_{\ul r},t_{0,\ul r})$ is meromorphic on $\calK_p\setminus \{\Pinfp,\Pinfm,$ $\Pzp,\Pzm\}$ for $(n,t_{\ul r},t_{0,\ul r})\in\bbZ\times\calI_{\mu}^2$.
\end{lemma}

The trace formulas recorded in Lemma \ref{lAL3.2} extend to the present 
time-dependent context without any change as $t_{\ul r}\in\bbR$ can be viewed as a fixed parameter. Further details are thus omitted. 

For completeness we next mention the Dubrovin-type equations for the time variation of the zeros $\mu_j$ of $(\dott)^{p_-}F_{\ul p}$ and $\nu_j$ of 
$(\dott)^{p_- -1}H_{\ul p}$ governed by the $\ti{\AL}_{\ul r}$ flow. 

\begin{lemma} [\cite{GesztesyHoldenMichorTeschl:2007a}]  \lb{lAL4.5}
Assume Hypothesis \ref{hAL4.1} and suppose that \eqref{ALzctilde}, 
\eqref{ALzctstat} hold on $\bbZ\times \calI_{\mu}$ with $\calI_{\mu} \subseteq\bbR$ an open interval. In addition, assume that the
zeros $\mu_j$, $j=1,\dots,p$, of $(\dott)^{p_-}F_{\ul p}(\dott)$ remain distinct and nonzero on 
$\bbZ\times \calI_{\mu}$. Then $\{\hmu_j\}_{j=1,\dots,p}$, defined in \eqref{AL4.20}, satisfies the following first-order system of differential equations on 
$\bbZ\times \calI_{\mu}$, 
\begin{equation}
\mu_{j,t_{\ul r}} = - i \ti F_{\ul r}(\mu_j)  y(\hmu_j) (\alpha^+)^{-1} 
\prod_{\substack{k=1\\k\neq j}}^{p} (\mu_j-\mu_k)^{-1},  \quad 
j=1,\dots,p,   \lb{AL4.69a}
\end{equation}
with 
\begin{equation}
\hmu_j(n,\cdot)\in C^\infty(\calI_\mu,\calK_p), \quad j=1,\dots,p, 
\; n\in\bbZ.   \lb{AL4.69b}
\end{equation} 
For the zeros $\nu_j$, $j=1,\dots,p$, of $(\dott)^{p_- -1}H_{\ul p}(\dott)$, identical statements hold with 
$\mu_j$ and $\calI_{\mu}$ replaced by $\nu_j$ and $\calI_{\nu}$, etc.\ $($with 
$\calI_{\nu} \subseteq\bbR$ an open interval\,$)$. In particular, 
$\{\hat \nu_j\}_{j=1,\dots,p}$, defined in \eqref{AL4.21}, satisfies the first-order 
system on $\bbZ\times \calI_{\nu}$, 
\begin{equation}
\nu_{j,t_{\ul r}} = i  \ti H_{\ul r}(\nu_j) y(\hnu_j) ( \beta \nu_j)^{-1} 
\prod_{\substack{k=1\\k\neq j}}^{p} (\nu_j-\nu_k)^{-1},  \quad 
j=1,\dots,p,   \lb{AL4.71a}
\end{equation}
with 
\begin{equation}
\hnu_j(n,\cdot)\in C^\infty(\calI_\nu,\calK_p), \quad j=1,\dots,p, 
\; n\in\bbZ.   \lb{AL4.71b}
\end{equation} 
\end{lemma}

When attempting to solve the Dubrovin-type systems \eqref{AL4.69a} and 
\eqref{AL4.71a}, they must be augmented with appropriate divisors 
$\calD_{\humu(n_0,t_{0,\ul r})}\in\sym^{p} \calK_p$, $t_{0,\ul r}\in \calI_\mu$, and 
$\calD_{\hunu(n_0,t_{0,\ul r})}\in\sym^{p} \calK_p$, $t_{0,\ul r}\in \calI_\nu$,  
as initial conditions. 

For the $t_{\ul r}$-dependence of $F_{\ul p}$, $G_{\ul p}$, and $H_{\ul p}$ one obtains the following result.

\begin{lemma} [\cite{GesztesyHoldenMichorTeschl:2007a}]  \lb{lAL4.3}
Assume Hypothesis \ref{hAL4.1} and suppose that 
\eqref{ALzctilde}, \eqref{ALzctstat} hold. In addition, let
$(z,n,t_{\ul r})\in\bbC\times\bbZ\times\bbR$. Then, 
\begin{align} \label{ALF_t}
F_{\ul p,t_{\ul r}} &= - 2 i G_{\ul p} \ti F_{\ul r} 
+ i \big(\ti G_{\ul r} + \ti K_{\ul r}\big)F_{\ul p}, \\ \label{ALG_t}
G_{\ul p,t_{\ul r}} &= i F_{\ul p} \ti H_{\ul r} - i H_{\ul p} \ti F_{\ul r}, \\  \label{ALH_t}
H_{\ul p,t_{\ul r}} &= 2 i G_{\ul p} \ti H_{\ul r} - i \big(\ti G_{\ul r} + \ti K_{\ul r}\big)H_{\ul p}.
\end{align}
In particular, \eqref{ALF_t}--\eqref{ALH_t} are equivalent to 
\begin{equation} \label{ALV_t}
V_{\ul p,t_{\ul r}} = \big[\ti V_{\ul r}, V_{\ul p}\big].
\end{equation}
\end{lemma}

It will be shown in Section \ref{ALTs6} that Lemma \ref{lAL4.3} yields a
first-order system of differential equations for $f_{\ell,\pm}$, $g_{\ell,\pm}$, and 
$h_{\ell,\pm}$, that serves as a pertinent substitute for the Dubrovin
equations \eqref{AL4.69a} even (in fact, especially) when some of the $\mu_j$
coincide.

Lemma \ref{lAL3.4} on nonspecial divisors and the linearization property of the Abel map extend to the present time-dependent setting. For this fact we need to introduce a particular differential of the second kind, $\ti \Omega_{\ul r}^{(2)}$, defined as follows. Let  $\omega_{\Pinfpm,q}^{(2)}$ and $\omega_{\Pzpm,q}^{(2)}$ be the normalized differentials of the second kind with a unique pole at $\Pinfpm$ and $\Pzpm$, respectively, and principal parts  
\begin{align}
\omega_{\Pinfpm,q}^{(2)}&\underset{\zeta\to 0}{=}
\big(\zeta^{-2-q}+\Oh(1) \big)d\zeta, \quad P\to\Pinfpm, \; \zeta=1/z, \; 
q\in\bbN_0, \label{AL4.69}  \\
\omega_{\Pzpm,q}^{(2)}&\underset{\zeta\to 0}{=}
\big(\zeta^{-2-q}+\Oh(1) \big)d\zeta, \quad P\to\Pzpm, \; \zeta=z, \; 
q\in\bbN_0, \label{AL4.69A}
\end{align}
with vanishing $a$-periods,
\begin{equation}
\int_{a_j}\omega_{\Pinfpm,q}^{(2)} = \int_{a_j}\omega_{\Pzpm,q}^{(2)} =0, 
\quad j=1,\dots,p.
\label{AL4.70}
\end{equation}
Moreover, we define
\begin{align} \no
\ti \Omega_{\ul r}^{(2)}&=\frac{i}{2}\bigg(\sum_{s=1}^{r_-} s \tilde c_{r_--s, -}
\big(\omega_{P_{0,+},s-1}^{(2)}-\omega_{P_{0,-},s-1}^{(2)} \big) \\
&\hspace*{1cm} + \sum_{s=1}^{r_+} s \tilde c_{r_+-s, +}
\big(\omega_{\Pinfp,s-1}^{(2)}-\omega_{\Pinfm,s-1}^{(2)} \big)\bigg),  \label{AL4.71}
\end{align}
where $\tilde c_{\ell,\pm}$ are the summation constants in $\ti F_{\ul r}$. 
The corresponding vector of $b$-periods of $\ti\Omega_{\ul r}^{(2)}/(2\pi i)$ is 
then denoted by 
\begin{equation}
\ti {\underline U}_{\ul r}^{(2)}=\big(\ti U_{\ul r,1}^{(2)},\dots,\ti U_{\ul r,p}^{(2)}
\big), \quad \ti U_{\ul r,j}^{(2)}
=\frac{1}{2\pi i}\int_{b_j}\ti \Omega_{\ul r}^{(2)}, \quad j=1,\dots,p.
\label{AL4.74}
\end{equation}

The time-dependent analog of Lemma \ref{lAL3.4} then reads as follows. 

\begin{lemma} [\cite{GesztesyHoldenMichorTeschl:2007a}, 
\cite{GesztesyHoldenMichorTeschl:2007}] \label{lAL4.6a} 
Assume Hypothesis \ref{hAL4.1} and suppose that \eqref{ALzctilde}, \eqref{ALzctstat} hold. Moreover, let $(n,t_{\ul r})\in\bbZ\times\bbR$. Denote by $\calD_{\humu}$,
$\humu=\{\hmu_1,\dots,\hmu_{p}\}$ and $\calD_{\hunu}$,
$\hunu=\{\hnu_1,\dots,\hnu_{p}\}$, the pole and zero divisors of degree
$p$, respectively, associated with $\alpha$, $\beta$, and $\phi$ defined
according to \eqref{AL4.20} and \eqref{AL4.21}, that is,
\begin{align}
\hat\mu_j (n,t_{\ul r}) &= (\mu_j (n,t_{\ul r}),(2/c_{0,+}) \mu_j(n,t_{\ul r})^{p_-} 
G_{\ul p}(\mu_j(n,t_{\ul r}),n,t_{\ul r})), 
\quad  j=1,\dots,p,  \\
\hat\nu_j (n,t_{\ul r}) &= (\nu_j (n,t_{\ul r}),- (2/c_{0,+}) \nu_j(n,t_{\ul r})^{p_-} 
G_{\ul p}(\nu_j(n,t_{\ul r}),n,t_{\ul r})), 
\quad j=1,\dots,p. 
\end{align}
Then $\calD_{\humu(n,t_{\ul r})}$ and $\calD_{\hunu(n,t_{\ul r})}$ are nonspecial for all
$(n,t_{\ul r})\in\bbZ\times\bbR$. Moreover, the Abel map linearizes the auxiliary divisors $\calD_{\humu(n,t_{\ul r})}$ and $\calD_{\hunu(n,t_{\ul r})}$ in the sense that
\begin{align}
\amap(\calD_{\humu(n,t_{\ul r})})
&=\amap(\calD_{\humu(n_0,t_{0,\ul r})})+ (n-n_0) \underline A_{\Pzm}(\Pinfp)
+ (t_{\ul r}-t_{0,\ul r}) \ti{\underline U}^{(2)}_{\ul r}, \label{AL4.82} \\
\amap(\calD_{\hunu(n,t_{\ul r})})
&=\amap(\calD_{\hunu(n_0,t_{0,\ul r})})
+ (n-n_0) \underline A_{\Pzm}(\Pinfp)+ (t_{\ul r}-t_{0,\ul r}) \ti{\underline U}^{(2)}_{\ul r},
\label{AL4.83}
\end{align}
where $Q_0 \in \calK_p$ is a given base point and $\ti{\underline U}^{(2)}_{\ul r}$
is the vector of $b$-periods introduced in \eqref{AL4.74}.
\end{lemma}

Again we refer to \cite{GesztesyHoldenMichorTeschl:2007} (and the references cited therein) for theta function representations of $\alpha$ and $\beta$. 
These theta function representations also show that $\gamma(n,t_{\ul r})\notin \{0,1\}$ 
for all $(n,t_{\ul r})\in\bbZ\times\bbR$, that is, the last condition in \eqref{AL4.1A} is 
satisfied for the time-dependent algebro-geometric AL solutions discussed in 
this section provided the associated Dirichlet and Neumann divisors are 
admissible.

\section{The Algebro-Geometric Ablowitz--Ladik Hierarchy\\ Initial Value 
Problem} \lb{ALTs6}

In this section we consider the algebro-geometric Ablowitz--Ladik hierarchy initial value
problem \eqref{ALal_ivp}, \eqref{AL4.3A} with complex-valued initial data. For a
generic set of initial data we will prove unique solvability of the
initial value problem globally in time. 

As mentioned in the introduction, the strategy of basing the solution of the 
algebro-geometric initial value problem on the Dubrovin-type equations \eqref{AL4.69a}, \eqref{AL4.71a}, and the trace formulas for $\alpha$ and $\beta$, meets with serious difficulties as the Dirichlet eigenvalues $\hat \mu_j$ may collide on $\calK_p$ and hence the denominator of \eqref{AL4.69a} can blow up. Hence, we will develop an alternative strategy based on the use of elementary symmetric functions of the variables
$\{\mu_j\}_{j=1,\dots,p}$ in this section, which can accommodate
collisions of $\hat\mu_j$. In short, our strategy will consist of the following:

$(i)$ Replace the first-order autonomous Dubrovin-type system \eqref{AL4.69a}
of differential equations in $t_{\ul r}$ for the Dirichlet eigenvalues
$\mu_j(n,t_{\ul r})$, $j=1,\dots,p$, augmented by appropriate initial
conditions, by the first-order autonomous system 
\eqref{ALT6.61}, \eqref{ALT6.62} for the coefficients $f_{\ell,\pm}$, $h_{\ell,\pm}$,
$\ell=1,\dots, p_\pm -1$, and $g_{\ell,\pm}$, $\ell=1,\dots, p_\pm$, with respect to 
$t_{\ul r}$. Solve this first-order autonomous system 
in some time interval $(t_{0,\ul r}-T_0,t_{0,\ul r}+T_0)$ under
appropriate initial conditions at $(n_0,t_{0,\ul r})$ derived from an
initial (nonspecial) Dirichlet divisor
$\calD_{\humu(n_0,t_{0,\ul r})}$.

$(ii)$ Use the stationary algorithm derived in Section \ref{ALSs4} to
extend the solution of step $(i)$ from
$\{n_0\}\times(t_{0,\ul r}-T_0,t_{0,\ul r}+T_0)$ to 
$\bbZ \times (t_{0,\ul r}-T_0,t_{0,\ul r}+T_0)$ (cf.\ Lemma \ref{ALTl6.2}).

$(iii)$ Prove consistency of this approach, that is, show that the
discrete algorithm of Section \ref{ALSs4} is compatible with the
time-dependent Lax and zero-curvature equations in the sense that first
solving the autonomous system \eqref{ALT6.61}, \eqref{ALT6.62} and then
applying the discrete algorithm, or first applying the discrete
algorithm and then solving the autonomous system \eqref{ALT6.61}, \eqref{ALT6.62} 
yields the same result whenever the same endpoint 
$(n,t_{\ul r})$ is reached (cf.\ Lemma \ref{ALTl6.3} and Theorem \ref{ALTt6.4}). 

$(iv)$ Prove that there is a dense set of initial conditions of full
measure for which this strategy yields global solutions of the
algebro-geometric Ablowitz--Ladik hierarchy initial value problem.

\smallskip

To set up this formalism we need some preparations.
From the outset we make the following assumption.  

\begin{hypothesis} \lb{ALTh6.1} 
Suppose that 
\begin{equation}
\alpha, \beta \in \bbC^{\bbZ} \, \text{ and } \, \alpha(n)\beta(n) \notin \{0,1\} 
\, \text{ for all } \,  n\in\bbZ,   \lb{ALT6.0}
\end{equation}
and assume that $\alpha, \beta$ satisfy the $\ul p$th stationary Ablowitz--Ladik equation
\eqref{ALstat}. In addition, suppose that the affine part of the hyperelliptic curve 
$\calK_p$ 
in \eqref{ALcalK_p} is nonsingular. 
\end{hypothesis} 

We introduce a deformation (time)
parameter $t_{\ul r}\in\bbR$ in $\alpha=\alpha(t_{\ul r})$ and $\beta=\beta(t_{\ul r})$ and hence obtain $t_{\ul r}$-dependent quantities $f_\ell=f_\ell(t_{\ul r})$, $g_\ell=g_\ell(t_{\ul r})$,
$F_{\ul p}(z)=F_{\ul p}(z,t_{\ul r})$, $G_{\ul p}(z)=G_{\ul p}(z,t_{\ul r})$, etc. At a fixed
initial time $t_{0,\ul r}\in\bbR$ we require that
\begin{equation}
(\alpha,\beta)|_{t_{\ul r}=t_{0,\ul r}} = \big(\alpha^{(0)},\beta^{(0)}\big),   \lb{ALT6.16}
\end{equation}
where $\alpha^{(0)}=\alpha(\dott,t_{0,\ul r}),\beta^{(0)}=\beta(\dott,t_{0,\ul r})$ satisfy the
$\ul p$th stationary Ablowitz--Ladik system \eqref{ALstat}. 
As discussed in Section \ref{ALSs4}, in order to
guarantee that the stationary solutions \eqref{ALT6.16} can be constructed
for all $n\in\bbZ$ one starts  from a particular divisor 
\begin{equation}
\calD_{\humu(n_0,t_{0,\ul r})}\in\calM_0,   \lb{ALT6.17}
\end{equation}
where $\humu(n_0,t_{0,\ul r})$ is of the form
\begin{align} \lb{ALT6.18}
&\humu(n_0,t_{0,\ul r})   \\
& \quad =\{\underbrace{\hat\mu_1(n_0,t_{0,\ul r}),\dots,
\hat\mu_1(n_0,t_{0,\ul r})}_{p_1(n_0,t_{0,\ul r}) \text{ times}},
\dots,\underbrace{\hat\mu_{q(n_0,t_{0,\ul r})}(n_0,t_{0,\ul r}),\dots,
\hat\mu_{q(n_0,t_{0,\ul r})}(n_0,t_{0,\ul r})}_{p_{q(n_0,t_{0,\ul r})}(n_0,t_{0,\ul r})
\text{ times}}\}. \no
\end{align}
Moreover, as in Section \ref{ALSs4} we prescribe the data 
\begin{equation}
\alpha(n_0,t_{0,\ul r}) \in\bbC\setminus\{0\} \, \text{ and } \, c_{0,+}\in\bbC\setminus\{0\},  
\lb{ALT6.18A}
\end{equation}
and of course the hyperelliptic curve $\calK_p$ with nonsingular affine part (cf.\ \eqref{ALS4.8}). In addition, we introduce
\begin{align} 
\alpha^+(n_0,t_{0,\ul r}) &= \alpha(n_0,t_{0,\ul r}) \bigg(\prod_{m=0}^{2p+1}E_m\bigg)^{1/2} \,  
\prod_{k=1}^{q(n_0,t_{0,\ul r})}\mu_k(n_0,t_{0,\ul r})^{-p_k(n_0,t_{0,\ul r})},   \lb{ALT6.18B} \\  
F_{\ul p}(z,n_0,t_{0,\ul r})&=\sum_{\ell=1}^{p_-} f_{p_--\ell,-}(n_0,t_{0,\ul r})z^{-\ell} 
+ \sum_{\ell=0}^{p_+ -1} f_{p_+ -1-\ell,+}(n_0,t_{0,\ul r})z^\ell   \no \\
&= - c_{0,+}\alpha^+(n_0,t_{0,\ul r})z^{-p_-}
\prod_{k=1}^{q(n_0,t_{0,\ul r})} (z-\mu_k(n_0,t_{0,\ul r}))^{p_k(n_0,t_{0,\ul r})},  \lb{ALT6.18C} \\
G_{\ul p}(z,n_0,t_{0,\ul r}) &= \frac{1}{2} \bigg(\frac{1}{\alpha(n_0,t_{0,\ul r})}-\frac{z}{\alpha^+(n_0,t_{0,\ul r})} \bigg) F_{\ul p}(z,n_0,t_{0,\ul r}) \no \\ \no
&\hspace*{-1.8cm} -\frac{z}{2\alpha^+(n_0,t_{0,\ul r})}F_{\ul p}(z,n_0,t_{0,\ul r})  \no \\
& \hspace*{-1.8cm}  \times
 \sum_{k=1}^{q(n_0,t_{0,\ul r})}\sum_{\ell=0}^{p_k(n_0,t_{0,\ul r})-1} 
\f{\big(d^\ell \big(\zeta^{-1} y(P)\big)/d\zeta^\ell\big)\big|_{P=(\zeta,\eta)
=\hat \mu_k(n_0,t_{0,\ul r})}}{\ell!(p_k(n_0,t_{0,\ul r})-\ell-1)!}    \lb{ALT6.18D} \\ \no
& \hspace*{-1.8cm} \times 
\Bigg(\f{d^{p_k(n_0,t_{0,\ul r})-\ell-1}}{d \zeta^{p_k(n_0,t_{0,\ul r})-\ell-1}}
\Bigg((z-\zeta)^{-1}  \no \\
& \hspace*{1.6cm} \times  \prod_{k'=1, \, k'\neq k}^{q(n_0,t_{0,\ul r})} 
(\zeta-\mu_{k'}(n_0,t_{0,\ul r}))^{-p_{k'}(n_0,t_{0,\ul r})} 
\Bigg)\Bigg)\Bigg|_{\zeta=\mu_k(n_0,t_{0,\ul r})},    \no 
\end{align}
in analogy to \eqref{ALS4.13}.

Our aim is to find an autonomous first-order system of ordinary differential equations
with respect to $t_{\ul r}$ for $f_{\ell,\pm}$, $g_{\ell,\pm}$, and $h_{\ell,\pm}$ 
rather than for $\mu_j$. We divide  the differential equation
\begin{equation} \label{AL--F_t}
F_{\ul p,t_{\ul r}} = - 2 i G_{\ul p} \ti F_{\ul r} + i (\ti G_{\ul r} + \ti K_{\ul r})F_{\ul p}
\end{equation}
by $c_{0,+}z^{-p_-}y$ and rewrite it using Theorem \ref{tALB.2} as   
\begin{align}  
& \sum_{\ell = 0}^\infty \hat f_{\ell,+,t_{\ul r}}\zeta^{\ell+1} =
- 2i \bigg( \sum_{s=1}^{r_-} \tilde f_{r_- -s,-} \zeta^s 
+ \sum_{s=0}^{r_+ -1} \tilde f_{r_+ -1-s,+}\zeta^{-s}\bigg)
\sum_{\ell = 0}^\infty \hat g_{\ell,+}\zeta^\ell  \no \\  \lb{ALT6.19} 
& \quad + i \bigg( 2 \sum_{s=0}^{r_-} \tilde g_{r_- -s,-} \zeta^s  + 
2 \sum_{s=1}^{r_+} \tilde g_{r_+ -s,+} \zeta^{-s} - \tilde g_{r_-,-} 
+ \tilde g_{r_+,+}\bigg)
\sum_{\ell = 0}^\infty \hat f_{\ell,+}\zeta^{\ell+1}, \\ \no
& \hspace*{8.1cm} P \to P_{\infty_-}, \; \zeta=1/z.
\end{align}
The coefficients of $\zeta^{-s}$, $s=0, \dots,r_+ -1$, cancel since
\begin{equation} \lb{ALT6.20}
\sum_{k=0}^\ell \tilde f_{\ell-k,+} \hat g_{k,+}= 
\sum_{k=0}^\ell \tilde g_{\ell-k,+} \hat f_{k,+}, \quad \ell \in \bbN_0.
\end{equation}
In \eqref{ALT6.20} we used \eqref{ALhat f},  
\begin{equation} 
\tilde f_{\ell,+} = \sum_{k=0}^\ell \tilde c_{\ell-k,+} \hat f_{k,+}, \quad
\tilde g_{\ell,+} = \sum_{k=0}^\ell \tilde c_{\ell-k,+} \hat g_{k,+}.
\end{equation}
Comparing coefficients in \eqref{ALT6.19} then yields  
\footnote{$m\vee n = \max\{m,n\}$.}
\begin{align} 
&\hat f_{\ell,+,t_{\ul r}}= i \hat f_{\ell,+}(\tilde g_{r_+,+}-\tilde g_{r_-,-}) 
+  2i \sum_{k=0}^{r_+ -1} 
\big( \tilde g_{k,+} \hat f_{r_++\ell-k,+} - \tilde f_{k,+} \hat g_{r_++\ell-k,+} \big) 
\lb{ALT6.21} \\  
 &\quad - 2i \sum_{k=(\ell+1-r_-)\vee0}^\ell 
 \hat g_{k,+} \tilde f_{r_- -1-\ell+k,-} + 2i \sum_{k=(\ell+2-r_-)\vee0}^\ell 
\hat f_{k,+} \tilde g_{r_- -\ell+k,-}, 
\quad \ell\in \bbN_0.  \no 
\end{align}
By \eqref{ALT6.20}, the last sum in \eqref{ALT6.21} can be rewritten as
\begin{align} \no
&\sum_{j=0}^{r_+-1} \big(\tilde g_{j,+}\hat f_{r_++\ell-j,+}-\tilde f_{j,+}\hat g_{r_++\ell-j,+}\big)  \\ \no
&\qquad = \Bigg(\sum_{j=0}^{r_++\ell} - \sum_{j=r_+}^{r_++\ell}\Bigg) 
\big(\tilde g_{j,+}\hat f_{r_++\ell-j,+}-\tilde f_{j,+}\hat g_{r_++\ell-j,+}\big)\\ \no
&\qquad = - \sum_{j=r_+}^{r_++\ell}
\big(\tilde g_{j,+}\hat f_{r_++\ell-j,+}-\tilde f_{j,+}\hat g_{r_++\ell-j,+}\big)\\ \lb{ALT6.14}
&\qquad = \sum_{j=0}^\ell \big(\hat g_{j,+}\tilde f_{r_++\ell-j,+}-\hat f_{j,+}\tilde g_{r_++\ell-j,+}\big).
\end{align}

One performs a similar computation for $\hat f_{\ell,-,t_{\ul r}}$ using 
Theorem \ref{tALB.2} at $P \to P_{0,+}$. In summary, since 
$f_{k,\pm}= \sum_{\ell=0}^k c_{k-\ell, \pm} \hat f_{\ell, \pm}$,
\eqref{ALT6.21} and \eqref{ALT6.14} yield the following autonomous first-order system (for fixed $n=n_0$) 
\begin{align} \lb{ALT6.22}
 f_{\ell,\pm,t_{\ul r}}&=\calF_{\ell,\pm}(f_{j,-},f_{j,+},g_{j,-},g_{j,+}),
\quad \ell=0, \dots,p_\pm -1,
\end{align}
with initial conditions
\begin{align} \lb{ALT6.33}
&f_{\ell,\pm}(n_0,t_{0,\ul r}), \quad \ell=0,\dots,p_\pm -1, \\  \no
&g_{\ell,\pm}(n_0,t_{0,\ul r}),\quad \ell=0,\dots,p_\pm, 
\end{align}
where $\calF_{\ell,\pm}$, $\ell=0,\dots,p_\pm -1$, are polynomials in $2p+3$ variables,  
\begin{align}  \no
\calF_{\ell,\pm} &= i f_{\ell,\pm}(\tilde g_{r_{\pm},\pm}-\tilde g_{r_{\mp},\mp})  \\ \no
&\quad + 2i \sum_{k=0}^\ell  
\big( f_{k,\pm} (\tilde g_{r_{\mp}-\ell+k,\mp} - \tilde g_{r_{\pm}+\ell-k,\pm}) + g_{k,\pm} (\tilde f_{r_{\pm}+\ell -k,\pm} -\tilde f_{r_{\mp}-1-\ell +k,\mp}) \big) \\ \lb{ALT6.34} 
&\quad +  2i \sum_{k=0}^\ell c_{\ell-k,\pm} \times
\begin{cases}      
0, &0\leq k <r_\mp-1, \\
\sum_{j=0}^{k-r_\mp}  \hat g_{j,\pm} \tilde f_{r_\mp -1-k+j,\mp}\\ 
- \sum_{j=0}^{k+1-r_\mp} 
\hat f_{j,\pm} \tilde g_{r_\mp -k+j,\mp}, & k \geq r_\mp-1.
\end{cases} 
\end{align}
Explicitly, one obtains (for simplicity, $r_\pm > 1$)
\begin{align}  \lb{ALT6.34_0}
\calF_{0,\pm} &= i f_{0,\pm}(\tilde g_{r_\mp,\mp}- \tilde g_{r_\pm,\pm}) + 
2i g_{0,\pm}(\tilde f_{r_\pm,\pm} -\tilde f_{r_\mp -1,\mp}), \\ \no
\calF_{1,\pm} &= 2i f_{0,\pm}(\tilde g_{r_\mp -1,\mp}- \tilde g_{r_\pm+1,\pm})
+ i f_{1,\pm}(\tilde g_{r_\mp,\mp}- \tilde g_{r_\pm,\pm})  \\ \no
& \quad + 2i g_{0,\pm}(\tilde f_{r_\pm+1,\pm} -\tilde f_{r_\mp-2,\mp})
+ 2i g_{1,\pm}(\tilde f_{r_\pm,\pm}- \tilde f_{r_\mp-1,\mp}), \, \text{ etc.} 
\end{align}

By \eqref{ALT6.18B}--\eqref{ALT6.18D}, the initial conditions 
\eqref{ALT6.33} are uniquely determined by the initial divisor 
$\calD_{\humu(n_0,t_{0,\ul r})}$ in \eqref{ALT6.17} and by the data in \eqref{ALT6.18A}.

Similarly, one transforms
\begin{align}  \label{G_t}
G_{\ul p,t_{\ul r}} &= i F_{\ul p} \ti H_{\ul r} - i H_{\ul p} \ti F_{\ul r}, \\  \label{H_t}
H_{\ul p,t_{\ul r}} &= 2 i G_{\ul p} \ti H_{\ul r} - i \big(\ti G_{\ul r} + \ti K_{\ul r}\big)H_{\ul p} 
\end{align}
into (for fixed $n=n_0$) 
\footnote{$m\vee n = \max\{m,n\}$.}
\begin{align} \no
 \hat g_{0,\pm,t_{\ul r}}&=0, \\ \no
 \hat g_{\ell,\pm,t_{\ul r}}&= 
 i \sum_{k=0}^{r_\pm -1} 
\big(\tilde h_{k,\pm} \hat f_{r_\pm -1+\ell-k,\pm} -\tilde f_{k,\pm} 
\hat h_{r_\pm -1+\ell-k,\pm} \big)\\  \no
&\quad +i \sum_{k=(\ell-r_\mp)\vee0}^{\ell-1} 
\big( \hat f_{k,\pm} \tilde h_{r_\mp-\ell+k,\mp} 
- \hat h_{k,\pm} \tilde f_{r_\mp -\ell+k,\mp} \big)  \\ \no
&= i \sum_{k=0}^{\ell-1} 
\big(\hat h_{k,\pm} \tilde f_{r_\pm -1+\ell-k,\pm} -\hat f_{k,\pm} 
\tilde h_{r_\pm -1+\ell-k,\pm} \big)\\  
&\quad + i \sum_{k=(\ell-r_\mp)\vee0}^{\ell-1} 
\big( \hat f_{k,\pm} \tilde h_{r_\mp-\ell+k,\mp} 
- \hat h_{k,\pm} \tilde f_{r_\mp -\ell+k,\mp} \big), \quad  \ell \in \bbN,   \lb{ALT6.35}  \\
 \hat h_{\ell,\pm,t_{\ul r}}&=i \hat h_{\ell,\pm} \big(\tilde g_{r_\mp,\mp} -\tilde g_{r_\pm,\pm}\big)  +
2i \sum_{k=0}^{r_\pm -1}
\big(\tilde h_{k,\pm} \hat g_{r_\pm +\ell-k,\pm} 
-\tilde g_{k,\pm} \hat h_{r_\pm +\ell-k,\pm} \big) \no  \\ \no
&\quad + 2 i \sum_{k=(\ell-r_\mp+1)\vee0}^\ell 
\hat g_{k,\pm}\tilde h_{r_\mp-1-\ell+k,\mp} 
-  2 i \sum_{k=(\ell-r_\mp)\vee0}^\ell  \hat h_{k,\pm}\tilde g_{r_\mp-\ell+k,\mp}
\\ \no
&=i \hat h_{\ell,\pm} \big(\tilde g_{r_\mp,\mp} -\tilde g_{r_\pm,\pm}\big)  +
2i \sum_{k=0}^{\ell}
\big(\hat h_{k,\pm} \tilde g_{r_\pm +\ell-k,\pm} 
-\hat g_{k,\pm} \tilde h_{r_\pm +\ell-k,\pm} \big) \no  \\ \no
&\quad + 2 i \sum_{k=(\ell-r_\mp+1)\vee0}^\ell 
\hat g_{k,\pm}\tilde h_{r_\mp-1-\ell+k,\mp} 
-  2 i \sum_{k=(\ell-r_\mp)\vee0}^\ell  \hat h_{k,\pm}\tilde g_{r_\mp-\ell+k,\mp},\\ \lb{ALT6.36}
&\hspace{9cm} \ell \in \bbN_0.
\end{align}
Summing over $\ell$ in \eqref{ALT6.35}, \eqref{ALT6.36} then yields the following 
first-order system 
\begin{align} \lb{ALT6.37}
 g_{\ell,\pm,t_{\ul r}}&=\calG_{\ell,\pm}(f_{k,-},f_{k,+},h_{k,-},h_{k,+}),
\quad \ell=0, \dots,p_\pm, \\ \lb{ALT6.37a}
 h_{\ell,\pm,t_{\ul r}}&=\calH_{\ell,\pm}(g_{k,-},g_{k,+},h_{k,-},h_{k,+}),
\quad \ell=0, \dots,p_\pm -1,
\end{align}
with initial conditions
\begin{align} \no 
&f_{\ell,\pm}(n_0,t_{0,\ul r}), \quad \ell=0,\dots,p_\pm -1, \\  \lb{ALT6.33a}
&g_{\ell,\pm}(n_0,t_{0,\ul r}),\quad \ell=0,\dots,p_\pm,   \\ \no
&h_{\ell,\pm}(n_0,t_{0,\ul r}),\quad \ell=0,\dots,p_\pm -1,
\end{align}
where $\calG_{\ell,\pm}$, $\calH_{\ell,\pm}$, are polynomials in $2p+2$, 
$2p+3$ variables  
\begin{align} \no
&\calG_{\ell,\pm} =  i \sum_{k=0}^{\ell-1} 
\big(f_{k,\pm} (\tilde h_{r_\mp -\ell+k,\mp} -\tilde h_{r_\pm -1+\ell-k,\pm}) 
+ h_{k,\pm} (\tilde f_{r_\pm -1+\ell-k,\pm}-\tilde f_{r_\mp -\ell+k,\mp}) \big)\\  
&\quad -i \sum_{k=0}^{\ell-1} c_{\ell-1-k,\pm} \times
\begin{cases}      
0, &0\leq k \leq r_\mp, \\ \lb{ALT6.38}
\sum_{j=0}^{k-r_\mp-1} (\hat f_{j,\pm} \tilde h_{r_\mp-k+j,\mp} 
- \hat h_{j,\pm} \tilde f_{r_\mp -k+j,\mp}), & k > r_\mp,
\end{cases} 
\\ \no
&\calH_{\ell,\pm} =  i h_{\ell,\pm} \big(\tilde g_{r_\mp,\mp} -\tilde g_{r_\pm,\pm}\big)\\ \no
&\quad + 2 i \sum_{k=0}^\ell \big(g_{k,\pm} (\tilde h_{r_\mp -1-\ell+k,\mp} 
-\tilde h_{r_\pm +\ell-k,\pm})
+ h_{k,\pm} (\tilde g_{r_\pm +\ell-k,\pm}-\tilde g_{r_\mp -\ell+k,\mp})\big)   \\  
&\quad + 2i \sum_{k=0}^{\ell} c_{\ell-k,\pm} \times
\begin{cases}      
0, &0\leq k < r_\mp, \\ \lb{ALT6.39}
- \sum_{j=0}^{k-r_\mp}\hat g_{j,\pm}\tilde h_{r_\mp-1-k+j,\mp}\\
+ \sum_{j=0}^{k-r_\mp-1} \hat h_{j,\pm}\tilde g_{r_\mp-k+j,\mp}, & k \geq r_\mp.
\end{cases} 
\end{align}
Explicitly (assuming $r_\pm >2$),
\begin{align} \no  
\calG_{0,\pm} &= 0, \\ \label{ALT6.40a}
\calG_{1,\pm} &= i f_{0,\pm}(\tilde h_{r_\mp -1,\mp}- \tilde h_{r_\pm,\pm}) + 
i h_{0,\pm} (\tilde f_{r_\pm,\pm} -\tilde f_{r_\mp -1,\mp}), \\ \no
\calG_{2,\pm} &= i f_{0,\pm}(\tilde h_{r_\mp -2,\mp}- \tilde h_{r_\pm +1,\pm}) + 
i f_{1,\pm} (\tilde h_{r_\mp -1,\mp} -\tilde h_{r_\pm,\pm}) \\ \no
& \quad + i h_{0,\pm}(\tilde f_{r_\pm +1,\pm}- \tilde f_{r_\mp -2,\mp}) + 
i h_{1,\pm} (\tilde f_{r_\pm,\pm} -\tilde f_{r_\mp -1,\mp}), \, \text{ etc.,}  
\\  \label{ALT6.40b}
\calH_{0,\pm} &= 2i g_{0,\pm}(\tilde h_{r_\mp -1,\mp}- \tilde h_{r_\pm,\pm}) + 
i h_{0,\pm} (\tilde g_{r_\pm,\pm} -\tilde g_{r_\mp,\mp}), \\ \no
\calH_{1,\pm} &= 2i g_{0,\pm} (\tilde h_{r_\mp -2,\mp}- \tilde h_{r_\pm +1,\pm})
+2i g_{1,\pm}(\tilde h_{r_\mp -1,\mp}- \tilde h_{r_\pm,\pm})\\ \no
& \quad + 2i h_{0,\pm}(\tilde g_{r_\pm +1,\pm}- \tilde g_{r_\mp -1,\mp})
+ i h_{1,\pm}(\tilde g_{r_\pm,\pm}- \tilde g_{r_\mp,\mp}), \, \text{ etc.}
\end{align}

Again by \eqref{ALT6.18B}--\eqref{ALT6.18D}, the initial conditions \eqref{ALT6.33a} are uniquely determined by the initial divisor $\calD_{\humu(n_0,t_{0,\ul r})}$ in 
\eqref{ALT6.17} and by the data in \eqref{ALT6.18A}.

Being autonomous with polynomial right-hand sides, there
exists a $T_0>0$, such that the first-order initial value problem
\eqref{ALT6.22}, \eqref{ALT6.37}, \eqref{ALT6.37a} with initial conditions 
\eqref{ALT6.33a} has a unique solution 
\begin{align}
&f_{\ell,\pm}=f_{\ell,\pm}(n_0,t_{\ul r}), \quad \ell=0,\dots,p_\pm -1, \no \\
&g_{\ell,\pm}=g_{\ell,\pm}(n_0,t_{\ul r}), \quad \ell=0,\dots,p_\pm, \lb{ALT6.41} \\
&h_{\ell,\pm}=h_{\ell,\pm}(n_0,t_{\ul r}), \quad \ell=0,\dots,p_\pm -1,  \no \\
& \text{for all } \, t_{\ul r}\in (t_{0,\ul r}-T_0,t_{0,\ul r}+T_0)   \no
\end{align} 
(cf., e.g., \cite[Sect.\ III.10]{Walter:1998}). Given the solution \eqref{ALT6.41},
we proceed as in Section \ref{ALSs4} and introduce the following quantities 
(where $t_{\ul r}\in (t_{0,\ul r}-T_0,t_{0,\ul r}+T_0)$):
\begin{align} 
\alpha^+(n_0,t_{\ul r}) &= \alpha(n_0,t_{\ul r}) \bigg(\prod_{m=0}^{2p+1}E_m\bigg)^{1/2} \,  
\prod_{k=1}^{q(n_0,t_{\ul r})}\mu_k(n_0,t_{\ul r})^{-p_k(n_0,t_{\ul r})}, \\  
F_{\ul p}(z,n_0,t_{\ul r})&=\sum_{\ell=1}^{p_-} f_{p_- -\ell,-}(n_0,t_{\ul r})z^{-\ell} 
+ \sum_{\ell=0}^{p_+ -1} f_{p_+ -1-\ell,+}(n_0,t_{\ul r})z^\ell   \no \\
&= - c_{0,+}\alpha^+(n_0,t_{\ul r})z^{-p_-}
\prod_{k=1}^{q(n_0,t_{\ul r})} (z-\mu_k(n_0,t_{\ul r}))^{p_k(n_0,t_{\ul r})},  \lb{ALT6.42} \\
G_{\ul p}(z,n_0,t_{\ul r}) &= \frac{1}{2} \bigg(\frac{1}{\alpha(n_0,t_{\ul r})}-\frac{z}{\alpha^+(n_0,t_{\ul r})} \bigg) F_{\ul p}(z,n_0,t_{\ul r})\\ \no
&\hspace*{-1.8cm} -\frac{z}{2\alpha^+(n_0,t_{\ul r})}F_{\ul p}(z,n_0,t_{\ul r})
\sum_{k=1}^{q(n_0,t_{\ul r})}\sum_{\ell=0}^{p_k(n_0,t_{\ul r})-1} 
\f{\big(d^\ell
\big(\zeta^{-1} y(P)\big)/d\zeta^\ell\big)\big|_{P=(\zeta,\eta)
=\hat \mu_k(n_0,t_{\ul r})}}{\ell!(p_k(n_0,t_{\ul r})-\ell-1)!}\\ \no
& \hspace*{-1.8cm} \times 
\Bigg(\f{d^{p_k(n_0,t_{\ul r})-\ell-1}}{d \zeta^{p_k(n_0,t_{\ul r})-\ell-1}}\Bigg(
(z-\zeta)^{-1}\prod_{k'=1, \, k'\neq k}^{q(n_0,t_{\ul r})} 
(\zeta-\mu_{k'}(n_0,t_{\ul r}))^{-p_{k'}(n_0,t_{\ul r})} 
\Bigg)\Bigg)\Bigg|_{\zeta=\mu_k(n_0,t_{\ul r})}.
\end{align}
In particular, this leads to the divisor 
\begin{equation}
\calD_{\humu(n_0,t_{\ul r})}\in\symq    
\end{equation} 
and the sign of $y$ in \eqref{ALT6.42} is chosen as usual by
\begin{align}
\begin{split}
\hat\mu_k(n_0,t_{\ul r})&=(\mu_k(n_0,t_{\ul r}),
(2/c_{0,+})\mu_j(n_0,t_{\ul r})^{p_-} G_{\ul p}(\mu_k(n_0,t_{\ul r}),n_0,t_{\ul r})),\\
&\hspace*{5.68cm}   k=1,\dots,q(n_0,t_{\ul r}),   \lb{ALT6.43}
\end{split}
\end{align}
and
\begin{equation}
\humu(n_0,t_{\ul r}) = \{\underbrace{\mu_1(n_0,t_{\ul r}),\dots,
\mu_1(n_0,t_{\ul r})}_{p_1(n_0,t_{\ul r}) \text{ times}},
\dots,\underbrace{\mu_{q(n_0,t_{\ul r})}(n_0,t_{\ul r}),\dots,
\mu_{q(n_0,t_{\ul r})}(n_0,t_{\ul r})}_{p_{q(n_0,t_{\ul r})}(n_0,t_{\ul r}) 
\text{ times}}\}  \lb{ALT6.40}
\end{equation}
with 
\begin{equation}
\mu_k(n_0,t_{\ul r})\neq \mu_{k'}(n_0,t_{\ul r}) \, \text{ for } \, k\neq
k', \;  k,k'=1,\dots,q(n_0,t_{\ul r}),    \lb{ALT6.44}
\end{equation}
and
\begin{equation}
p_k(n_0,t_{\ul r})\in\bbN, \; k=1,\dots,q(n_0,t_{\ul r}), \quad 
\sum_{k=1}^{q(n_0,t_{\ul r})} p_k(n_0,t_{\ul r}) = p.   
\end{equation}
By construction (cf.\ \eqref{ALT6.43}), the divisor
$\calD_{\humu(n_0,t_{\ul r})}$ is nonspecial for all $t_{\ul r}\in (t_{0,\ul r}-T_0,t_{0,\ul r}+T_0)$.

In exactly the same manner as in \eqref{ALS4.16}--\eqref{ALS4.18} one then
infers that $F_{\ul p}(\dott,n_0,t_{\ul r})$ divides $R_{\ul p}-G_{\ul p}^2$ (since
$t_{\ul r}$ is just a fixed parameter).

As in Section \ref{ALSs4}, the assumption that the Laurent polynomial $F_{\ul p}(\dott,n_0-1,t_{\ul r})$ is
of full order is implied by the hypothesis that
\begin{equation}
\calD_{\humu(n_0,t_{\ul r})}\in\calM_0 \, \text{ for all } \, t_{\ul r}\in
(t_{0,\ul r}-T_0,t_{0,\ul r}+T_0).   \lb{ALT6.45}
\end{equation}

The explicit formula for $\beta(n_0,t_{\ul r})$ then reads (for $t_{\ul r}\in
(t_{0,\ul r}-T_0,t_{0,\ul r}+T_0)$)
\begin{align}
& \alpha^+(n_0,t_{\ul r})\beta(n_0,t_{\ul r})  \\ 
& \quad = -\f{1}{2} \sum_{k=1}^{q(n_0,t_{\ul r})}
\f{\big(d^{p_k(n_0,t_{\ul r})-1} \big(\zeta^{-1}y(P)\big)/d \zeta^{p_k(n_0,t_{\ul r})-1}\big)\big|_{P=(\zeta,\eta)
=\hat\mu_k(n_0,t_{\ul r})}}{(p_k(n_0,t_{\ul r})-1)!}  \no \\
& \qquad \times \prod_{k'=1, \, k'\neq k}^{q(n_0,t_{\ul r})}
(\mu_k(n_0,t_{\ul r})-\mu_{k'}(n_0,t_{\ul r}))^{-p_k(n_0,t_{\ul r})}  \no \\
& \qquad
+\f{1}{2}\bigg(\f{\alpha^+(n_0,t_{\ul r})}{\alpha(n_0,t_{\ul r})} 
+ \sum_{k=1}^{q(n_0,t_{\ul r})}p_k(n_0,t_{\ul r})\mu_k(n_0,t_{\ul r}) 
- \f{1}{2}\sum_{m=0}^{2p+1}E_m \bigg). \lb{ALT6.45A}
\end{align}

With \eqref{ALT6.35}--\eqref{ALT6.45A} in place, we can now apply the stationary
formalism as summarized in Theorem \ref{ALSt4.3}, subject to the
additional hypothesis \eqref{ALT6.45}, for each fixed $t_{\ul r}\in
(t_{0,\ul r}-T_0,t_{0,\ul r}+T_0)$. This yields, in particular, the quantities
\begin{equation} 
F_{\ul p}, \, G_{\ul p}, \, H_{\ul p}, \, \alpha, \, \beta, \, \text{ and } \, \humu, \, \hunu \,\text{ for } \, 
(n,t_{\ul r}) \in \bbZ\times(t_{0,\ul r}-T_0,t_{0,\ul r}+T_0),  \lb{ALT6.45a}
\end{equation}
which are of the form \eqref{ALT6.42}--\eqref{ALT6.45A}, replacing the fixed
$n_0\in\bbZ$ by an arbitrary $n\in\bbZ$. In addition, one has the following result.

\begin{lemma}  \lb{ALTl6.2}
Assume Hypothesis \ref{ALTh6.1} and condition \eqref{ALT6.45}. Then the
following relations are valid on $\bbC\times \bbZ \times (t_{0,\ul r}-T_0,t_{0,\ul r}+T_0)$, 
\begin{align}
G_{\ul p}^2 - F_{\ul p} H_{\ul p} &= R_{\ul p},  \lb{ALT6.46} \\ 
z (G_{\ul p}^- - G_{\ul p}) + z \beta F_{\ul p} + \alpha H_{\ul p}^- &=0,  \lb{ALT6.47} \\
z \beta F_{\ul p}^- + \alpha H_{\ul p} - G_{\ul p} + G_{\ul p}^- &=0, \lb{ALT6.48} \\
- F_{\ul p} + z F_{\ul p}^- + \alpha (G_{\ul p} + G_{\ul p}^-)&=0,  \lb{ALT6.49} \\
z \beta (G_{\ul p} + G_{\ul p}^-) - z H_{\ul p} + H_{\ul p}^-&=0,  \lb{ALT6.49a} 
\end{align}
and hence the stationary part, \eqref{ALzctstat}, of the algebro-geometric
initial value problem holds, 
\begin{equation}
UV_{\ul p} - V^+_{\ul p}U = 0 \, \text{ on } \, \bbC\times \bbZ \times
(t_{0,\ul r}-T_0,t_{0,\ul r}+T_0).   \lb{ALT6.50}   
\end{equation}
In particular, Lemmas \ref{lAL3.1}--\ref{lAL3.4} apply.
\end{lemma}

Lemma \ref{ALTl6.2} now raises the following important consistency issue:
On the one hand, one can solve the initial value problem \eqref{ALT6.61},
\eqref{ALT6.62} at $n=n_0$ in some interval $t_{\ul r} \in
(t_{0,\ul r}-T_0,t_{0,\ul r}+T_0)$, and then extend the quantities
$F_{\ul p}, G_{\ul p}, H_{\ul p}$ to all $\bbC\times\bbZ\times(t_{0,\ul r}-T_0,t_{0,\ul r}+T_0)$
using the stationary algorithm summarized in Theorem \ref{ALSt4.3} as just 
recorded in Lemma \ref{ALTl6.2}. On the other hand, one can solve the
initial value problem \eqref{ALT6.61}, \eqref{ALT6.62} at $n=n_1$, $n_1\neq
n_0$, in some interval $t_{\ul r} \in (t_{0,\ul r}-T_1,t_{0,\ul r}+T_1)$ with the
initial condition obtained by applying the discrete algorithm to the
quantities $F_{\ul p}, G_{\ul p}, H_{\ul p}$ starting at $(n_0,t_{0,\ul r})$ and ending at
$(n_1,t_{0,\ul r})$. Consistency then requires that the two approaches yield
the same result at $n=n_1$ for $t_{\ul r}$ in some open neighborhood of
$t_{0,\ul r}$. 

Equivalently, and pictorially speaking, envisage a vertical $t_{\ul r}$-axis and
a horizontal $n$-axis. Then, consistency demands that first solving the
initial value problem \eqref{ALT6.61}, \eqref{ALT6.62} at $n=n_0$ in some
$t_{\ul r}$-interval around $t_{0,\ul r}$ and using the stationary algorithm to
extend $F_{\ul p}, G_{\ul p}, H_{\ul p}$ horizontally to $n=n_1$ and the same $t_{\ul r}$-interval
around $t_{0,\ul r}$, or first applying the stationary algorithm starting at
$(n_0,t_{0,\ul r})$ to extend $F_{\ul p}, G_{\ul p}, H_{\ul p}$ horizontally to $(n_1,t_{0,\ul r})$
and then solving  the initial value problem \eqref{ALT6.61}, \eqref{ALT6.62} at
$n=n_1$ in some $t_{\ul r}$-interval around $t_{0,\ul r}$ should produce the same
result at $n=n_1$ in a sufficiently small open $t_{\ul r}$ interval around
$t_{0,\ul r}$.

To settle this consistency issue, we will prove the following result. To this end 
we find it convenient to replace the initial value problem \eqref{ALT6.61}, 
\eqref{ALT6.62} by the original
$t_{\ul r}$-dependent zero-curvature equation \eqref{ALzctilde},
$U_{t_{\ul r}}+U\ti V_{\ul r}- \ti V^+_{\ul r}U=0$ on
$\bbC\times\bbZ\times(t_{0,\ul r}-T_0,t_{0,\ul r}+T_0)$.  

\begin{lemma}  \lb{ALTl6.3}
Assume Hypothesis \ref{ALTh6.1} and condition \eqref{ALT6.45}. Moreover,
suppose that \eqref{ALF_t}--\eqref{ALH_t} hold on $\bbC\times\{n_0\}\times
(t_{0,\ul r}-T_0,t_{0,\ul r}+T_0)$. Then \eqref{ALF_t}--\eqref{ALH_t} hold on $\bbC\times\bbZ\times (t_{0,\ul r}-T_0,t_{0,\ul r}+T_0)$, that is,
\begin{align}
F_{\ul p,t_{\ul r}}(z,n,t_{\ul r}) &= - 2 i G_{\ul p}(z,n,t_{\ul r}) \ti F_{\ul r}(z,n,t_{\ul r})  \no \\ \lb{ALT6.51}
& \quad  + i \big(\ti G_{\ul r}(z,n,t_{\ul r}) + \ti K_{\ul r}(z,n,t_{\ul r})\big)F_{\ul p}(z,n,t_{\ul r}), \\  \lb{ALT6.52}
G_{\ul p,t_{\ul r}}(z,n,t_{\ul r}) &= i F_{\ul p}(z,n,t_{\ul r}) \ti H_{\ul r}(z,n,t_{\ul r}) - i H_{\ul p}(z,n,t_{\ul r}) \ti F_{\ul r}(z,n,t_{\ul r}), \\  
H_{\ul p,t_{\ul r}}(z,n,t_{\ul r}) &= 2 i G_{\ul p}(z,n,t_{\ul r}) \ti H_{\ul r}(z,n,t_{\ul r}) 
\no \\ \lb{ALT6.53}
& \quad  - i \big(\ti G_{\ul r}(z,n,t_{\ul r}) + \ti K_{\ul r}(z,n,t_{\ul r})\big)H_{\ul p}(z,n,t_{\ul r}),\\
& \hspace*{.75cm} 
(z,n,t_{\ul r})\in \bbC\times\bbZ\times (t_{0,\ul r}-T_0,t_{0,\ul r}+T_0). \no
\end{align}
Moreover,
\begin{align}
 \phi_{t_{\ul r}}(P,n,t_{\ul r}) &=i \ti F_{\ul r}(z,n,t_{\ul r}) \phi^2(P,n,t_{\ul r}) \no \\
& \quad \lb{ALT6.54}
- i \big(\ti G_{\ul r}(z,n,t_{\ul r}) + \ti K_{\ul r}(z,n,t_{\ul r})\big)\phi(P,n,t_{\ul r}) +  i \ti H_{\ul r}(z,n,t_{\ul r}),\\ 
\alpha_{t_{\ul r}}(n,t_{\ul r}) &= i z \ti F_{\ul r}^-(z,n,t_{\ul r}) \no
\\  \lb{ALT6.55}
& \quad
+ i \alpha(n,t_{\ul r}) \big(\ti G_{\ul r}(z,n,t_{\ul r}) + \ti {K}_{\ul r}^-(z,n,t_{\ul r})\big) - i \ti F_{\ul r}(z,n,t_{\ul r}),\\  
\beta_{t_{\ul r}}(n,t_{\ul r}) &= - i \beta(n,t_{\ul r}) \big(\ti G_{\ul r}^-(z,n,t_{\ul r}) + \ti K_{\ul r}(z,n,t_{\ul r})\big) \no
\\ \lb{ALT6.56}
& \quad
+ i\ti H_{\ul r}(z,n,t_{\ul r}) - iz^{-1} \ti H_{\ul r}^-(z,n,t_{\ul r}), \\
& \hspace*{-.5cm} 
(z,n,t_{\ul r})\in \bbC\times\bbZ\times (t_{0,\ul r}-T_0,t_{0,\ul r}+T_0). \no
\end{align}
\end{lemma} 
\begin{proof}
By Lemma \ref{ALTl6.2} we have \eqref{AL4.22}, \eqref{AL4.23}, \eqref{ALriccati_time}, 
\eqref{ALphi 1t}--\eqref{ALphi 3t}, and \eqref{ALT6.46}--\eqref{ALT6.49a} for 
$(n,t_{\ul r})\in \bbZ\times (t_{0,\ul r}-T_0,t_{0,\ul r}+T_0)$ at our disposal.

Differentiating \eqref{AL4.22} at $n=n_0$ with respect to $t_{\ul r}$ and inserting
\eqref{ALT6.51} and \eqref{ALT6.52} at $n=n_0$ then yields \eqref{ALT6.54} at $n=n_0$.

We note that the sequences $\tilde f_{\ell,\pm}$, $\tilde g_{\ell,\pm}$, $\tilde h_{\ell,\pm}$ satisfy the 
recursion relations \eqref{ALg_l+}--\eqref{ALh_l-} (since the homogeneous sequences satisfy these relations). 
Hence, to prove \eqref{ALT6.55} and \eqref{ALT6.56} at $n=n_0$ it remains to show 
\begin{align} 
\begin{split} \label{ALT6.58}
\alpha_{t_{\ul r}} &=  i \alpha(\tilde g_{r_+,+} + \tilde g_{r_-,-}^-) + 
i(\tilde f_{r_- -1,-}^- -\tilde f_{r_+ -1,+}), \\
\beta_{t_{\ul r}}&=  -i \beta(\tilde g_{r_+,+}^- + \tilde g_{r_-,-}) - 
i(\tilde h_{r_+ -1,+}^-  -\tilde h_{r_- -1,-}).
\end{split}
\end{align}
But this follows from \eqref{ALT6.51},  \eqref{ALT6.53} at $n=n_0$  
(cf.\ \eqref{ALT6.34_0},  \eqref{ALT6.40b})  
\begin{align} \no
\alpha_{t_{\ul r}} &=
i \alpha(\tilde g_{r_+,+}- \tilde g_{r_-,-}) + 
i(\tilde f_{r_-,-} -\tilde f_{r_+ -1,+}),\\ \no
\beta_{t_{\ul r}} &=
i \beta(\tilde g_{r_+,+}- \tilde g_{r_-,-}) + 
i(\tilde h_{r_- -1,-} -\tilde h_{r_+,+}).
\end{align}
Inserting now \eqref{ALf_l-} at $\ell =r_- -1$ and \eqref{ALh_l+} at $\ell =r_+ -1$ 
then yields \eqref{ALT6.58}.

For the step $n=n_0\mp1$ we differentiate \eqref{F^-}--\eqref{G^+} (which are equivalent to \eqref{ALT6.46}--\eqref{ALT6.49a})
and insert \eqref{ALT6.51}--\eqref{ALT6.53}, \eqref{ALalpha_t}--\eqref{AL22} at 
$n=n_0$. 
For the case $n>n_0$ we obtain $\alpha^+_{t_{\ul r}}$ and $\beta^+_{t_{\ul r}}$  from  
\eqref{ALT6.51},  \eqref{ALT6.53} at $n=n_0$ 
as before using the other two signs in \eqref{ALT6.34_0}, \eqref{ALT6.40b}. Iterating these arguments proves \eqref{ALT6.51}--\eqref{ALT6.56} for 
$(z,n,t_{\ul r})\in \bbC\times\bbZ\times (t_{0,\ul r}-T_0,t_{0,\ul r}+T_0)$.
\end{proof}

We summarize Lemmas \ref{ALTl6.2} and \ref{ALTl6.3} next.

\begin{theorem}  \lb{ALTt6.4}
Assume Hypothesis \ref{ALTh6.1} and condition \eqref{ALT6.45}. Moreover,
suppose that 
\begin{align}
&f_{\ell,\pm}=f_{\ell,\pm}(n_0,t_{\ul r}), \quad \ell=0,\dots,p_\pm -1, \no \\
&g_{\ell,\pm}=g_{\ell,\pm}(n_0,t_{\ul r}), \quad \ell=0,\dots,p_\pm, \lb{ALT6.60} \\
&h_{\ell,\pm}=h_{\ell,\pm}(n_0,t_{\ul r}), \quad \ell=0,\dots,p_\pm -1  \no \\
& \text{for all } \, t_{\ul r}\in (t_{0,\ul r}-T_0,t_{0,\ul r}+T_0),   \no
\end{align}
satisfy the autonomous first-order system of ordinary differential equations 
$($for fixed $n=n_0$$)$ 
\begin{align}
& f_{\ell,\pm,t_{\ul r}}=\calF_{\ell, \pm}(f_{k,-}, f_{k,+}, g_{k,-}, g_{k,+}), \quad 
\ell=0,\dots,p_\pm -1, \no \\
& g_{\ell,\pm,t_{\ul r}}=\calG_{\ell, \pm}(f_{k,-}, f_{k,+}, h_{k,-}, h_{k,+}), \quad 
\ell=0,\dots,p_\pm, \lb{ALT6.61} \\ \no
& h_{\ell,\pm,t_{\ul r}}=\calH_{\ell, \pm}(g_{k,-}, g_{k,+}, h_{k,-}, h_{k,+}), \quad 
\ell=0,\dots,p_\pm -1,  
\end{align}
with $\calF_{\ell, \pm}$, $\calG_{\ell, \pm}$, $\calH_{\ell, \pm}$ given by 
\eqref{ALT6.34}, \eqref{ALT6.38},  \eqref{ALT6.39}, and with initial conditions
\begin{align} \no 
&f_{\ell,\pm}(n_0,t_{0,\ul r}), \quad \ell=0,\dots,p_\pm -1, \\  \lb{ALT6.62}
&g_{\ell,\pm}(n_0,t_{0,\ul r}),\quad \ell=0,\dots,p_\pm,   \\ \no
&h_{\ell,\pm}(n_0,t_{0,\ul r}),\quad \ell=0,\dots,p_\pm -1.
\end{align}
Then $F_{\ul p}$, $G_{\ul p}$, and $H_{\ul p}$ as constructed in \eqref{ALT6.42}--\eqref{ALT6.45a} 
on $\bbC\times\bbZ\times(t_{0,\ul r}-T_0,t_{0,\ul r}+T_0)$ satisfy the
zero-curvature equations \eqref{ALzctilde}, \eqref{ALzctstat}, and \eqref{ALV_t}  
on $\bbC\times \bbZ \times (t_{0,\ul r}-T_0,t_{0,\ul r}+T_0)$, 
\begin{align}
U_{t_{\ul r}}+U\ti{V}_{\ul r}-\ti V_{\ul r}^+U&=0, \label{ALT6.63} \\
UV_{\ul p}-V^+_{\ul p}U&=0, \label{ALT6.64} \\
V_{\ul p, t_{\ul r}} - \big[\ti V_{\ul r}, V_{\ul p}\big]&= 0   \label{ALT6.64a} 
\end{align}
with $U$, $V_{\ul p}$, and $\ti V_{\ul r}$ given by \eqref{ALv_osv}. In
particular, $\alpha, \beta$ satisfy \eqref{AL4.1A} and the algebro-geometric initial value problem \eqref{ALal_ivp}, \eqref{AL4.3A} on $\bbZ \times (t_{0,\ul r}-T_0,t_{0,\ul r}+T_0)$, 
\begin{align}
\begin{split}
& \ti \AL_{\ul r} (\alpha, \beta) =
\begin{pmatrix}-i\alpha_{t_{\ul r}} - \alpha(\tilde g_{r_+,+} 
+ \tilde g_{r_-,-}^-) + \tilde f_{r_+ -1,+} - \tilde f_{r_- -1,-}^-\\
-i\beta_{t_{\ul r}}+ \beta(\tilde g_{r_+,+}^- + \tilde g_{r_-,-}) 
- \tilde h_{r_- -1,-} + \tilde h_{r_+ -1,+}^-     \end{pmatrix} = 0, \label{ALT6.65}  \\
& (\alpha, \beta)\big|_{t=t_{0,\ul r}} = \big(\alpha^{(0)}, \beta^{(0)}\big) ,  
\end{split} \\
& \sAL_{\ul p} \big(\alpha^{(0)}, \beta^{(0)}\big) =  \begin{pmatrix} 
-\alpha^{(0)}(g_{p_+,+} + g_{p_-,-}^-) + f_{p_- -1,+} - f_{p_- -1,-}^-\\
  \beta^{(0)}(g_{p_+,+}^- + g_{p_-,-}) - h_{p_- -1,-} + h_{p_+ -1,+}^-       
\end{pmatrix}=0.   \label{ALT6.66}  
\end{align}
In addition,  $\alpha, \beta$ are given by 
\begin{align}
& \alpha^+(n,t_{\ul r})= \alpha(n,t_{\ul r}) \bigg(\prod_{m=0}^{2p+1}E_m\bigg)^{1/2} \, 
\prod_{k=1}^{q(n,t_{\ul r})}\mu_k(n,t_{\ul r})^{-p_k(n,t_{\ul r})}, \lb{ALT6.75} \\
& \alpha^+(n,t_{\ul r}) \beta(n,t_{\ul r}) = -\f{1}{2} \sum_{k=1}^{q(n,t_{\ul r})}
\f{\big(d^{p_k(n,t_{\ul r})-1} \big(\zeta^{-1}y(P)\big)/d \zeta^{p_k(n,t_{\ul r})-1}\big)\big|_{P=(\zeta,\eta)
=\hat\mu_k(n,t_{\ul r})}}{(p_k(n,t_{\ul r})-1)!}  \no \\
& \hspace*{4.5cm} \times \prod_{k'=1, \, k'\neq k}^{q(n,t_{\ul r})}
(\mu_k(n,t_{\ul r})-\mu_{k'}(n,t_{\ul r}))^{-p_k(n,t_{\ul r})}    \lb{ALT6.76}  \\
& \quad
+\f{1}{2}\bigg(\bigg(\prod_{m=0}^{2p+1}E_m\bigg)^{1/2} \, 
\prod_{k=1}^{q(n,t_{\ul r})}\mu_k(n,t_{\ul r})^{-p_k(n,t_{\ul r})}  
+ \sum_{k=1}^{q(n,t_{\ul r})}p_k(n,t_{\ul r})\mu_k(n,t_{\ul r})  \no \\
& \hspace*{1.4cm}  - \f{1}{2}\sum_{m=0}^{2p+1}E_m \bigg), \quad   
(z,n,t_{\ul r})\in \bbZ \times (t_{0,\ul r}-T_0,t_{0,\ul r}+T_0).    \no 
\end{align}
Moreover, Lemmas \ref{lAL3.1}--\ref{lAL3.4} and \ref{lAL4.2}--\ref{lAL4.3} apply. 
\end{theorem}

As in Lemma \ref{ALSl4.2} we now show that also in the time-dependent case, 
most initial divisors are well-behaved in the sense that the corresponding divisor
trajectory stays away from $\Pinfpm, \Pzpm$ for all $(n,t_{\ul r})\in\bbZ\times\bbR$.

\begin{lemma}  \lb{ALTl6.5}
The set $\calM_1$ of initial divisors $\calD_{\humu(n_0,t_{0,\ul r})}$
for which $\calD_{\humu(n,t_{\ul r})}$ and $\calD_{\hunu(n,t_{\ul r})}$, defined via 
\eqref{AL4.82} and \eqref{AL4.83}, are admissible 
$($i.e., do not contain $\Pinfpm$, $\Pzpm$$)$ and hence are nonspecial
for all $(n,t_{\ul r})\in\bbZ\times\bbR$, forms a dense set of full
measure in the set $\symq$ of nonnegative divisors of degree $p$.  
\end{lemma}
\begin{proof}
Let $\calM_{\rm sing}$ be as introduced in the proof of Lemma \ref{ALSl4.2}.
Then
\begin{align}
& \bigcup_{t_{\ul r}\in\bbR} \Big(\underline{\alpha}_{Q_0}(\calM_{\rm sing}) 
+ t_{\ul r} \ti{\ul U}^{(2)}_{\ul r}\Big) \no \\
& \quad \subseteq \bigcup_{P\in\{\Pinfpm, \Pzpm\}}
\bigcup_{t_{\ul r}\in\bbR} \Big(\underline{A}_{Q_0}(P) + 
\underline{\alpha}_{Q_0}(\sym^{p-1} (\calK_p)) 
+ t_{\ul r} \ti{\ul U}^{(2)}_{\ul r} \Big)       \lb{ALT6.77}
\end{align}
is of  measure zero as well, since it is contained in the image of 
$\bbR\times \sym^{p-1} (\calK_p)$ which misses one real dimension in
comparison to the $2p$ real dimensions of $J(\calK_p)$.  But then
\begin{align}
& \bigcup_{(n,t_{\ul r})\in\bbZ\times\bbR}
\Big(\underline{\alpha}_{Q_0}(\calM_{\rm sing})  +
n\underline{A}_{P_{0,-}}(P_{\infty_+}) 
+ t_{\ul r} \ti{\ul U}^{(2)}_{\ul r} \Big)   \no \\
& \quad \; \cup \bigg(\bigcup_{(n,t_{\ul r})\in\bbZ\times\bbR}
\Big(\underline{\alpha}_{Q_0}(\calM_{\rm sing})  +
n\underline{A}_{P_{0,-}}(P_{\infty_+}) 
+ t_{\ul r} \ti{\ul U}^{(2)}_{\ul r} \Big) + \ul A_{\Pzm}(\Pinfm)\bigg)   \lb{ALT6.78}                 
\end{align}
is also of measure zero. Applying $\underline{\alpha}_{Q_0}^{-1}$ to the
complement of the set in \eqref{ALT6.78} then yields a set $\calM_1$ of
full measure in  $\symq$. In particular, $\calM_1$ is
necessarily dense in $\symq$.
\end{proof}

\begin{theorem}  \lb{ALTt6.6}
Let $\calD_{\humu(n_0,t_{0,\ul r})}\in\calM_1$ be an initial divisor as in 
Lemma \ref{ALTl6.5}. Then the sequences $\alpha, \beta$ constructed from
$\humu(n_0,t_{0,\ul r})$ as described in Theorem \ref{ALTt6.4} satisfy
Hypothesis \ref{hAL4.1}. In particular, the solution $\alpha, \beta$ of the
algebro-geometric initial value problem \eqref{ALT6.75}, \eqref{ALT6.76} is
global in $(n,t_{\ul r})\in\bbZ\times\bbR$.
\end{theorem}
\begin{proof}
Starting with $\calD_{\humu(n_0,t_{0,\ul r})}\in\calM_1$, the procedure
outlined in this section and summarized in Theorem \ref{ALTt6.4} leads to 
$\calD_{\humu(n,t_{\ul r})}$ and $\calD_{\hunu(n,t_{\ul r})}$ for all 
$(n,t_{\ul r})\in \bbZ \times (t_{0,\ul r}-T_0,t_{0,\ul r}+T_0)$ such that \eqref{AL4.82} and 
\eqref{AL4.83} hold. But if
$\alpha, \beta$ should blow up, then $\calD_{\humu(n,t_{\ul r})}$ or 
$\calD_{\hunu(n,t_{\ul r})}$ must hit one of $\Pinfpm$ or $\Pzpm$,   
which is excluded by our choice of initial condition.
\end{proof}

We note, however, that in general (i.e., unless one is, e.g., in the special
periodic case), $\calD_{\humu(n,t_{\ul r})}$ will get
arbitrarily close to $P_{\infty_{\pm}}$, $\Pzpm$ since straight motions on the torus
are generically dense (see e.g. \cite[Sect.\ 51]{Arnold:1989} or 
\cite[Sects.\ 1.4, 1.5]{KatokHasselblatt:1995}) and hence no uniform bound 
(and no uniform bound away from zero) on the sequences 
$\alpha(n,t_{\ul r}), \beta(n,t_{\ul r})$ exists as $(n,t_{\ul r})$ varies in $\bbZ\times\bbR$. In particular, these complex-valued algebro-geometric solutions of the Ablowitz--Ladik hierarchy initial value problem, in general, will not be quasi-periodic with respect to $n$ or $t_{\ul r}$ 
(cf.\ the usual definition of quasi-periodic functions, e.g., in 
\cite[p.\ 31]{PasturFigotin:1992}).

\appendix
\section{Hyperelliptic Curves in a Nutshell} \lb{AL.sA}
\renewcommand{\theequation}{A.\arabic{equation}}
\renewcommand{\thetheorem}{A.\arabic{theorem}}
\setcounter{theorem}{0} 
\setcounter{equation}{0}

We provide a very brief summary of some of the fundamental properties
and notations needed from the theory of hyperelliptic curves.  More details can be found in some of the standard textbooks \cite{FarkasKra:1992}, 
\cite{Fay:1973}, and \cite{Mumford:1984}, as well as monographs
dedicated to integrable systems such as
\cite[Ch.\ 2]{BelokolosBobenkoEnolskiiItsMatveev:1994}, 
\cite[App.\ A, B]{GesztesyHolden:2005}, \cite[App.\ A]{Teschl:2000}.

Fix $\N\in\bbN$. The hyperelliptic curve $\calK_\N$
of genus $\N$ used in Sections \ref{s3}--\ref{ALTs6} is defined by
\begin{align}
&\calK_\N\colon \calF_\N(z,y)=y^2-R_{2\N+2}(z)=0, \quad
R_{2\N+2}(z)=\prod_{m=0}^{2\N+1}(z-E_m), \lb{b0} \\
& \{E_m\}_{m=0,\dots,2\N+1}\subset\bbC, \quad
E_m \neq E_{m'} \text{ for } m \neq m', \, m,m'=0,\dots,2\N+1.
\label{b1}
\end{align}
The curve \eqref{b0} is compactified by adding the
points $\Pinfp$ and $\Pinfm$,
$\Pinfp \neq \Pinfm$, at infinity.
One then introduces an appropriate set of
$\N+1$ nonintersecting cuts $\calC_j$ joining
$E_{m(j)}$ and $E_{m^\prime(j)}$ and denotes
\begin{equation}
\calC=\bigcup_{j\in \{1,\dots,\N+1 \}}\calC_j,
\quad
\calC_j\cap\calC_k=\emptyset,
\quad j\neq k.\label{b2}
\end{equation}
Defining the cut plane
\begin{equation}
\Pi=\bbC\setminus\calC, \label{b3}
\end{equation}
and introducing the holomorphic function
\begin{equation}
R_{2\N+2}(\dott)^{1/2}\colon \Pi\to\bbC, \quad
z\mapsto \left(\prod_{m=0}^{2\N+1}(z-E_m) \right)^{1/2}\label{b4}
\end{equation}
on $\Pi$ with an appropriate choice of the square root
branch in \eqref{b4}, one considers 
\begin{equation}
\calM_{\N}=\{(z,\sigma R_{2\N+2}(z)^{1/2}) \,|\, 
z\in\bbC,\; \sigma\in\{\pm 1\}
\}\cup\{\Pinfp,\Pinfm\} \label{b5}
\end{equation}
by extending $R_{2\N+2}(\dott)^{1/2}$ to $\calC$. The
hyperelliptic curve $\calK_\N$ is then the set
$\calM_{\N}$ with its natural complex structure obtained
upon gluing the two sheets of $\calM_{\N}$
crosswise along the cuts. The set of branch points
$\calB(\calK_\N)$ of $\calK_\N$ is given by
\begin{equation}
\calB(\calK_\N)=\{(E_m,0)\}_{m=0,\dots,2\N+1} \lb{5a}
\end{equation}
and finite points $P$ on $\calK_\N$ are denoted by
$P=(z,y)$, where $y(P)$ denotes the meromorphic function
on $\calK_\N$ satisfying $\calF_\N(z,y)=y^2-R_{2\N+2}(z)=0$.
Local coordinates near $P_0=(z_0,y_0)\in\calK_\N\setminus
(\calB(\calK_\N)\cup\{\Pinfp,\Pinfm\})$ are
given by $\zeta_{P_0}=z-z_0$, near $\Pinfpm$ by
$\zeta_{\Pinfpm}=1/z$, and near branch points
$(E_{m_0},0)\in\calB(\calK_\N)$ by
$\zeta_{(E_{m_0},0)}=(z-E_{m_0})^{1/2}$. The Riemann surface
$\calK_\N$ defined in this manner has topological genus $\N$.
Moreover, we introduce the holomorphic sheet exchange map
(involution)
\begin{equation}
* \colon \calK_\N \to \calK_\N,\quad
P=(z,y)\mapsto P^*=(z,-y),\;
P_{\infty_\pm} \mapsto P_{\infty_\pm}^*=P_{\infty_\mp}.    \lb{a12}
\end{equation}

One verifies that $dz/y$ is a holomorphic differential
on $\calK_\N$ with zeros of order $\N-1$ at $\Pinfpm$
and hence
\begin{equation}
\eta_j=\frac{z^{j-1}dz}{y}, \quad j=1,\dots,\N,   \lb{b24}
\end{equation}
form a basis for the space of holomorphic differentials
on $\calK_\N$.  Introducing the
invertible matrix $C$ in $\bbC^\N$,
\begin{align}
\begin{split}
& C =(C_{j,k})_{j,k=1,\dots,\N}, \quad C_{j,k}
= \int_{a_k} \eta_j, \\
& \underline{c} (k) = (c_1(k), \dots,
c_p(k)), \quad c_j (k) =
C_{j,k}^{-1}, \;\, j,k=1,\dots,\N, \lb{A.7}
\end{split}
\end{align}
the corresponding basis of normalized holomorphic
differentials $\omega_j$, $j=1,\dots,\N$, on $\calK_\N$ is given by
\begin{equation}
\omega_j = \sum_{\ell=1}^\N c_j (\ell) \eta_\ell,
\quad \int_{a_k} \omega_j =
\delta_{j,k}, \quad j,k=1,\dots,\N. \lb{b26}
\end{equation}
Here $\{a_j,b_j\}_{j=1,\dots,\N}$ is a homology basis for
$\calK_\N$ with intersection matrix of the cycles satisfying
\begin{equation}
a_j \circ b_k=\delta_{j,k}, \; a_j \circ a_k=0,
\; b_j \circ b_k=0, \quad j,k=1,\dots,\N. \lb{c26}
\end{equation}

Associated with the homology basis
$\{a_j, b_j\}_{j=1,\dots,\N}$ we
also recall the canonical dissection of $\calK_\N$
along its cycles yielding
the simply connected interior $\hatt \calK_\N$ of the
fundamental polygon $\partial {\hatt \calK}_\N$ given by
\begin{equation}
\partial  {\hatt \calK}_\N =a_1 b_1 a_1^{-1} b_1^{-1}
a_2 b_2 a_2^{-1} b_2^{-1} \cdots
a_\N^{-1} b_\N^{-1}.
\lb{a25}
\end{equation}
Let $\calM (\calK_\N)$ and $\calM^1 (\calK_\N)$ denote the
set of meromorphic
functions (0-forms) and meromorphic
differentials (1-forms)
on $\calK_\N$. Holomorphic
differentials are also called Abelian differentials
of the first kind. Abelian differentials of the
second kind, $\omega^{(2)} \in \calM^1 (\calK_\N)$, are characterized
by the property that all their residues vanish.  They will usually be
normalized by demanding that all their $a$-periods vanish, that is,
$\int_{a_j} \omega^{(2)} =0$, $j=1,\dots,\N$. 
Any meromorphic differential $\omega^{(3)}$ on
$\calK_\N$ not of the first or
second kind is said to be of the third
kind. A differential of the third kind $\omega^{(3)} \in \calM^1 (\calK_\N)$
is usually normalized by the vanishing of its
$a$-periods, that is, $\int_{a_j} \omega^{(3)} =0$, $j=1,\dots, \N$. 
A normal differential of the third kind $\omega_{P_1, P_2}^{(3)}$ associated
with two points $P_1$, $P_2 \in \hatt \calK_\N$, $P_1 \neq P_2$, by definition, 
has simple poles at $P_j$ with residues $(-1)^{j+1}$, $j=1,2$ and
vanishing $a$-periods.  

Next, define the matrix $\tau=(\tau_{j,\ell})_{j,\ell=1,\dots,\N}$ by
\begin{equation}
\tau_{j,\ell}=\int_{b_\ell}\omega_j, \quad j,\ell=1,
\dots,\N. \label{b8}
\end{equation}
Then
\begin{equation}
\Im(\tau)>0 \quad \text{and} \quad \tau_{j,\ell}=\tau_{\ell,j},
\quad j,\ell =1,\dots,\N.  \lb{a18a}
\end{equation}
Associated with $\tau$ one introduces the period lattice
\begin{equation}
L_\N = \{ \ul z \in\bbC^\N \,|\, \ul z = \ul m + \ul n\tau,
\; \ul m, \ul n \in\bbZ^\N\}.
\lb{a28}
\end{equation}

Next, fix a base point $Q_0\in\calK_\N\setminus
\{\Pzpm,\Pinfpm\}$, denote by
$J(\calK_\N) = \bbC^\N/L_\N$ the Jacobi variety of $\calK_\N$,
and define the
Abel map $\underline{A}_{Q_0}$ by
\begin{equation}
\underline{A}_{Q_0} \colon \calK_\N \to J(\calK_\N), \quad
\underline{A}_{Q_0}(P)=
\bigg(\int_{Q_0}^P \omega_1,\dots,\int_{Q_0}^P \omega_\N \bigg)
\pmod{L_\N}, \quad P\in\calK_\N. \label{aa46}
\end{equation}
Similarly, we introduce
\begin{equation}
\ul \alpha_{Q_0}  \colon
\Div(\calK_\N) \to J(\calK_\N),\quad
\calD \mapsto \ul \alpha_{Q_0} (\calD)
=\sum_{P \in \calK_\N} \calD (P) \ul A_{Q_0} (P),
\label{aa47}
\end{equation}
where $\Div(\calK_\N)$ denotes the set of
divisors on $\calK_\N$. Here $\calD \colon \calK_\N \to \bbZ$
is called a divisor on $\calK_\N$ if $\calD(P)\neq0$ for only
finitely many $P\in\calK_\N$. (In the main body of this paper
we will choose $Q_0$ to be one of the branch points, i.e.,
$Q_0\in\calB(\calK_\N)$, and for simplicity we will always choose
the same path of integration from $Q_0$ to $P$ in all Abelian
integrals.) 

In connection with divisors on $\calK_\N$ we shall employ the
following
(additive) notation,
\begin{align} \lb{A.17}
&\calD_{Q_0\ul Q}=\calD_{Q_0}+\calD_{\ul Q}, \quad \calD_{\ul
Q}=\calD_{Q_1}+\cdots +\calD_{Q_m}, \\
& {\ul Q}=\{Q_1, \dots ,Q_m\} \in \sym^m \calK_\N,
\quad Q_0\in\calK_\N, \; m\in\bbN, \no
\end{align}
where for any $Q\in\calK_\N$,
\begin{equation} \lb{A.18}
\calD_Q \colon  \calK_\N \to\bbN_0, \quad
P \mapsto  \calD_Q (P)=
\begin{cases} 1 & \text{for $P=Q$},\\
0 & \text{for $P\in \calK_\N\setminus \{Q\}$}, \end{cases}
\end{equation}
and $\sym^n \calK_\N$ denotes the $n$th symmetric product of
$\calK_\N$. In particular, $\sym^m \calK_\N$ can be
identified with
the set of nonnegative
divisors $0 \leq \calD \in \Div(\calK_\N)$ of degree $m$.

For $f\in \calM (\calK_\N) \setminus \{0\}$,
$\omega \in \calM^1 (\calK_\N) \setminus \{0\}$ the
divisors of $f$ and $\omega$ are denoted
by $(f)$ and
$(\omega)$, respectively.  Two
divisors $\calD$, $\calE\in \Div(\calK_\N)$ are
called equivalent, denoted by
$\calD \sim \calE$, if and only if $\calD -\calE
=(f)$ for some
$f\in\calM (\calK_\N) \setminus \{0\}$.  The divisor class
$[\calD]$ of $\calD$ is
then given by $[\calD]
=\{\calE \in \Div(\calK_\N) \,|\, \calE \sim \calD\}$.  We
recall that
\begin{equation}
\deg ((f))=0,\, \deg ((\omega)) =2(\N-1),\,
f\in\calM (\calK_\N) \setminus
\{0\},\,  \omega\in \calM^1 (\calK_\N) \setminus \{0\},
\lb{a38}
\end{equation}
where the degree $\deg (\calD)$ of $\calD$ is given
by $\deg (\calD)
=\sum_{P\in \calK_\N} \calD (P)$.  It is customary to call
$(f)$ (respectively,
$(\omega)$) a principal (respectively, canonical)
divisor.

Introducing the complex linear spaces
\begin{align}
\calL (\calD) & =\{f\in \calM (\calK_\N) \,|\, f=0
        \text{ or } (f) \geq \calD\}, \quad
r(\calD) =\dim \calL (\calD),
\lb{a39}\\
\calL^1 (\calD) & =
        \{ \omega\in \calM^1 (\calK_\N) \,|\, \omega=0
        \text{ or } (\omega) \geq
\calD\}, \quad i(\calD) =\dim \calL^1 (\calD),   \lb{a40}
\end{align}
with $i(\calD)$ the index of speciality of $\calD$, one infers
that $\deg(\calD)$, $r(\calD)$, and $i(\calD)$ only depend on
the divisor class $[\calD]$ of $\calD$.  Moreover, we recall the
following fundamental fact. 

\begin{theorem} \lb{thm3}
Let $\calD_{\ul Q} \in \sym^\N \calK_\N$,
$\ul Q=\{Q_1, \ldots, Q_\N\}$.  Then,
\begin{equation}
1 \leq i (\calD_{\ul Q} ) =s   \lb{a46}
\end{equation}
if and only if $\{Q_1,\ldots, Q_\N\}$ contains $s$ pairings of the type 
$\{P, P^*\}$. $($This includes, of course, branch
points for which $P=P^*$.$)$ One has $s\leq \N/2$.
\end{theorem}

\section{Some Interpolation Formulas} \lb{AL.sB}
\renewcommand{\theequation}{B.\arabic{equation}}
\renewcommand{\thetheorem}{B.\arabic{theorem}}
\setcounter{theorem}{0}
\setcounter{equation}{0}

In this appendix we recall a useful interpolation formula which goes
beyond the standard Lagrange interpolation formula for polynomials in the
sense that the zeros of the interpolating polynomial need not be distinct.

\begin{lemma}[\cite{GesztesyHoldenTeschl:2007}] \lb{ALlB.1} 
Let $p\in\bbN$ and $S_{p-1}$ be a polynomial of degree $p-1$. In
addition, let $F_p$ be a monic polynomial of degree $p$ of the form
\begin{equation}
F_p (z)=\prod_{k=1}^q (z-\mu_k)^{p_k}, \quad p_j \in\bbN, \; 
\mu_j\in\bbC, \; j=1,\dots,q, \quad \sum_{k=1}^q p_k =p.   \lb{B.1}
\end{equation}
Then,
\begin{align}
S_{p-1}(z)&=F_p (z)\sum_{k=1}^q \sum_{\ell=0}^{p_k -1} 
\f{S_{p-1}^{(\ell)}(\mu_k)}{\ell! (p_k-\ell-1)!}  \lb{B.2} \\
& \quad \times \Bigg(\f{d^{p_k-\ell-1}}{d \zeta^{p_k-\ell-1}}
\Bigg((z-\zeta)^{-1}\prod_{k'=1, \, k'\neq k}^q 
(\zeta-\mu_{k'})^{-p_{k'}}\Bigg)
\Bigg)\Bigg|_{\zeta=\mu_k}, \quad z\in\bbC.   \no
\end{align}
In particular, $S_{p-1}$ is uniquely determined by prescribing the $p$
values
\begin{equation}
S_{p-1}(\mu_k), S_{p-1}'(\mu_k),\dots,S_{p-1}^{(p_{k}-1)}(\mu_k), \quad 
k=1,\dots,q,   \lb{B.3}
\end{equation}
at the given points $\mu_1.\dots,\mu_q$. \\
Conversely, prescribing the $p$ complex numbers
\begin{equation}
\alpha_k^{(0)}, \alpha_k^{(1)},\dots,\alpha_k^{(p_k-1)}, \quad
k=1,\dots,q, 
\end{equation}
there exists a unique polynomial $T_{p-1}$ of degree $p-1$,
\begin{align}
T_{p-1}(z)&=F_p (z)\sum_{k=1}^q \sum_{\ell=0}^{p_k -1} 
\f{\alpha_k^{(\ell)}}{\ell! (p_k-\ell-1)!}  \lb{B.3a} \\
& \quad \times \Bigg(\f{d^{p_k-\ell-1}}{d \zeta^{p_k-\ell-1}}
\Bigg((z-\zeta)^{-1}\prod_{k'=1, \, k'\neq k}^q 
(\zeta-\mu_{k'})^{-p_{k'}}\Bigg)
\Bigg)\Bigg|_{\zeta=\mu_k}, \quad z\in\bbC,    \no 
\end{align}
such that
\begin{equation}
T_{p-1}(\mu_k)=\alpha_k^{(0)}, \, T_{p-1}'(\mu_k)=\alpha_k^{(1)},
\dots,\, T_{p-1}^{(p_{k}-1)}(\mu_k)=\alpha_k^{(p_k-1)}, \quad 
k=1,\dots,q.   \lb{B.3b}
\end{equation}
\end{lemma}

We briefly mention two special cases of \eqref{B.2}. First, assume the
generic case where all zeros of $F_p$ are distinct, that is,
\begin{equation}
q=p, \quad p_k=1, \quad \mu_k \neq \mu_{k'} \, 
\text{ for } \, k\neq k', \; k,k'=1,\dots,p. \lb{B.10}
\end{equation}
In this case \eqref{B.2} reduces to the classical Lagrange interpolation
formula
\begin{equation}
S_{p-1}(z)=F_p (z)\sum_{k=1}^p 
\f{S_{p-1}(\mu_k)}{((d
F_p (\zeta)/d\zeta)|_{\zeta=\mu_k})(z-\mu_k)}, \quad z\in\bbC.  \lb{B.11}
\end{equation}
Second, we consider the other extreme case where all zeros of $F_p$
coincide, that is,
\begin{equation}
q=1, \quad p_1=p, \quad F_p (z)=(z-\mu_1)^p, \quad z\in\bbC.  \lb{B.12}
\end{equation}
In this case \eqref{B.2} reduces of course to the Taylor expansion of
$S_{p-1}$ around $z=\mu_1$, 
\begin{equation}
S_{p-1}(z)=\sum_{\ell=0}^{p-1} \f{S_{p-1}^{(\ell)}(\mu_1)}{\ell!}
(z-\mu_1)^\ell, \quad z\in\bbC.   \lb{B.13}
\end{equation}

\section{Asymptotic Spectral Parameter Expansions} \lb{ALApp.high}
\renewcommand{\theequation}{C.\arabic{equation}}
\renewcommand{\thetheorem}{C.\arabic{theorem}}
\setcounter{theorem}{0}
\setcounter{equation}{0}

In this appendix we consider asymptotic spectral parameter expansions of
$F_{\ul p}/y$, $G_{\ul p}/y$, and $H_{\ul p}/y$, the resulting recursion relations for
the homogeneous coefficients  $\hat f_\ell$, $\hat g_\ell$, and $\hat
h_\ell$, their connection with the nonhomogeneous coefficients $f_\ell$,
$g_\ell$, and $h_\ell$, and the connection between $c_{\ell,\pm}$ and
$c_{\ell}(\ul E^{\pm 1})$ (cf.\ \eqref{ALB2.26h}). For detailed proofs of the material in this section we refer to \cite{GesztesyHolden:2005},  
\cite{GesztesyHoldenMichorTeschl:2007}. We will employ the notation
\begin{equation}
 \ul E^{\pm 1}=\big(E_0^{\pm 1},\dots,E_{2p+1}^{\pm 1}\big).   \lb{ALEpm}
\end{equation}

We start with the following elementary result (a consequence of the  binomial
expansion) assuming $\eta\in\bbC$ such that
$|\eta|<\min\{|E_0|^{-1},\dots, |E_{2p+1}|^{-1}\}$:
\begin{equation}
\left(\prod_{m=0}^{2p+1} \big(1-{E_m}{\eta}\big)
\right)^{1/2}=\sum_{k=0}^{\infty}c_k(\ul
E)\eta^{k}, \lb{ALB2.26g}
\end{equation}
where
\begin{align}
c_0(\ul E)&=1,\no \\
c_k(\ul E)&=\!\!\!\!\!\!\!\!\sum_{\substack{j_0,\dots,j_{2p+1}=0\\
       j_0+\cdots+j_{2p+1}=k}}^{k}\!\!
\f{(2j_0)!\cdots(2j_{2p+1})!\, E_0^{j_0}\cdots E_{2p+1}^{j_{2p+1}}}
{2^{2k} (j_0!)^2\cdots (j_{2p+1}!)^2 (2j_0-1)\cdots(2j_{2p+1}-1)},
\quad k\in\bbN.  \label{ALB2.26h}
\end{align}
The first few coefficients explicitly are given by
\begin{align}
c_0(\ul E)=1, \;
c_1(\ul E)=-\f12\sum_{m=0}^{2p+1} E_m, \;
c_2(\ul E)=\f14\sum_{\substack{m_1,m_2=0\\ m_1< m_2}}^{2p+1}
E_{m_1} E_{m_2}-\f18 \sum_{m=0}^{2p+1} E_m^2,
\quad \text{etc.} \lb{ALB2.26i}
\end{align}

Next we turn to asymptotic expansions. We recall the convention
$y(P)= \mp \zeta^{-p-1}+\Oh(\zeta^{-p})$ near $\Pinfpm$
(where $\zeta=1/z$) and $y(P) = \pm (c_{0,-}/c_{0,+})+\Oh(\zeta)$ near $\Pzpm$ 
(where $\zeta=z$). 

\begin{theorem} [\cite{GesztesyHoldenMichorTeschl:2007}] \lb{tALB.2}
Assume \eqref{ALneq 0,1}, $\sAL_{\ul p}(\alpha,\beta)=0$, and suppose
$P=(z,y)\in\calK_p\setminus\{\Pinfp,\Pinfm\}$. Then $z^{p_-} F_{\ul p}/y$,
$z^{p_-} G_{\ul p}/y$, and $z^{p_-} H_{\ul p}/y$ have the following convergent expansions 
as $P\to \Pinfpm$, respectively, $P\to\Pzpm$,  
\begin{align} 
\frac{z^{p_-}}{c_{0,+}} \frac{F_{\ul p}(z)}{y} &= \begin{cases} 
\mp \sum_{\ell=0}^\infty \hat f_{\ell,+} \zeta^{\ell+1},  &
P\to \Pinfpm, \qquad \zeta=1/z, \\
\pm \sum_{\ell=0}^\infty \hat f_{\ell,-} \zeta^\ell,  &
P\to \Pzpm, \qquad \zeta=z,
\end{cases}  \label{ALF/y 0} \\
\frac{z^{p_-}}{c_{0,+}} \frac{G_{\ul p}(z)}{y} &= \begin{cases}
\mp \sum_{\ell=0}^\infty \hat g_{\ell,+} \zeta^\ell,  &
P\to \Pinfpm, \qquad \zeta=1/z, \\
\pm \sum_{\ell=0}^\infty \hat g_{\ell,-} \zeta^\ell,  &
P\to \Pzpm, \qquad \zeta=z, 
\end{cases}    \label{ALG/y 0} \\
\frac{z^{p_-}}{c_{0,+}} \frac{H_{\ul p}(z)}{y} &= \begin{cases}
\mp \sum_{\ell=0}^\infty \hat h_{\ell,+} \zeta^\ell,  &
P\to \Pinfpm, \qquad \zeta=1/z, \\
\pm \sum_{\ell=0}^\infty \hat h_{\ell,-} \zeta^{\ell+1},  &
P\to \Pzpm, \qquad \zeta=z,
\end{cases}    \label{ALH/y 0} 
\end{align}
where $\zeta=1/z$ $($resp., $\zeta=z$$)$ is the local coordinate near $\Pinfpm$ 
$($resp., $\Pzpm$$)$ and $\hat f_{\ell,\pm}$, $\hat g_{\ell,\pm}$, and $\hat h_{\ell,\pm}$ are the homogeneous versions of the coefficients $f_{\ell,\pm}$, $g_{\ell,\pm}$, 
and $h_{\ell,\pm}$ introduced in \eqref{AL2.04a}--\eqref{AL2.04c}. 

Moreover, the $E_m$-dependent summation constants
$c_{\ell,\pm}$, $\ell=0,\dots, p_{\pm}$, in $F_{\ul p}$, $G_{\ul p}$, and 
$H_{\ul p}$ are given by 
\begin{equation}
c_{\ell,\pm}= c_{0,\pm} c_\ell(\ul E^{\pm 1}), \quad \ell=0,\dots,p_{\pm}.   \lb{ALBc}
\end{equation}
\end{theorem}

\medskip

\noindent {\bf Acknowledgments.} F.G., J.M., and G.T.
gratefully acknowledge the extraordinary hospitality of the Department of
Mathematical Sciences of the Norwegian University of Science and
Technology, Trondheim, during extended stays in the summer of 2004--2006,
where parts of this paper were written.


\end{document}